\documentclass{article}

\usepackage{arxiv}

\usepackage[utf8]{inputenc} 
\usepackage[T1]{fontenc}    
\usepackage{hyperref}       
\usepackage{url}            
\usepackage{booktabs}       
\usepackage{amsfonts}       
\usepackage{nicefrac}       
\usepackage{microtype}      
\usepackage{graphicx}
\usepackage{natbib}
\usepackage{doi}

\usepackage{amssymb}
\usepackage{amsmath}
\usepackage{amsthm}

\newtheorem{theorem}{Theorem}

\newtheorem{corollary}{Corollary}

\theoremstyle{definition}

\title{Approximation of bivariate densities with compositional splines}

\author{Stanislav Škorňa\thanks{Corresponding author: stanislav.skorna01@upol.cz}, Jitka Machalová, Jana Burkotová, Karel Hron, Sonja Greven
}

\date{}

\hypersetup{
pdftitle={A template for the arxiv style},
pdfsubject={q-bio.NC, q-bio.QM},
pdfauthor={David S.~Hippocampus, Elias D.~Striatum},
pdfkeywords={First keyword, Second keyword, More},
}

\begin{document}
\maketitle

\begin{abstract}
Reliable estimation and approximation of probability density functions is fundamental for their further processing. However, their specific properties, i.e. scale invariance and relative scale, prevent the use of standard methods of spline approximation and have to be considered when building a suitable spline basis. Bayes Hilbert space methodology allows to account for these properties of densities and enables their conversion to a standard Lebesgue space of square integrable functions using the centered \mbox{log-ratio} transformation. As the transformed densities fulfill a zero integral constraint, the constraint should likewise be respected by any spline basis used. Bayes Hilbert space methodology also allows to decompose bivariate densities into their interactive and independent parts with univariate marginals. As this yields a useful framework for studying the dependence structure between random variables, a spline basis ideally should admit a corresponding decomposition. This paper proposes a new spline basis for (transformed) bivariate densities respecting the desired zero integral property. We show that there is a one-to-one correspondence of this  basis to a corresponding basis in the Bayes Hilbert space of bivariate densities using tools of this methodology. Furthermore, the spline  representation and the resulting decomposition into interactive and independent parts are derived. Finally, this novel spline representation is evaluated in a simulation study and applied to empirical geochemical data.
\end{abstract}


\section{Introduction} \label{intro}
Probability density functions (PDFs), as a specific instance of distributional data resulting from massive data collection, are of increasing importance in applications, e.g., as age distributions \citep{hron16}, income distributions \cite{talska20,maier2021additive}, particle size distributions \citep{menafoglio16,talska21}, share index \citep{iacopini19} or anthropometric distributions \citep{hron22}. Typically, the interest lies in the relative structure of the collected values on a given compact domain. Accordingly, the input data need to be aggregated into tractable form before further statistical processing. This can be achieved through histograms \citep{compositional} or using a nonparametric method like kernel density estimation \citep{steyer23}. The latter already reflects by construction the fact that continuous distributional data are in essence characterized by PDFs. However, to make full use of the functional nature of  PDFs and use methods of functional data analysis (FDA) \citep{ramsay05,kokoszka17}  for further statistical processing of samples of PDFs, usually some proper basis expansion is preferred as a representation of the PDFs. There are several options, with one of the most popular being $B$-splines \cite{Dierckx}, which offer, among other properties, high flexibility for efficient smoothing, local support resulting in numerical efficiency, and the possibility to perform FDA via multivariate analysis on spline coefficients. The resulting spline expansions can be further used as a building block in more complex methods, e.g. regression \cite{maier2021additive} or classification \citep{pavlu23} with densities.

While the framework for their statistical analysis has so far mostly been developed for univariate PDFs \cite{boogaart14,hron16,talska20}, many real world distributional data are in essence bivariate, or more generally multivariate. For example, the joint height and weight distributions (densities) can be of interest in anthropometric studies \citep{hron22} or multielemental PDFs resulting from the monitoring of contamination in soil \citep{genest22}. As in the univariate case, there is a demand for an~appropriate preprocessing of such bivariate density data using spline approximations. 

However, PDFs as positive functions on a given domain are characterized by their proportionality a.k.a. scale invariance, which results in their usual unit integral representation, therefore, suitable care needs to be taken in their meaningful representation, including their spline approximation. This is achieved through the concept of Bayes spaces \cite{egozcue06,boogaart14}, which provide a geometric representation of PDFs with a separable Hilbert space structure. This enables a one-to-one transformation to 
the $L^2_0$ space of $L^2$ functions with zero integral, where a $B$-spline basis can be constructed. For densities with square integrable logarithm (the space of which we refer to as $\mathcal{B}^2$ space), this transformation is called the centred log-ratio (clr) transformation.

Note that the use of a standard $B$-spline basis would not respect the zero integral constraint. This constraint can be imposed by additional requirements on the spline coefficients as proposed in \cite{machalova16,hron22}. Constrained coefficients are, however, inconvenient for further processing e.g.\ using FDA methods. To overcome this limitation and to respect the nature of the $L^2_0$ space, an alternative spline basis, known as $Z\!B$-splines, was introduced for the case of univariate densities in \cite{compositional}. These basis functions have the desired property of zero integral and thus avoid inconvenient coefficient constraints. 

Moreover, a bivariate PDF can be orthogonally decomposed into its interactive and independent part with appealing probabilistic implications, as introduced in \citep{hron22}, e.g. one can derive a measure of dependence called \textit{simplicial deviance}. However, a standard $B$-spline basis does not fully reflect the potential of this decomposition due to the constrained coefficients. In this paper, our aim is to fill this gap and provide a~comprehensive approach. Our newly introduced spline basis for bivariate densities naturally fulfills both desiderata:
\begin{enumerate}
    \item
    The zero integral constraint, which is beneficial from both theoretical and practical perspectives, without the need for additional constraints on spline coefficients.
    \smallskip
    \item
    A decomposition of the basis into its interactive and independent parts, which directly corresponds to the analogous useful decomposition of the PDFs.
\end{enumerate}

\noindent Finally, its construction opens a path to further desirable generalization of the resulting approximation and decomposition to multivariate densities in the Bayes space framework.

The paper is organized as follows. In the next section, the main idea of the Bayes space framework with emphasis on bivariate densities is outlined. Section \ref{twodim} is devoted to a brief summary of previous developments of the $B$-spline representation and mainly to our novel introduction of $Z\!B$-splines and the $Z\!B$-spline representation for bivariate densities in the $L^2_0$ space, supported with further theoretical background in Section~\ref{ZBsmoothing}. The expression of $Z\!B$-splines in the original Bayes spaces is studied in Section~\ref{backtoBayes}. Sensitivity of the new spline representation to the choice of related tuning parameters in the approximation of histogram data is analyzed with a simulation study in Section~\ref{simul}. The potential of the new approximation approach is additionally demonstrated in Section~\ref{appl} with empirical geochemical data. The final Section~\ref{concl} concludes with a discussion and outlook.

\section{Bayes Hilbert spaces} \label{Bayes}
Bayes spaces, which have the structure of separable Hilbert spaces, were originally developed for $\sigma$-finite measures defined on a compact interval $I\subset\mathbb{R}$, see \cite{boogaart14}. The idea can be extended to the bivariate case \citep{hron22}, where we consider a domain $\Omega\subset\mathbb{R}^2$ as a Cartesian product of two univariate domains $\Omega_X = [a,b]$ and $\Omega_Y = [c,d]$. Using \mbox{Radon-Nikodym} derivatives, the Bayes Hilbert space can be equivalently expressed with densities, see \cite{boogaart14, genest22} for details. We assume a Bayes space ${\cal B}^{2}(\Omega)$ as a space of positive bivariate densities with respect to the Lebesgue measure, defined on a product domain $\Omega$ as defined above, and being square-log integrable.

For densities $f,\ g \in {\cal B}^{2}(\Omega)$ and $\alpha\in\mathbb{R}$, the operations \textit{perturbation} and \textit{powering} are defined as
$$
(f\oplus g) =_{{\cal B}^{2}(\Omega)} f \cdot g,\qquad (\alpha \odot f) =_{{\cal B}^{2}(\Omega)} f^{\alpha},
$$
respectively, where $=_{{\cal B}^{2}(\Omega)}$ indicates equality up to rescaling with a constant (which does not alter the relative information carried by density functions in ${\cal B}^{2}(\Omega)$). To complete the Hilbert structure of ${\cal B}^{2}(\Omega)$, the inner product is defined as
\begin{equation*}
\left\langle f, g\right\rangle_{\mathcal{B}^2(\Omega)} = \frac{1}{2\,(d-c)\,(b-a)} \iint_{\Omega} \iint_{\Omega} \ln \frac{f(x,y)}{f(s,t)}\, \ln \frac{g(x,y)}{g(s,t)} \, \mbox{d}x\,\mbox{d}y\,\mbox{d}s\,\mbox{d}t,
\end{equation*}
which induces the norm and distance $d$ in the standard way as
\begin{equation*}
\|f\|_{\mathcal{B}^2(\Omega)}=\sqrt{\langle f, f\rangle_{\mathcal{B}^2(\Omega)}},\quad d_{\mathcal{B}^2(\Omega)}(f,g)=\|f\ominus g\|_{\mathcal{B}^2(\Omega)},
\end{equation*}
where $f\ominus g=f\oplus[(-1)\odot g]$. See \cite{boogaart14, egozcue06, hron22, genest22} for more details.\\

For bivariate and multivariate densities in general, one of their main useges is to study the dependence structure between random variables. For this purpose, it is useful that Bayes space methodology enables to orthogonally decompose a density function into its \textit{interactive} and \textit{independent} parts, the latter being the Bayes space sum (perturbation or product) of univariate \textit{geometric marginals}. In the independence case, independent part and full density are equal. The initial idea was introduced in \cite{hron22} and further adjusted and extended in \cite{genest22}.

Geometric marginals of a density $f\in{\cal B}^{2}(\Omega)$ are $f_{1}\in{\cal B}^{2}_{\{1\}}(\Omega)$ and $f_{2}\in{\cal B}^{2}_{\{2\}}(\Omega)$ defined for any $(x,y)\in\Omega=[a,b] \times [c,d]$ as
\begin{align*}
f_{1}(x,y) &\equiv f_{1}(x) =_{{\cal B}^{2}(\Omega)} 
\mbox{exp}\left\{\frac{1}{d-c}\int_c^d\ln f(x,t)\,\mbox{d}t\right\}, \\
f_{2}(x,y) &\equiv f_{2}(y) =_{{\cal B}^{2}(\Omega)} 
\mbox{exp}\left\{\frac{1}{b-a}\int_a^b\ln f(t,y)\,\mbox{d}t \right\},
\end{align*}
where ${\cal B}^{2}_{\{1\}}(\Omega)$ and ${\cal B}^{2}_{\{2\}}(\Omega)$ stand for subspaces of ${\cal B}^{2}(\Omega)$ with PDFs depending only on one variable and being constant in the respective other variable, which corresponds to embeddings in the bivariate space. Geometric marginals represent orthogonal projections of the density $f$ onto ${\cal B}^{2}_{\{1\}}(\Omega)$ and ${\cal B}^{2}_{\{2\}}(\Omega)$, as stated in \cite{genest22}. In particular, the following theorem holds.
\begin{theorem}[\cite{genest22}] \label{Bmargin}
For the independent and interactive parts of $f\in{\cal B}^{2}(\Omega)$ defined as
$$
f_{ind} = f_{1} \oplus f_{2} =_{{\cal B}^{2}} f_{1}\,f_{2},\quad f_{int} = f\ominus f_{ind} =_{{\cal B}^{2}} \frac{f}{f_{1}\,f_{2}},
$$
the following holds: 
\begin{itemize}
\item[(i)] The independent part $f_{ind}$ is the unique orthogonal projection of $f$ onto ${\cal B}_{ind}^{2}(\Omega)$, defined as the direct sum of ${\cal B}^{2}_{\{1\}}(\Omega)$ and ${\cal B}^{2}_{\{2\}}(\Omega)$, i.e., ${\cal B}_{ind}^{2}(\Omega)={\cal B}^{2}_{\{1\}}(\Omega)\oplus{\cal B}^{2}_{\{2\}}(\Omega)$.
\item[(ii)] The interactive part $f_{int}$ is the unique orthogonal projection of $f$ onto ${\cal B}_{int}^{2}(\Omega)$, being the orthogonal complement of ${\cal B}_{ind}^{2}(\Omega)$ in ${\cal B}^{2}(\Omega)$ . 
\end{itemize}
\end{theorem}

\begin{corollary}[\cite{genest22}] \label{Bdecomp}
Any bivariate density $f\in{\cal B}^{2}(\Omega)$ can be orthogonally decomposed as
$$
f = f_{1}\oplus f_{2}\oplus f_{int} = f_{ind}\oplus f_{int}.
$$
\end{corollary}

Orthogonal decomposition in ${\cal B}^{2}(\Omega)$ is a strong theoretical result for the study of the dependence structure. However, since the decomposition is defined in ${\cal B}^{2}(\Omega)$, it prevents direct use of standard (functional data) methods embedded in the $L^{2}$ space, which would have to be adjusted for Bayes spaces. It is, therefore, highly desirable to perform the dependence analysis (or any related FDA analysis) in the $L^{2}$ space. This is possible using the \textit{clr transformation} \cite{boogaart14}, which represents a one-to-one transformation between ${\cal B}^{2}(\Omega)$ and
$$
L^{2}_{0}(\Omega) = \left\{ h:\Omega\rightarrow\mathbb{R}\ \Big\vert\ \iint_{\Omega} h^{2}(x,y)\,\mbox{d}x\,\mbox{d}y < \infty,\ \iint_{\Omega} h(x,y)\,\mbox{d}x\,\mbox{d}y = 0\right\},
$$
defined for every $f\in{\cal B}^{2}(\Omega)$ as
$$
f^{c} = \mbox{clr}\left(f\right) = \ln f - \frac{1}{(d-c)\,(b-a)}\iint_{\Omega}\ln f(x,y)\,\mbox{d}x\,\mbox{d}y.
$$
As a result,
$$
\iint_{\Omega}f^{c}(x,y)\,\mbox{d}x\,\mbox{d}y = 0;
$$
moreover, for $f,\,g\in{\cal B}^{2}(\Omega)$ and $\alpha\in\mathbb{R}$ it follows that
$$
\mbox{clr}(f\oplus g) = \mbox{clr}(f) + \mbox{clr}(g),\quad \mbox{clr}(\alpha\odot f) = \alpha \cdot \mbox{clr}(f).
$$
Consequently, geometric marginals can be expressed in $L^{2}_{0}(\Omega)$ using the following assertions.
\\
\begin{theorem}[\cite{genest22}] \label{clrmarg}
Given any $f\in{\cal B}^{2}(\Omega)$ with geometric marginals $f_{1}$ and $f_{2}$, one has
$$
f_{1}^{c}(x,y) = \frac{1}{
d-c}\int_c^d f^c(x,t)\,\textnormal{d}t,\quad f_{2}^{c}(x,y) = \frac{1}{ b-a}\int_a^b f^{c}(t,y)\,\textnormal{d}t.
$$
\end{theorem}

\begin{theorem}[\cite{genest22}] \label{clrrepr}
The clr transformed independent and interactive parts of a PDF $f$ can be expressed as
$$
f_{ind}^{c} = f_{1}^{c} + f_{2}^{c},\quad f_{int}^{c} = f^{c} - f_{1}^{c} - f_{2}^{c},
$$
and, finally,
$$
f^{c} = f_{1}^{c} + f_{2}^{c} + f_{int}^{c} = f_{ind}^{c} + f_{int}^{c}.
$$
\end{theorem}
In addition, $f^{c}$ (as well as the other functions from Theorem \ref{clrrepr}) can be transformed from $L_{0}^{2}(\Omega)$ back to the original Bayes space ${\cal B}^{2}(\Omega)$ using the inverse clr transformation defined as
\begin{equation} \label{iclr}
\mbox{clr}^{-1}\left(f^c(x,y))\right) =_{{\cal B}^{2}(\Omega)} \exp(f^c(x,y)).
\end{equation}

\section{Spline representation of bivariate clr transformed densities} \label{twodim}
A standard spline approximation of transformed PDFs does not guarantee integration of the resulting spline function to zero, as required. Therefore, we introduce a new spline basis in $L_0^2$ having the desired property. We start with a summary of the univariate case in order to introduce new suitable notation, and then introduce our novel  construction of a bivariate spline basis for PDF representation. Finally, we show that the resulting spline representation allows the same decomposition as PDFs into independent and interactive parts.


\subsection{Bivariate spline representation} \label{bivariaterepr}
In this subsection, we first briefly review the bivariate tensor product $B$-spline representation, which will be needed for further developments. Let $\Omega = [a,b]\times[c,d]\subset \mathbb{R}^2$ be given and let us focus on the construction of the bivariate tensor product spline $s_{kl}(x,y)$. We select two increasing sequences of knots
\begin{align*}
  \Delta\lambda\, = \, \{\lambda_i\}_{i=0}^{g+1}, \quad  a=\lambda_{0}<\lambda_{1}<\ldots<\lambda_{g}<\lambda_{g+1}=b,\\
  \Delta\mu \, = \, \{\mu_j\}_{j=0}^{h+1}, \quad c=\mu_{0}<\mu_{1}<\ldots<\mu_{h}<\mu_{h+1}=d
\end{align*}
and denote the vector space of tensor product splines ${\cal S}_{kl}^{\Delta\lambda,\Delta\mu}(\Omega)$ defined on $\Omega$ of degree $k\in \mathbb{N}_0$ in $x$ with knots $\Delta\lambda$ and of degree $l\in \mathbb{N}_0$ in $y$ with knots $\Delta\mu$. From spline theory \cite{Dierckx,deboor}, it is known that
$$
\dim\left({\cal S}_{kl}^{\Delta\lambda,\Delta\mu}(\Omega)\right)\, = \, (g+k+1)\,(h+l+1).
$$
For the construction of all basis functions of the space ${\cal S}_{kl}^{\Delta\lambda,\Delta\mu}(\Omega)$, it is necessary to introduce additional knots \cite{Dierckx}. 
Without loss of generality, we can add coincident knots and denote the extended sequences of knots as
\noindent
\mbox{$\Delta\Lambda = \{\lambda_i\}_{i=-k}^{g+k+1},\, \Delta\text{M} =\{ \mu_j\}_{j=-l}^{h+l+1}$},
\begin{align}
\lambda_{-k}&=\cdots=\lambda_{0}=a <\lambda_{1}<\ldots<\lambda_{g}<
   b=\lambda_{g+1}=\cdots=\lambda_{g+k+1}\label{exknotsX}, \\
\mu_{-l}&=\cdots=\mu_{0}=c <\mu_{1}<\ldots<\mu_{h}<d=\mu_{h+1}=\cdots=\mu_{h+l+1} \label{exknotsY},
\end{align}
and construct $B$-splines $B_{i}^{k+1}(x)$, $i=-k,\dots,g$ 
of degree $k$ in variable $x$ and $B$-splines $B_{j}^{l+1}(y)$, $j=-l,\dots,h$ of degree $l$ in $y$. The basis of the vector space ${\cal S}_{kl}^{\Delta\lambda,\Delta\mu}(\Omega)$ is formed by functions $B_{ij}^{k+1,l+1}(x,y)$ defined as the products of univariate $B$-splines, 
\begin{equation}\label{basis_prod}
B_{ij}^{k+1,l+1}(x,y) = B_{i}^{k+1}(x)B_{j}^{l+1}(y),\quad i=-k,\dots,g,\ j=-l,\dots,h,
\end{equation}
see \cite{Dierckx}. Consequently, every tensor product spline $s_{kl}(x,y)\in{\cal S}_{kl}^{\Delta\lambda,\Delta\mu}(\Omega)$ has a unique representation
\begin{equation*}
s_{kl}(x,y) \; = \; \sum\limits_{i=-k}^{g} \sum\limits_{j=-l}^{h} \, b_{ij} \, B_{ij}^{k+1,l+1}(x,y),
\end{equation*}
where $b_{ij}$ are its spline coefficients. Let $\mathbf{B}_{k+1}(x)= \left(B_{-k}^{k+1}(x),\dots,B_g^{k+1}(x)\right)^{\top}$ and $\mathbf{B}_{l+1}(y)= \left(B_{-l}^{l+1}(y),\dots,B_h^{l+1}(y)\right)^{\top}$. Then the spline representation can be written in matrix notation as
\begin{equation*}
s_{kl}(x,y) \; = \; \mathbf{B}_{k+1}^{\top}(x)\, \mathbf{B} \, \mathbf{B}_{l+1}(y),
\end{equation*}
where $\mathbf{B}=\left(b_{ij}\right)_{i=-k,j=-l}^{g,h}$ is a matrix of the spline coefficients. In the case of clr transformed bivariate densities, the corresponding spline is desired to fulfill the zero integral constraint. This could be achieved by additional constraints on spline coefficients as done in \citep{hron22} or directly by introducing a spline basis in $L^2_0(\Omega)$. As the former approach has several disadvantages, we proceed to the construction of a suitable basis in the next section.


\subsection{Spline representation in $L_0^2(\Omega)$} \label{ZBconstruct}
The idea of this subsection is to develop a spline basis for representation of clr transformed bivariate densities that fulfill zero integral constraints. Considering the given domain $\Omega$, extended sequences of knots \eqref{exknotsX} and \eqref{exknotsY}, and degrees of spline $k,\ l$ in variables $x,\ y$, respectively, we define the vector space ${\cal Z}_{kl}^{\Delta\lambda, \Delta\mu}(\Omega)$ as a subspace of splines from ${\cal S}_{kl}^{\Delta\lambda, \Delta\mu}(\Omega)$ having zero integral. In order to build its basis, we first consider the univariate case.

Let ${\cal S}_{k}^{\Delta\lambda}([a,b])$ be the vector space of univariate splines $s_{k}(x)$ of degree $k$ on~the~interval $[a,b]$ with the extended sequence of knots \eqref{exknotsX}. It is well known that $\dim\left({\cal S}_{k}^{\Delta\lambda}[a,b]\right) = g+k+1$ and $B$-splines $B_{-k}^{k+1}(x),\,\dots,\,B_{g}^{k+1}(x)$ form a basis of~${\cal S}_{k}^{\Delta\lambda}([a,b])$, see \cite{Dierckx}. A basis of the vector subspace ${\cal Z}_{k}^{\Delta\lambda}([a,b])\subset{\cal S}_{k}^{\Delta\lambda}([a,b])$ is composed of $Z\!B$-splines defined in \cite{compositional} as
$$
Z_{i}^{k+1}(x) = \frac{\mbox{d}}{\mbox{d}x} B_{i}^{k+2}(x),\quad i=-k,\dots,g-1.
$$
When coincident knots \eqref{exknotsX} are considered, $Z\!B$-splines have zero integral, i.e.
\begin{equation}\label{zeroInt}
\int_a^b Z_{i}^{k+1}(x)\,\mbox{d}x \; = \; 0,
\end{equation}
which is a crucial and most desirable property of basis splines for a suitable spline representation in ${\cal Z}_{k}^{\Delta\lambda}([a,b])$. 
Moreover, $Z\!B$-splines can be expressed using $B$-splines as
$$
Z_{i}^{k+1}(x) = (k+1)\left(\frac{B_{i}^{k+1}(x)}{\lambda_{i+k+1}-\lambda_{i}} - \frac{B_{i+1}^{k+1}(x)}{\lambda_{i+k+2}-\lambda_{i+1}}\right),\quad i=-k,\dots,g-1,\\
$$
which can be formulated in matrix notation as
\begin{equation} \label{ZBX_to_B}
\mathbf{Z}_{k+1}(x)\,=\,\mathbf{K}_{g k}\,\mathbf{D}_{\lambda}\,\mathbf{B}_{k+1}(x),
\end{equation}
where $\mathbf{Z}_{k+1}(x)= \left(Z_{-k}^{k+1}(x),\dots,Z_{g-1}^{k+1}(x)\right)^{\top}$ 
and 
\begin{equation} \label{matrix_Dx}
\mathbf{D}_{\lambda}\,=\,(k+1)\left[\mbox{diag}\left(\mathbf{d}_{\lambda}\right)\right]^{-1}
\end{equation}
with vector $\mathbf{d}_{\lambda} = (\lambda_{1} - \lambda_{-k},\ \dots\ ,\,\lambda_{g+k+1} - \lambda_{g})$. Matrix $\mathbf{K}_{ts}=(k_{ij}^{ts})\in\mathbf{R}^{t+s,t+s+1}$ is defined as a matrix with elements
\begin{equation}\label{matrixK}
k_{ij}^{ts} \, = \, 
\begin{cases}
     1 & \mbox{if} \; i=j, \; j=1,\ldots,t+s,\\
    -1 & \mbox{if} \; i=j-1, \; j=2,\ldots,t+s+1,\\
     0 & \mbox{otherwise}.
\end{cases}
\end{equation}
More details on $Z\!B$-splines can be found in \cite{compositional}.
As already mentioned, $B$-splines $B_{-k}^{k+1}(x),\,\dots,\,B_{g}^{k+1}(x)$ form a basis of ${\cal S}_{k}^{\Delta\lambda}([a,b])$, therefore, if we multiply the vector $\mathbf{B}_{k+1}(x)$ with an arbitrary regular transformation matrix, we get a~vector of splines which are linearly independent  and form a basis in ${\cal S}_{k}^{\Delta\lambda}([a,b])$. Consider the full-rank transformation matrix
\begin{equation} \label{matrixT_lambda}
\mathbf{T}_{\lambda} = \begin{pmatrix} \mathbf{K}_{gk}\,\mathbf{D}_{\lambda} \\ \mathbf{e}_{g+k+1}^{\top} \end{pmatrix}\in\mathbb{R}^{(g+k+1)\times(g+k+1)},
\end{equation}
where matrices $\mathbf{D}_{\lambda}$ and $\mathbf{K}_{gk}$ are defined in \eqref{matrix_Dx} and \eqref{matrixK}, respectively, i.e.
$$
\mathbf{T}_{\lambda} = \,\begin{pmatrix} (k+1)\left(\frac{1}{\lambda_{1}-\lambda_{k}} - \frac{1}{\lambda_{2}-\lambda_{k+1}}\right) & 0 & \dots  & 0 \\ 
0 & (k+1)\left(\frac{1}{\lambda_{2}-\lambda_{k+1}} - \frac{1}{\lambda_{3}-\lambda_{k+2}}\right) & \dots & 0 \\
\vdots & \vdots & \ddots & \vdots\\
0 & 0 & \dots & (k+1)\left(\frac{1}{\lambda_{g+k}-\lambda_{g-1}} - \frac{1}{\lambda_{g+k+1}-\lambda_{g}}\right) \\
1 & 1 & \dots & 1 \end{pmatrix}.
$$
The symbol $\mathbf{e}_{n}$ stands for vector $\mathbf{e}_{n} = \left(1,\,\dots,\,1\right)^{\top}\in\mathbb{R}^{n}$. Then the vector
\begin{equation*}
\overline{\mathbf{Z}}_{k+1}(x) = \mathbf{T}_{\lambda}\,\mathbf{B}_{k+1}(x)
\end{equation*}
forms a new basis of ${\cal S}_{k}^{\Delta\lambda}([a,b])$. Considering \eqref{ZBX_to_B} and
$$
\mathbf{e}_{g+k+1}^{\top}\,\mathbf{B}_{k+1}(x) = \sum_{i=-k}^{g} B_{i}^{k+1}(x) = 1,
$$
the vector $\overline{\mathbf{Z}}_{k+1}(x)$ can also be written as
\begin{equation}\label{transf_basis_x_2}
\overline{\mathbf{Z}}_{k+1}(x) = \left(Z_{-k}^{k+1}(x),\,\dots,\,Z_{g-1}^{k+1}(x),\,1\right)^{\top} = \left(\mathbf{Z}_{k+1}^{\top}(x),1\right)^{\top}.
\end{equation}
Since the functions
$Z_{-k}^{k+1}(x),\,\dots,\,Z_{g-1}^{k+1}(x)$ are linearly independent splines having zero integral, they form a basis of ${\cal Z}_{k}^{\Delta\lambda}([a,b])$. Clearly,
$$
\dim\left({\cal Z}_{k}^{\Delta\lambda}([a,b])\right) = g+k.
$$
Similarly, a basis of the vector space ${\cal S}_{l}^{\Delta\mu}([c,d])$ of splines $s_{l}(y)$ of degree $l$ defined on the interval $[c,d]$ with extended sequence of knots \eqref{exknotsY} is generated with $B$-spline basis
$B_{-l}^{l+1}(y),\,\dots,\,B_{h}^{l+1}(y)$
and $\dim({\cal S}_{l}^{\Delta\mu}([c,d])) = h+l+1$. Analogously to above, we define $ZB$-splines $Z_{j}^{l+1}(y),\ j=-l,\,\dots,\,h-1$, which can be expressed in terms of $B$-splines in matrix notation as
\begin{equation}\label{ZBY_to_B}
\mathbf{Z}_{l+1}(y) = \mathbf{K}_{hl}\,\mathbf{D}_{\mu}\,\mathbf{B}_{l+1}(y),
\end{equation}
where $\mathbf{K}_{hl}$ is defined in \eqref{matrixK} and
\begin{equation}\label{matrix_Dy}
\mathbf{D}_{\mu}=(l+1)[\textnormal{diag}(\mathbf{d}_{\mu})]^{-1}
\end{equation}
with $\mathbf{d}_{\mu}~=~(\mu_{1} - \mu_{-l},\,\dots\, ,\,\mu_{h+l+1}-\mu_{h})$. 
Transformation of a vector of $B$-splines $\mathbf{B}_{l+1}(y)$ with matrix
\begin{equation} \label{matrixT_mu}
\mathbf{T}_{\mu} = \begin{pmatrix} \mathbf{K}_{hl}\,\mathbf{D}_{\mu} \\ \mathbf{e}_{h+l+1}^{\top} \end{pmatrix}\in\mathbb{R}^{(h+l+1)\times(h+l+1)}
\end{equation}
implies that the vector
\begin{equation}\label{transf_basis_y}
\overline{\mathbf{Z}}_{l+1}(y)=\mathbf{T}_{\mu}\,\mathbf{B}_{l+1}(y) = \left(Z_{-l}^{l+1}(y),\,\dots,\,Z_{h-1}^{l+1}(y),\,1\right)^{\top} = \left(\mathbf{Z}_{l+1}^{\top}(y),1\right)^{\top}
\end{equation}
is a vector of linearly independent splines, 
which form a basis of ${\cal S}_{l}^{\Delta\mu}([c,d])$. Consequently, splines 
$Z_{-l}^{l+1}(y),\,\dots,\,Z_{h-1}^{l+1}(y)$ form a basis of ${\cal Z}_{l}^{\Delta\mu}([c,d])$. In addition,
$$
\dim({\cal Z}_{l}^{\Delta\mu}([c,d])) = h+l.
$$

By considering the introduced notation, we can now proceed to extend this approach to the bivariate case. We use products of univariate $B$-splines defined in \eqref{basis_prod} as basis functions of the vector space ${\cal S}_{kl}^{\Delta\lambda,\Delta\mu}(\Omega)$. By transforming them with matrices $\mathbf{T}_{\lambda}$ and $\mathbf{T}_{\mu}$ defined in \eqref{matrixT_lambda} and \eqref{matrixT_mu}, 
we obtain a matrix of functions
\begin{align*}
\overline{\mathbf{Z}}_{k+1,l+1}(x,y) & =
\overline{\mathbf{Z}}_{k+1}(x)\,\overline{\mathbf{Z}}_{l+1}^{\top}(y)
= \begin{pmatrix}
Z_{-k,-l}^{k+1,l+1}(x,y) & \dots & Z_{-k,h-1}^{k+1,l+1}(x,y) & Z_{-k}^{k+1}(x,y) \\[0.3cm]
\vdots & \ddots & \vdots & \vdots \\[0.3cm]
Z_{g-1,-l}^{k+1,l+1}(x,y) & \dots & Z_{g-1,h-1}^{k+1,l+1}(x,y) & Z_{g-1}^{k+1}(x,y) \\[0.3cm]
Z_{-l}^{l+1}(x,y) & \dots & Z_{h-1}^{l+1}(x,y) & 1
\end{pmatrix},
\end{align*}
which are linearly independent bivariate splines and form a basis of ${\cal S}_{kl}^{\Delta\lambda,\Delta\mu}(\Omega)$. 
Here we used the notation
\begin{equation}\label{ZBprod}
Z_{ij}^{k+1,l+1}(x,y) \; = \; Z_{i}^{k+1}(x)\,Z_{j}^{l+1}(y)
\end{equation}
and
\begin{equation}\label{ConstZ}
Z_{i}^{k+1}(x,y) \, = \, Z_{i}^{k+1}(x),\ \qquad
Z_{j}^{l+1}(x,y) \, = \, Z_{j}^{l+1}(y)
\end{equation}
for $i=-k,\dots,g-1$, $j=-l,\dots,h-1$. Considering the zero integral constraint, definitions \eqref{ZBprod}, \eqref{ConstZ} and relation \eqref{zeroInt} yield that
\begin{equation*}
\iint\limits_{\Omega} Z_{ij}^{k+1,l+1}(x,y) \, \mbox{d}x\,\mbox{d}y\, = \, \int_{a}^{b} Z_{i}^{k+1}(x)\,\mbox{d}x \,  \int_c^d Z_{j}^{l+1}(y) \, \mbox{d}y \, = \, 0,
\end{equation*}
and similarly
\begin{equation*}
\iint\limits_{\Omega} Z_{i}^{k+1}(x,y) \, \mbox{d}x\,\mbox{d}y\, = 
\iint\limits_{\Omega} Z_{j}^{l+1}(x,y) \, \mbox{d}x\,\mbox{d}y\,  =  \, 0.
\end{equation*}
This allows us to summarize the results in the following theorem. 
\\
\begin{theorem}\label{basisJM}
The functions $Z_{ij}^{k+1,l+1}(x,y)$, $ Z_{i}^{k+1}(x,y)$ and $Z_{j}^{l+1}(x,y)$, for $i=-k,\dots,g-1$, $j=-l,\dots,h-1$ defined on the extended sequences of knots (\ref{exknotsX}), (\ref{exknotsY}) form a basis of the space ${\cal Z}_{kl}^{\Delta\lambda,\Delta\mu}(\Omega)$, where
$$
\dim\left({\cal Z}_{kl}^{\Delta\lambda,\Delta\mu}(\Omega)\right) = (g+k+1)(h+l+1)-1 = (g+k)(h+l)+g+k+h+l.
$$
\end{theorem}

\noindent Examples of the introduced basis functions can be seen in Figure \ref{basis_ZB}. Theorem~\ref{basisJM} can be used to construct the spline representation in the considered space ${\cal Z}_{kl}^{\Delta\lambda,\Delta\mu}(\Omega)$. Note that functions $Z_{ij}^{k+1,l+1}(x,y)$ have zero marginals by construction, which is a property shared by the clr representation of the interactive parts of densities \citep{hron22}. We will use this advantage in the sequel.

\begin{figure}[h]
    \centering
    \includegraphics[width=0.35\textwidth]{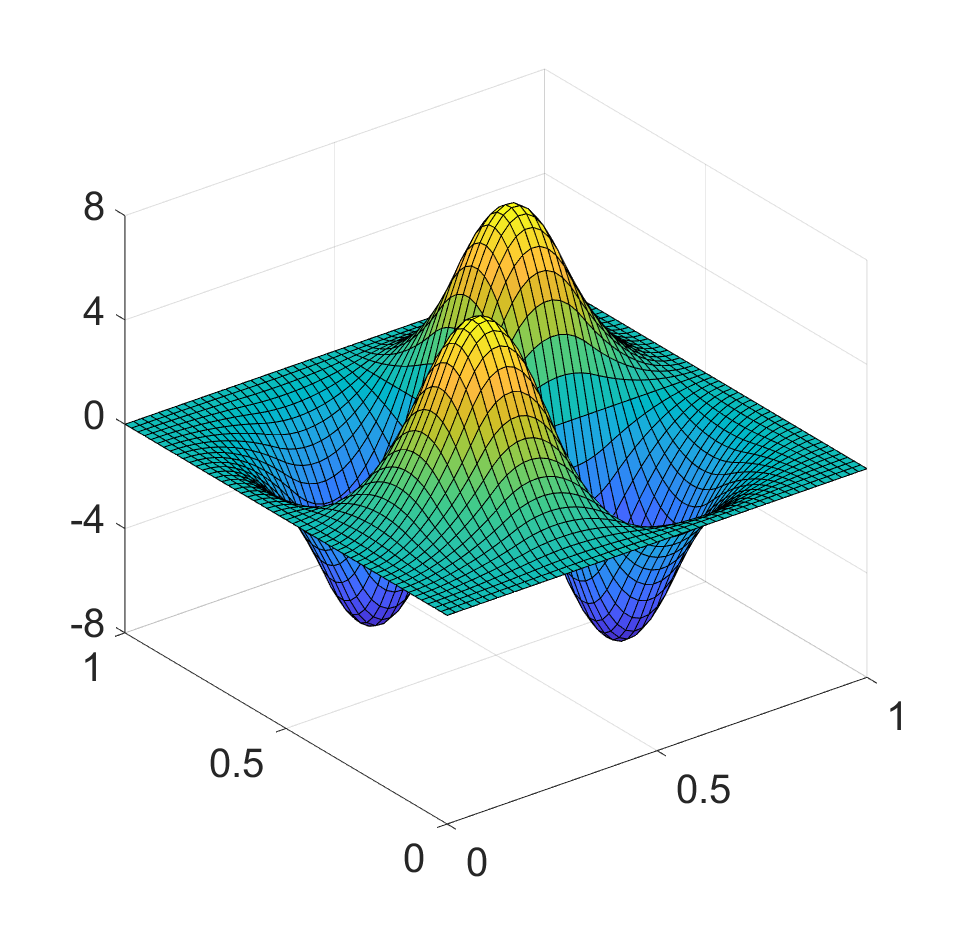}
    \includegraphics[width=0.35\textwidth]{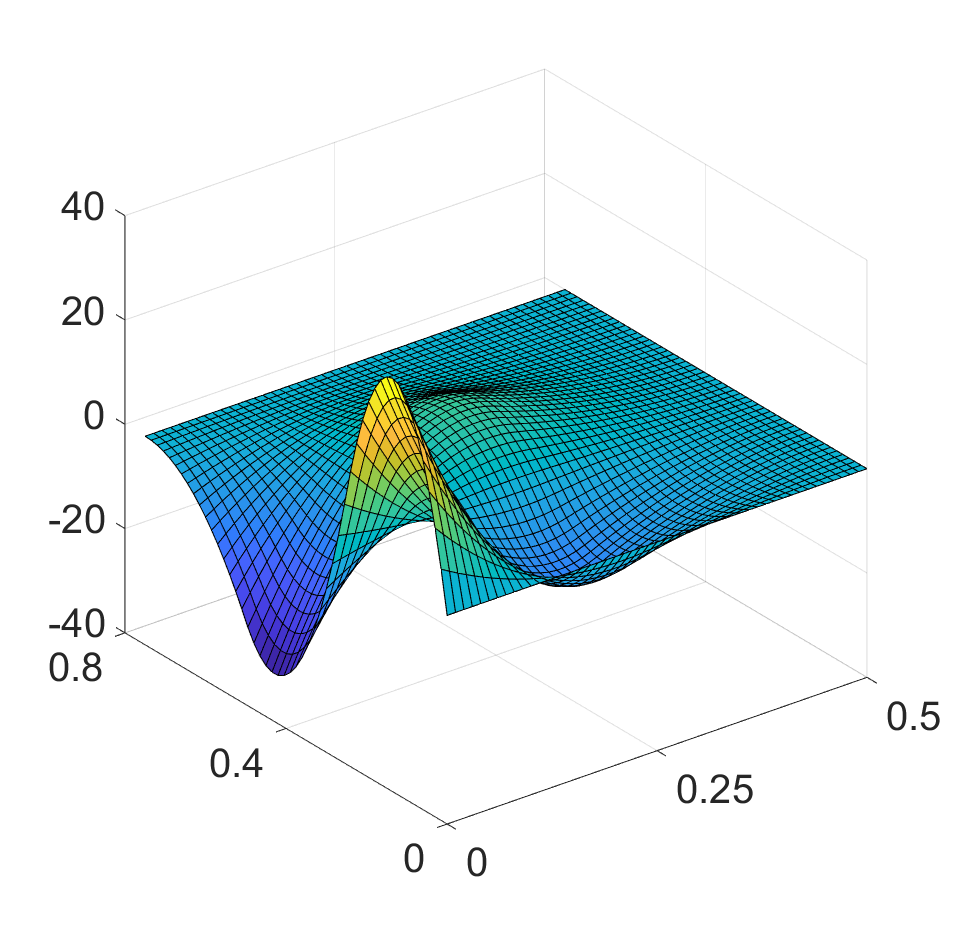}\\
    \includegraphics[width=0.35\textwidth]{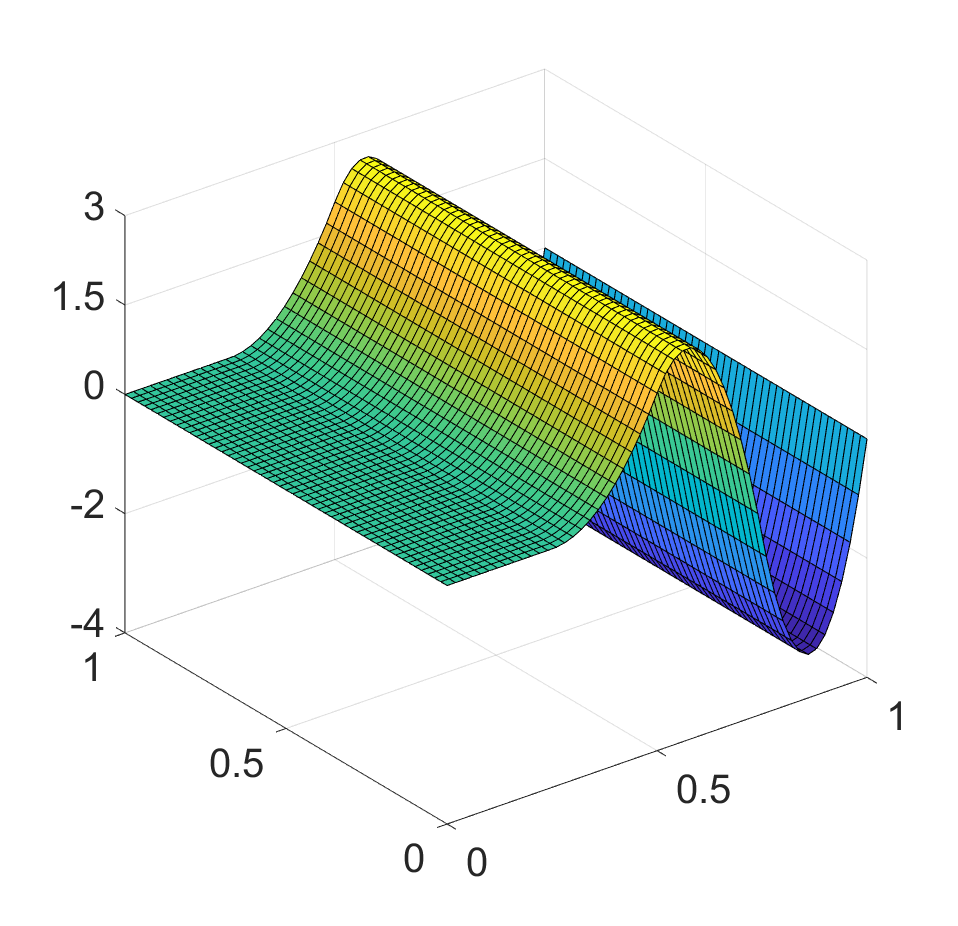}
    \includegraphics[width=0.35\textwidth]{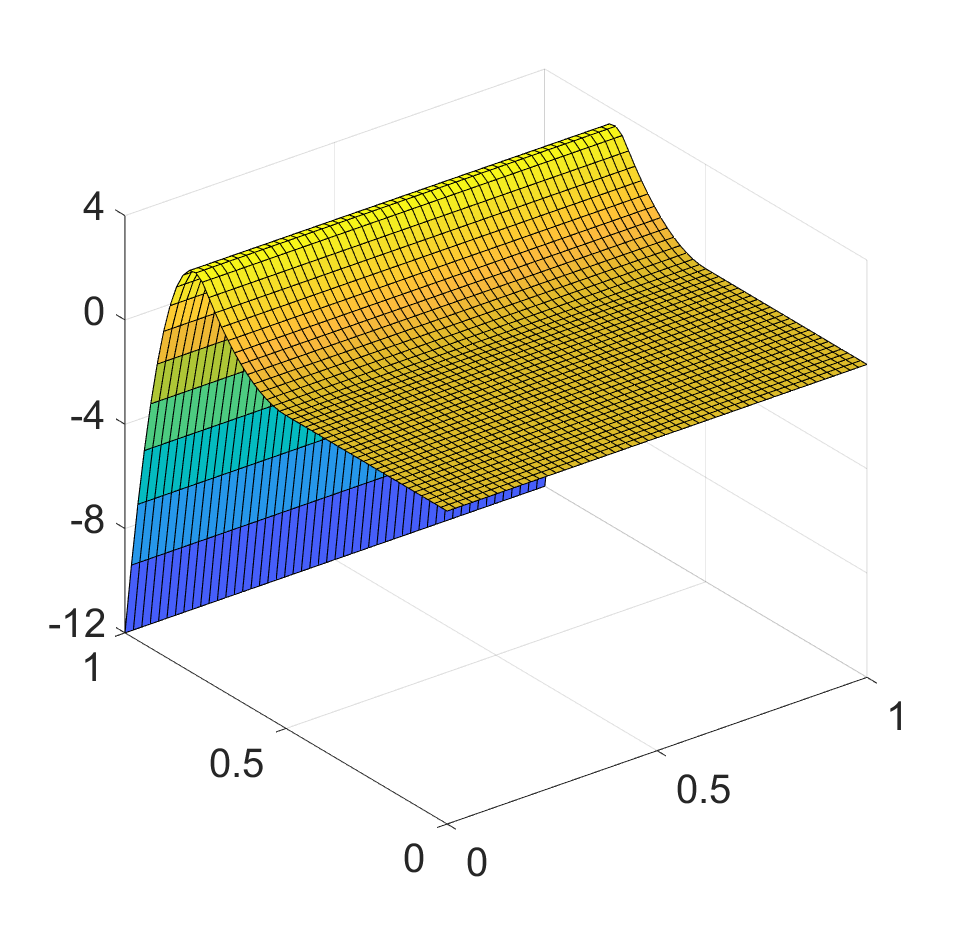}
    \caption{Two examples of  $Z\!B$-spline  basis functions $Z_{ij}^{k+1,l+1}(x,y)$ (top) and examples of  $Z\!B$-spline basis functions $Z_{i}^{k+1}(x,y)$ (bottom, left) and $Z_{j}^{l+1}(x,y)$ (bottom, right)}
    \label{basis_ZB}
\end{figure}

\begin{corollary}
Every spline $s_{kl}(x,y)\in{\cal Z}_{kl}^{\Delta\lambda,\Delta\mu}(\Omega)$ has a unique representation
\begin{equation}\label{ZB_repr_1}
s_{kl}(x,y) \, = \, \sum_{i=-k}^{g-1}\sum_{j=-l}^{h-1}z_{ij}\,Z_{ij}^{k+1,l+1}(x,y) + \sum_{i=-k}^{g-1}v_{i}\,Z_{i}^{k+1}(x,y) + \sum_{j=-l}^{h-1}\,u_{j}\,Z_{j}^{l+1}(x,y),
\end{equation}
where $z_{ij},v_{i},u_{j},\ i=-k,\dots,g-1,\ j=-l,\dots,h-1$ are corresponding spline coefficients.
\end{corollary}

\noindent
Spline representation \eqref{ZB_repr_1} can be rewritten using definitions \eqref{ZBprod}, \eqref{ConstZ} as
\begin{equation*}
s_{kl}(x,y) \, = \, \sum_{i=-k}^{g-1}\sum_{j=-l}^{h-1}z_{ij}\,Z_{i}^{k+1}(x)\, Z_{j}^{l+1}(y) + \sum_{i=-k}^{g-1}v_{i}\,Z_{i}^{k+1}(x) + \sum_{j=-l}^{h-1}\,u_{j}\,Z_{j}^{l+1}(y).
\end{equation*}
Let $\mathbf{Z}=\left(z_{ij}\right)_{i=-k, j=-l}^{g-1, h-1}$, $\mathbf{u}=(u_{-l},\dots,u_{h-1})^{\top}$ and $\mathbf{v}=(v_{-k},\dots,v_{g-1})^{\top}$. Then the spline $s_{kl}(x,y)\in{\cal Z}_{kl}^{\Delta\lambda\Delta\mu}(\Omega)$ can be expressed in matrix notation as
\begin{equation}\label{ZB_repr_2}
s_{kl}(x,y) = \overline{\mathbf{Z}}_{k+1}^{\top}(x)\,\mathbf{R}\,\overline{\mathbf{Z}}_{l+1}(y), 
\end{equation}
where
$$
\mathbf{R}= \begin{pmatrix} \mathbf{Z} & \mathbf{v} \\ \mathbf{u}^{\top} & 0 \end{pmatrix}
$$
is the matrix of spline coefficients and vectors $\overline{\mathbf{Z}}_{k+1}(x)$ and $\overline{\mathbf{Z}}_{l+1}(y)$ are expressed in \eqref{transf_basis_x_2} and \eqref{transf_basis_y}. Using a Kronecker tensor product, the spline $s_{kl}(x,y)$ can also be written as
\begin{equation*}
s_{kl}(x,y) = \left(\overline{\mathbf{Z}}_{l+1}^{\top}(y)\otimes\overline{\mathbf{Z}}_{k+1}^{\top}(x)\right)cs(\mathbf{R}),
\end{equation*}
where $cs(\mathbf{R})$ is the vectorized form of the matrix $\mathbf{R}$, columnwise.
However, the vectorized form of the matrix $\mathbf{R}$ is not always convenient in applications. Therefore, we express the spline in the form
\begin{equation}\label{ZB_repr_3}
s_{kl}(x,y) = \left(\overline{\mathbf{Z}}_{l+1}^{\top}(y)\otimes\overline{\mathbf{Z}}_{k+1}^{\top}(x)\right)\,\mathbf{P}\begin{pmatrix}cs(\mathbf{Z})\\ \mathbf{v}\\ \mathbf{u} \end{pmatrix},
\end{equation}
where $\mathbf{P}$ is a corresponding matrix composed of zeros and ones.

From a practical perspective, it is very convenient to express the bivariate spline $s_{kl}(x,y)\in{\cal Z}_{kl}^{\Delta\lambda,\Delta\mu}(\Omega)$ using a~\mbox{$B$-spline} basis, similarly as it is possible to obtain (univariate) $Z\!B$-splines in terms of $B$-splines using \eqref{ZBX_to_B} and \eqref{ZBY_to_B}. It enables to express $Z\!B$-splines with classical $B$-splines easily, which has a positive impact in applications, since standard $B$-spline software packages can be employed effectively. The relation is stated in the following theorem.

\begin{theorem} \label{biv_mn}
Spline $s_{kl}(x,y)\in{\cal Z}_{kl}^{\Delta\lambda,\Delta\mu}(\Omega)$ defined in \eqref{ZB_repr_1} can be expressed as
\begin{align*}
s_{kl}(x,y) = \mathbf{B}_{k+1}^{\top}(x)\,\mathbf{B}\,\mathbf{B}_{l+1}(y)
\end{align*}
with the matrix of $B$-spline coefficients
$$
\mathbf{B} = \mathbf{D}_{\lambda}^{\top}\,\mathbf{K}_{g k}^{\top}\,\mathbf{Z}\,\mathbf{K}_{h l}\,\mathbf{D}_{\mu} + \left(\mathbf{V}\,\mathbf{K}_{g k}\,\mathbf{D}_{\lambda}\right)^{\top} + \mathbf{U}\,\mathbf{K}_{hl}\,\mathbf{D}_{\mu},
$$
where $\mathbf{V} = \mathbf{e}_{h+l+1} \otimes \mathbf{v}^{\top}$, $\mathbf{U} = \mathbf{e}_{g+k+1} \otimes \mathbf{u}^{\top}$ and $\mathbf{D}_{\lambda},\,\mathbf{D}_{\mu},\,\mathbf{K}_{g k},\,\mathbf{K}_{h l}$ are defined in \eqref{matrix_Dx}, \eqref{matrix_Dy} and \eqref{matrixK}.
\end{theorem}
\begin{proof}
The spline representation \eqref{ZB_repr_1} can be written in matrix notation as
$$
s_{kl}(x,y) = \mathbf{Z}_{k+1}^{\top}(x)\,\mathbf{Z}\,\mathbf{Z}_{l+1}(y) + \mathbf{Z}_{k+1}^{\top}(x)\,\mathbf{V}^{\top}\,\mathbf{B}_{l+1}(y) + \mathbf{B}_{k+1}^{\top}(x)\,\mathbf{U}\,\mathbf{Z}_{l+1}(y),
$$
since
\begin{align}\label{VBv}
\mathbf{V}^{\top}\,\mathbf{B}_{l+1}(y) &= \begin{pmatrix}
    v_{-k} & v_{-k} & \dots & v_{-k} \\
    \vdots & \vdots & & \vdots \\
    v_{g-1} & v_{g-1} & \dots & v_{g-1}
\end{pmatrix}\begin{pmatrix} B_{-l}^{l+1}(y) \\ 
\vdots \\ B_{h}^{l+1}(y) \end{pmatrix} \\ \nonumber
&= \left(
v_{-k} \sum\limits_{i=-l}^{h} B_{i}^{l+1}(y),\dots,v_{g-1} \sum\limits_{i=-l}^{h} B_{i}^{l+1}(y)
\right)^{\top} = \left( v_{-k},\dots,v_{g-1}\right)^{\top} = 
\mathbf{v},
\end{align}
and analogously $\mathbf{B}_{k+1}^{\top}(x)\, \mathbf{U} = \mathbf{u}^{\top}$, due to
$$
\sum_{i=-k}^{g} B_{i}^{k+1}(x) = \sum_{j=-l}^{h} B_{j}^{l+1}(y) = 1 \quad \forall x, y\in\Omega. 
$$
The expression of spline $s_{kl}(x,y)$ can be further modified using relations \eqref{ZBX_to_B} and \eqref{ZBY_to_B} as
\begin{align*}
s_{kl}(x,y) &= \mathbf{B}_{k+1}^{\top}(x)\,\mathbf{D}_{\lambda}^{\top}\,\mathbf{K}_{g k}^{\top}\,\mathbf{Z}\,\mathbf{K}_{h l}\mathbf{D}_{\mu}\mathbf{B}_{l+1}(y) + \mathbf{B}_{k+1}^{\top}(x)\mathbf{D}_{\lambda}^{\top}\mathbf{K}_{g k}^{\top}\mathbf{V}^{\top}\mathbf{B}_{l+1}(y) + \mathbf{B}_{k+1}^{\top}(x)\mathbf{U}\mathbf{K}_{h l}\mathbf{D}_{\mu}\mathbf{B}_{l+1}(y) \\
&=\mathbf{B}_{k+1}^{\top}(x)\,\left[\mathbf{D}_{\lambda}^{\top}\,\mathbf{K}_{g k}^{\top}\,\mathbf{Z}\,\mathbf{K}_{h l}\,\mathbf{D}_{\mu} + \left(\mathbf{V}\,\mathbf{K}_{g k}\,\mathbf{D}_{\lambda}\right)^{\top} + \mathbf{U}\,\mathbf{K}_{hl}\,\mathbf{D}_{\mu} \right]\,\mathbf{B}_{l+1}(y) \\
&= \mathbf{B}_{k+1}^{\top}(x)\,\mathbf{B}\,\mathbf{B}_{l+1}(y),
\end{align*}
where
$$
\mathbf{B} = \mathbf{D}_{\lambda}^{\top}\,\mathbf{K}_{g k}^{\top}\,\mathbf{Z}\,\mathbf{K}_{h l}\,\mathbf{D}_{\mu} + \left(\mathbf{V}\,\mathbf{K}_{g k}\,\mathbf{D}_{\lambda}\right)^{\top} + \mathbf{U}\,\mathbf{K}_{hl}\,\mathbf{D}_{\mu}
$$
represents the matrix of $B$-spline coefficients as stated in the assertion.
\end{proof}

The orthogonal decomposition of density functions was summarized in Section~\ref{Bayes}. The spline representation \eqref{ZB_repr_1} for clr transformed densities in $L^{2}_{0}(\Omega)$ has the advantage that it can be decomposed accordingly, allowing interchangeability of decomposition and spline approximation, as follows.
\\
\begin{theorem} \label{decomposition}
Let $s_{kl}(x,y)\in{\cal Z}_{kl}^{\Delta\lambda,\Delta\mu}(\Omega)$. Then
\begin{equation*}
s_{kl}^{int}(x,y) = 
\sum_{i=-k}^{g-1}\sum_{j=-l}^{h-1}z_{ij}\,Z_{ij}^{k+1,l+1}(x,y)
\end{equation*}
represents the \textit{interactive part} of spline $s_{kl}(x,y)$ and
$$
s_{kl}^{ind}(x,y) = 
\sum_{i=-k}^{g-1}v_{i}\,Z_{i}^{k+1}(x,y) + \sum_{j=-l}^{h-1}u_{j}\,Z_{j}^{l+1}(x,y)
$$
is the \textit{independent part} with univariate \textit{clr marginals}
$$
s_k^1(x,y) = 
\sum_{i=-k}^{g-1}v_{i}\,Z_{i}^{k+1}(x,y),\quad
s_l^2(x,y) = 
\sum_{j=-l}^{h-1}u_{j}\,Z_{j}^{l+1}(x,y).
$$
\end{theorem}
\begin{proof}
Using Theorem \ref{clrmarg}, clr marginals are obtained as
\begin{align*}
s_{k}^{1}(x,y) &= \frac{1}{d-c}\int\limits_{c}^{d} s_{kl}(x,y)\,\mbox{d}y \\
&= \frac{1}{d-c}\Bigg(\int\limits_{c}^{d}\sum_{i=-k}^{g-1}\sum_{j=-l}^{h-1}z_{ij}\,Z_{ij}^{k+1,l+1}(x,y)\,\mbox{d}y + \int\limits_{c}^{d} \sum_{i=-k}^{g-1}v_{i}\,Z_{i}^{k+1}(x,y)\,\mbox{d}y + \int\limits_{c}^{d} \sum_{j=-l}^{h-1}u_{j}\,Z_{j}^{l+1}(x,y)\,\mbox{d}y\Bigg) \\
&= \frac{1}{d-c}\Bigg(\sum_{i=-k}^{g-1}\sum_{j=-l}^{h-1}z_{ij}\,Z_{i}^{k+1}(x)\int\limits_{c}^{d}Z_{j}^{l+1}(y)\,\mbox{d}y + \sum_{i=-k}^{g-1}v_{i}\,Z_{i}^{k+1}(x)\int\limits_{c}^{d}\, 1\,\mbox{d}y + \sum_{j=-l}^{h-1}u_{j} \int\limits_{c}^{d}Z_{j}^{l+1}(y)\,\mbox{d}y\Bigg) \\
&= \sum_{i=-k}^{g-1}v_{i}\,Z_{i}^{k+1}(x)
\end{align*}
and analogously
\begin{align*}
s_{l}^{2}(x,y) = \frac{1}{b-a}\int\limits_{a}^{b} s_{kl}(x,y)\, \mbox{d}x = \sum_{j=-l}^{h-1}u_{j}\,Z_{j}^{l+1}(y).
\end{align*}
From Theorem \ref{clrrepr}, the independent and interactive parts can be expressed as
\begin{align*}
s_{kl}^{ind}(x,y) = s_{k}^{1}(x,y) + s_{l}^{2}(x,y) = \sum_{i=-k}^{g-1}v_{i}\,Z_{i}^{k+1}(x,y) + \sum_{j=-l}^{h-1}u_{j}\,Z_{j}^{l+1}(x,y)
\end{align*}
and
\begin{align*}
s_{kl}^{int}(x,y) &= s_{kl}(x,y) - s_{kl}^{ind}(x,y) = \sum_{i=-k}^{g-1}\sum_{j=-l}^{h-1}z_{ij}\,Z_{ij}^{k+1,l+1}(x,y).
\end{align*}
\end{proof}

By construction, the spline basis of ${\cal Z}_{kl}^{\Delta\lambda,\Delta\mu}(\Omega)$ is not orthogonal. However, the orthogonality holds mutually between 
$Z_{ij}^{k+1,l+1}(x,y)$, $Z_{i}^{k+1}(x,y)$ and $Z_{j}^{l+1}(x,y)$ for any $i,j$, reflecting hereby the orthogonal decomposition of PDFs from Corollary \ref{Bdecomp}. 

Additional properties of interactive and independent parts $s_{kl}^{int}(x,y)$, $s_{kl}^{ind}(x,y)$ are stated in \ref{secA1}.

\section{Smoothing splines in $L^{2}_{0}(\Omega)$} \label{ZBsmoothing}
In the context of FDA, smoothing splines are very popular for modelling and approximation of the observed data according to the combination of two criteria: sufficient data fitting and smoothness of the resulting spline. In this section, we provide a formula for tensor product smoothing splines in $Z\!B$-spline representation. Although the following derivations are rather lengthy and technical, they cannot be easily skipped since the notation needed for the final formula is introduced and connections with already obtained results are linked.

We assume domain $\Omega = [a,b]\times[c,d]$ and that the density after clr-transformation is observed in the form of values $f_{ij}$ at observation points
$(x_i,y_j)~\in~\Omega,\ i=1,\dots,m,\ j=1,\dots,n$, i.e. $f_{ij}$ represents measurements of an $L^2_0(\Omega)$ function, e.g.\ obtained as a clr-transformed histogram aggregation of a larger dataset. To employ $Z\!B$-splines, let $k,l\in\mathbb{N}_{0}$ be degrees of splines in variables $x$ and $y$ with chosen extended sequences of knots \eqref{exknotsX}, \eqref{exknotsY}, respectively. The aim is to find the spline $s_{kl}(x,y)\in{\cal Z}_{kl}^{\Delta\lambda,\Delta\mu}(\Omega)$, which minimizes the functional
\begin{equation}\label{compfcional}
J_{p,q}(s_{kl}) = \sum_{i=1}^{m} \sum_{j=1}^{n} \big[f_{ij}-s_{kl}(x_i,y_j) \big]^2 + \rho \iint\limits_{\Omega} \big[ s_{kl}^{(p,q)}(x,y) \big]^2\,\mbox{d}x\,\mbox{d}y,
\end{equation}
where $p\in\lbrace 0,1,\dots,k-1\rbrace$ and $q\in\lbrace 0,1,\dots,l-1\rbrace$ denote degrees of derivatives of spline $s_{kl}(x,y)$ in variables $x$ and $y$. The first part of the functional \eqref{compfcional} is an ordinary least squares criterion. The second part can be seen as a roughness penalty, which controls the level of smoothness of spline $s_{kl}(x,y)$ and can be modulated by choosing the value of the so-called \textit{smoothing parameter} $\rho > 0$. Higher values of parameter $\rho$ lead to a higher level of smoothness of the resulting spline approximation.
Functional \eqref{compfcional} is also referred to as a \textit{smoothing functional}, and the minimizing spline $s_{kl}(x,y)$ is called \textit{smoothing spline}. 
Just to recall that for a smoothing problem we assume that the number of the given data is greater than the dimension of the corresponding space 
${\cal Z}_{kl}^{\Delta\lambda,\Delta\mu}(\Omega)$,
i.e.
$m\,n \; > \;(g+k)(h+l)+g+k+h+l$.

Using representation \eqref{ZB_repr_2}, function values of the spline $s_{kl}(x,y)$ at given $\mathbf{x}=(x_1,\dots,x_m)^{\top}$ and $\mathbf{y}=(y_1,\dots,y_n)^{\top}$, collected in matrix $s_{kl}(\mathbf{x},\mathbf{y}) = (s_{kl}(x_i,y_j))_{i=1,\,j=1}^{m,\,n}$, can be expressed as
\begin{align}\label{ZB_repr_val}
s_{kl}(\mathbf{x},\mathbf{y}) &= \overline{\mathbf{Z}}_{k+1}^{\top}(\mathbf{x})\,\mathbf{R}\,\overline{\mathbf{Z}}_{l+1}(\mathbf{y}),
\end{align}
where $\overline{\mathbf{Z}}_{k+1}(\mathbf{x}) = \begin{pmatrix}
\mathbf{Z}_{k+1}^{\top}(\mathbf{x}) \\ \mathbf{e}_{m}^{\top}\end{pmatrix}$ and $\overline{\mathbf{Z}}_{l+1}(\mathbf{y}) = \begin{pmatrix}
\mathbf{Z}_{l+1}^{\top}(\mathbf{y}) \\ \mathbf{e}_{n}^{\top}\end{pmatrix}$ with matrices
\begin{equation}\label{Z_matrices}
\mathbf{Z}_{k+1}(\mathbf{x}) = \begin{pmatrix} Z_{-k}^{k+1}(x_1) & \dots & Z_{g-1}^{k+1}(x_1) \\
\vdots & & \vdots \\
Z_{-k}^{k+1}(x_m) & \dots & Z_{g-1}^{k+1}(x_m)
\end{pmatrix},\quad \mathbf{Z}_{l+1}(\mathbf{y}) = \begin{pmatrix} Z_{-l}^{l+1}(y_1) & \dots & Z_{h-1}^{l+1}(y_1) \\
\vdots & & \vdots \\
Z_{-l}^{l+1}(y_n) & \dots & Z_{h-1}^{l+1}(y_n)
\end{pmatrix}.
\end{equation}
Expression \eqref{ZB_repr_val} can be further reformulated as
\begin{align*}
s_{kl}(\mathbf{x},\mathbf{y}) &= \left(\mathbf{Z}_{k+1}(\mathbf{x})\ \mathbf{e}_{m}\right)\,\begin{pmatrix} \mathbf{Z} & \mathbf{v} \\ \mathbf{u}^{\top} & 0
\end{pmatrix}\,\begin{pmatrix} \mathbf{Z}_{l+1}^{\top}(\mathbf{y}) \\ \mathbf{e}_{n}^{\top}\end{pmatrix} =\mathbf{Z}_{k+1}(\mathbf{x})\,\mathbf{Z}\,\mathbf{Z}_{l+1}^{\top}(\mathbf{y}) + 
\mathbf{e}_{m}\,\mathbf{u}^{\top}\,\mathbf{Z}_{l+1}^{\top}(\mathbf{y}) + 
\mathbf{Z}_{k+1}(\mathbf{x})\,\mathbf{v}\,\mathbf{e}_{n}^{\top}.
\end{align*}
Using Kronecker tensor products, $s_{kl}(\mathbf{x},\mathbf{y})$ can be expressed in vectorized form as
\begin{align} \label{ZB_repr_val_2}
cs(s_{kl}\left(\mathbf{x},\mathbf{y})\right) &= \left(\mathbf{Z}_{l+1}(\mathbf{y})\,\otimes\,\mathbf{Z}_{k+1}(\mathbf{x})\right)\,cs(\mathbf{Z}) + \left(\mathbf{e}_{n}\,\otimes\,\mathbf{Z}_{k+1}(\mathbf{x})\right)\mathbf{v} + \left(\mathbf{Z}_{l+1}(\mathbf{y})\,\otimes\,\mathbf{e}_{m}\right)\mathbf{u}\nonumber \\ &=\mathbb{Z}\,cs(\mathbf{Z}) + \mathbb{Z}_{x}\,\mathbf{v} + \mathbb{Z}_{y}\,\mathbf{u},
\end{align}
where 
\begin{equation}\label{tensor_matrices}
\mathbb{Z} = \mathbf{Z}_{l+1}(\mathbf{y})\,\otimes\,\mathbf{Z}_{k+1}(\mathbf{x}), \quad \mathbb{Z}_{x} = \mathbf{e}_{n}\,\otimes\,\mathbf{Z}_{k+1}(\mathbf{x}),\quad \mathbb{Z}_{y} = \mathbf{Z}_{l+1}(\mathbf{y})\,\otimes\,\mathbf{e}_{m}.
\end{equation}
Spline $s_{kl}(x,y)\in{\cal Z}_{kl}^{\Delta\lambda,\Delta\mu}(\Omega)$ is determined uniquely by spline coefficients $cs(\mathbf{Z})$, $\mathbf{v}$, $\mathbf{u}$. Therefore, the aim is to express the smoothing functional as $J_{p,q}(cs(\mathbf{Z}),\mathbf{v},\mathbf{u})$, i.e. as~a~function of the unknown spline coefficients. To this end, the functional \eqref{compfcional} can be divided into two parts,
$$
J_1 = \sum_{i=1}^{m} \sum_{j=1}^{n} \left[f_{ij}-s_{kl}(x_i,y_j) \right]^2,\qquad
J_2 = \rho \iint\limits_{\Omega} \big[ s_{kl}^{(p,q)}(x,y) \big]^2\,\mbox{d}x\,\mbox{d}y.
$$
Using relation \eqref{ZB_repr_val_2}, $J_{1}$ can be formulated with $\mathbf{F} = (f_{ij})_{i=1,\dots,m; j=1,\dots, n}$ as
\begin{align}\label{J1}
J_1 &=\left[cs(\mathbf{F}) - \mathbb{Z}\, cs(\mathbf{Z}) - \mathbb{Z}_{x}\,\mathbf{v} - \mathbb{Z}_{y}\, \mathbf{u}\right]^{\top} \left[cs(\mathbf{F}) - \mathbb{Z}\, cs(\mathbf{Z}) - \mathbb{Z}_{x}\,\mathbf{v} - \mathbb{Z}_{y}\, \mathbf{u}\right]  \nonumber\\
&= \left[ cs(\mathbf{F})  - \begin{pmatrix}\mathbb{Z} & \mathbb{Z}_{x} & \mathbb{Z}_{y}
\end{pmatrix}\begin{pmatrix} cs(\mathbf{Z}) \\ \mathbf{v} \\ \mathbf{u}
\end{pmatrix}\right]^{\top} \left[ cs(\mathbf{F})  - \begin{pmatrix}\mathbb{Z} & \mathbb{Z}_{x} & \mathbb{Z}_{y}
\end{pmatrix}\begin{pmatrix} cs(\mathbf{Z}) \\ \mathbf{v} \\ \mathbf{u}
\end{pmatrix}\right] \nonumber\\
&= \begin{pmatrix} 
cs(\mathbf{Z})^{\top} & \mathbf{v}^{\top} & \mathbf{u}^{\top}
\end{pmatrix}
\begingroup
\setlength\arraycolsep{0.2cm}
\begin{pmatrix}
\mathbb{Z}^{\top}\mathbb{Z} & \mathbb{Z}^{\top}\mathbb{Z}_{x} & \mathbb{Z}^{\top}\mathbb{Z}_{y} \\[0.1cm]
\mathbb{Z}_{x}^{\top}\mathbb{Z} & \mathbb{Z}_{x}^{\top}\mathbb{Z}_{x} & \mathbb{Z}_{x}^{\top}\mathbb{Z}_{y} \\[0.1cm]
\mathbb{Z}_{y}^{\top}\mathbb{Z} & \mathbb{Z}_{y}^{\top}\mathbb{Z}_{x} & \mathbb{Z}_{y}^{\top}\mathbb{Z}_{y}
\end{pmatrix}
\endgroup
\begin{pmatrix} cs(\mathbf{Z}) \\[0.1cm] \mathbf{v} \\[0.1cm] \mathbf{u}
\end{pmatrix} \nonumber - 2\,cs(\mathbf{F})^{\top}\,\begin{pmatrix}\mathbb{Z} & \mathbb{Z}_{x} & \mathbb{Z}_{y}
\end{pmatrix}\begin{pmatrix} cs(\mathbf{Z}) \\ \mathbf{v} \\ \mathbf{u}
\end{pmatrix} + cs(\mathbf{F})^{\top}\,cs(\mathbf{F}).\nonumber
\end{align}
To express the derivative $s_{kl}^{(p,q)}(x,y)$ and the functional $J_{2}$, we first introduce the following notation. Let
\begin{align*}
\mathbf{S}_{\lambda}^{p} &= \left(\mathbf{D}_{\lambda}^p \mathbf{L}_{\lambda}^p \cdots \mathbf{D}_{\lambda}^1 \mathbf{L}_{\lambda}^1\right) \in \mathbb{R}^{g+k+1-p,g+k+1},\\
\mathbf{D}_{\lambda}^{j} &= (k+1-j)[\mbox{diag}(\mathbf{d}_{\lambda}^{j})]^{-1},\quad j=1,\dots,p,
\end{align*}
where $\mathbf{d}_{\lambda}^{j}=(\lambda_1 - \lambda_{-k+j},\ \dots\ , \lambda_{g+k+1-j} - \lambda_{g})$ and $\mathbf{L}_{\lambda}^j=(l_{\alpha\beta}^{\lambda})\in\mathbb{R}^{g+k+1-j,g+k+2-j}$ be a matrix with elements $l_{\alpha\beta}^{\lambda}$ defined as
$$
l_{\alpha\beta}^{\lambda}=\begin{cases}
-1 \quad \textnormal{if}\ \alpha=\beta, \\
1 \quad \textnormal{if}\ \alpha=\beta - 1,\quad \alpha=1,\dots,g+k+1-j,\ \beta=1,\dots,g+k+2-j \\
0 \quad \textnormal{otherwise}.
\end{cases}
$$
Similarly, let $\mathbf{S}_{\mu}^{q} = \left(\mathbf{D}_{\mu}^q \mathbf{L}_{\mu}^q \cdots \mathbf{D}_{\mu}^1 \mathbf{L}_{\mu}^1\right) \in \mathbb{R}^{h+l+1-q,h+l+1}$, $\mathbf{D}_{\mu}^{j}=(l+1-j)[\mbox{diag}(\mathbf{d}_{\mu}^{j})]^{-1},\ j=1,\dots,q$ with $\mathbf{d}_{\mu}^{j}=(\mu_1 - \mu_{-l+j},\ \dots\ , \mu_{h+l+1-j} - \mu_{h})$
and $\mathbf{L}_{\mu}^{j}=(l_{\alpha\beta}^{\mu})\in\mathbb{R}^{h+l+1-j,h+l+2-j}$ for $\alpha=1,\dots,h+l+1-j$ and $\beta=1,\dots,h+l+2-j$ defined as
$$
l_{\alpha\beta}^{\mu}=\begin{cases}
-1 \quad \textnormal{if}\ \alpha=\beta, \\
1 \quad \textnormal{if}\ \alpha=\beta-1, \\
0 \quad \textnormal{otherwise}.
\end{cases}
$$
Using Theorem \ref{biv_mn}, relation \eqref{VBv} and the fact that $p$-th derivative of a $B$-spline $B_{i}^{k+1}(x)$ is a $B$-spline $B_{i}^{k+1-p}(x)$ (and analogously for $B_{j}^{l+1}(y)$) as further derived in \cite{hron22, Dierckx, machalova16, machalova2002}, derivative $s_{kl}(p,q)(x,y)$ can be formulated as

\begin{align*}
s_{kl}^{(p,q)}(x,y) &= \frac{\partial^p}{\partial x^p}\frac{\partial^q}{\partial y^q} \mathbf{B}_{k+1}^{\top}(x)\,\mathbf{B}\,\mathbf{B}_{l+1}(y) \\
&= \frac{\partial^p}{\partial x^p}\frac{\partial^q}{\partial y^q} \big[\mathbf{B}_{k+1}^{\top}(x)\,\mathbf{D}_{\lambda}\,\mathbf{K}_{g k}^{\top}\,\mathbf{Z}\,\mathbf{K}_{h l}\mathbf{D}_{\mu}\mathbf{B}_{l+1}(y) + \mathbf{B}_{k+1}^{\top}(x)\,\mathbf{D}_{\lambda}\,\mathbf{K}_{g k}^{\top}\,\mathbf{v} + \mathbf{u}^{\top}\,\mathbf{K}_{h l}\,\mathbf{D}_{\mu}\,\mathbf{B}_{l+1}(y) \big] \\
&= \mathbf{B}_{k+1-p}^{\top}(x)\,\mathbf{S}_{\lambda}^{p}\,\mathbf{D}_{\lambda}\,\mathbf{K}_{g k}^{\top}\,\mathbf{Z}\,\mathbf{K}_{h l}\mathbf{D}_{\mu}\left(\mathbf{S}_{\mu}^{q}\right)^{\top}\,\mathbf{B}_{l+1-q}(y) \\
&= \left[\left(\mathbf{B}_{l+1-q}^{\top}(y)\,\mathbf{S}_{\mu}^{q}\,\mathbf{D}_{\mu}\,\mathbf{K}_{h l}^{\top}\big) \otimes (\mathbf{B}_{k+1-p}^{\top}(x)\,\mathbf{S}_{\lambda}^{p}\,\mathbf{D}_{\lambda}\,\mathbf{K}_{g k}^{\top}\right)\right]\, cs(\mathbf{Z}) \\
&= \left(\mathbf{B}_{l+1-q}(y) \otimes \mathbf{B}_{k+1-p}(x)\right)^{\top}
\left(\mathbf{S}_{\mu}^{q} \otimes \mathbf{S}_{\lambda}^{p}\right)
\left(\mathbf{D}_{\mu} \otimes \mathbf{D}_{\lambda}\right)
\left(\mathbf{K}_{h l} \otimes \mathbf{K}_{g k}\right)^{\top}\, cs(\mathbf{Z})\\
&= \left(\mathbf{B}_{l+1-q}(y) \otimes \mathbf{B}_{k+1-p}(x)\right)^{\top}\,\mathbb{S}\,\mathbb{D}\,\mathbb{K}^{\top}\, cs(\mathbf{Z}),
\end{align*}
where
\begin{equation}\label{matrix_SDK}
\mathbb{S} \, = \,  \mathbf{S}_{\mu}^{q}\otimes\mathbf{S}_{\lambda}^{p}, \quad
\mathbb{D} \, = \,  \mathbf{D}_{\mu}\otimes\mathbf{D}_{\lambda}, \quad
\mathbb{K} \, = \,  \mathbf{K}_{hl}\otimes\mathbf{K}_{gk}.
\end{equation}
See \citep{hron22} for more details. Then the functional $J_{2}$ can be expressed as
\begin{align*}
J_{2} &= \rho\,\iint\limits_{\Omega}\left[s_{kl}^{(p,q)}(x,y)\right]^{2}\mbox{d}x\, \mbox{d}y 
= \rho\,\iint\limits_{\Omega} \left[\left(\mathbf{B}_{l+1-q}(y) \otimes \mathbf{B}_{k+1-p}(x)\right)^{\top}\,\mathbb{S}\,\mathbb{D}\,\mathbb{K}^{\top}\, cs(\mathbf{Z})\right]^{2}\,\mbox{d}x\,\mbox{d}y \\
&= \rho\,cs(\mathbf{Z})^{\top} \mathbb{K}\,\mathbb{D}\,\mathbb{S}^{\top}\,\iint\limits_{\Omega}
\left[\mathbf{B}_{l+1-q}(y) \otimes \mathbf{B}_{k+1-p}(x)\right]\left[\mathbf{B}_{l+1-q}(y) \otimes \mathbf{B}_{k+1-p}(x)\right]^{\top}\,\mbox{d}x\,\mbox{d}y\,\mathbb{S}\,\mathbb{D}\,\mathbb{K}^{\top}\,cs(\mathbf{Z}) \\ 
&= \rho\,cs(\mathbf{Z})^{\top} \mathbb{K}\,\mathbb{D}\,\mathbb{S}^{\top}\left(\int\limits_{c}^{d}
\mathbf{B}_{l+1-q}(y)\mathbf{B}_{l+1-q}^{\top}(y)\mbox{d}y\right)\otimes\left(\int\limits_{a}^{b}\mathbf{B}_{k+1-p}(x) \mathbf{B}_{k+1-p}^{\top}(x)\mbox{d}x\right)\,\mathbb{S}\,\mathbb{D}\,\mathbb{K}^{\top}cs(\mathbf{Z}) \\
&= \rho\,cs(\mathbf{Z})^{\top} \mathbb{K}\,\mathbb{D}\,\mathbb{S}^{\top}\,\left(\mathbf{M}_{l,q}^{y} \otimes \mathbf{M}_{k,p}^{x}\right) \mathbb{S}\, \mathbb{D}\, \mathbb{K}^{\top}\, cs(\mathbf{Z}) = \rho\,cs(\mathbf{Z})^{\top}
\mathbb{K}\,\mathbb{D}\,\mathbb{S}^{\top}\,\mathbb{M}\, \mathbb{S}\, \mathbb{D}\, \mathbb{K}^{\top}\, cs(\mathbf{Z}) \\
&= \rho\,cs(\mathbf{Z})^{\top}\mathbb{N}\, cs(\mathbf{Z}),
\end{align*}
where 
\begin{equation}\label{matrix_M}
\mathbb{M} \, = \, \mathbf{M}_{l,q}^{y} \otimes \mathbf{M}_{k,p}^{x},
\qquad \mathbb{N} \, = \, \mathbb{K}\,\mathbb{D}\,\mathbb{S}^{\top}\,\mathbb{M}\, \mathbb{S}\, \mathbb{D}\, \mathbb{K}^{\top},
\end{equation}
with 
$\mathbf{M}_{k,p}^{x} = \left(m_{ij}^x \right)_{i,j=-k+p}^g$,
$\mathbf{M}_{l,q}^{y}= \left(m_{ij}^y \right)_{i,j=-l+q}^h$,
$$
m_{ij}^x=\int\limits_a^b B_i^{k+1-p}(x)B_j^{k+1-p}(x) \, \mbox{d}x, \qquad
m_{ij}^y=\int\limits_a^b B_i^{l+1-q}(y)B_j^{l+1-q}(y) \, \mbox{d}y.
$$
Moreover, $J_2$ can be expressed as a function of the coefficients $cs(\mathbf{Z}),\,\mathbf{v},\,\mathbf{u}$ as
$$
J_ 2 = \begin{pmatrix} cs(\mathbf{Z})^{\top} & \mathbf{v}^{\top} & \mathbf{u}^{\top} \end{pmatrix}
\begingroup
\setlength\arraycolsep{0.2cm}
\begin{pmatrix} \rho\,\mathbb{N}&0&0\\0&0&0\\0&0&0 \end{pmatrix}
\endgroup
\begin{pmatrix}
cs(\mathbf{Z} \\ \mathbf{v} \\ \mathbf{u} \end{pmatrix}.
$$

\noindent
Note that from construction respective marginals are not penalized. Additional penalization can be imposed, if desired, by adding suitable separate penalty for smoothness. Altogether, functional $J_{pq}(s_{kl})$ can be formulated as a quadratic function
\begin{equation}\label{compfunction}
J_{pq}(\mathbf{w}) = \mathbf{w}^{\top}\,\mathbf{G}(\rho)\,\mathbf{w} - 2\,\mathbf{w}^{\top}\,\mathbf{g} + c,
\end{equation}
where $c = cs(\mathbf{F})^{\top}cs(\mathbf{F})$,
\begin{align}\label{compfunction_vectors}
\mathbf{w} = \begin{pmatrix}
cs(\mathbf{Z}) \\ 
\mathbf{v} \\ 
\mathbf{u}
\end{pmatrix},\qquad
\mathbf{g} = \begin{pmatrix}
\mathbb{Z}^{\top} cs(\mathbf{F}) \\[0.2cm]
\mathbb{Z}_{x}^{\top} cs(\mathbf{F}) \\[0.2cm]
\mathbb{Z}_{y}^{\top} cs(\mathbf{F})
\end{pmatrix}
\end{align}
and
\begin{align}\label{compfunction_matrix}
\mathbf{G}(\rho) & = 
\begin{pmatrix} \mathbb{Z}^{\top} \\ \mathbb{Z}_{x}^{\top} \\ \mathbb{Z}_{y}^{\top}\end{pmatrix}\,
\begin{pmatrix}
\mathbb{Z} & \mathbb{Z}_{x} & \mathbb{Z}_{y}
\end{pmatrix} \, + \, 
\begingroup
\setlength\arraycolsep{0.2cm}
\begin{pmatrix} \rho\,\mathbb{N}&0&0\\0&0&0\\0&0&0 \end{pmatrix}
\endgroup \, =
\begin{pmatrix}
\rho\,\mathbb{N} + \mathbb{Z}^{\top}\mathbb{Z} & \mathbb{Z}^{\top}\mathbb{Z}_{x} & \mathbb{Z}^{\top}\mathbb{Z}_{y} \\[0.2cm]
\mathbb{Z}_{x}^{\top}\mathbb{Z} & \mathbb{Z}_{x}^{\top}\mathbb{Z}_{x} & \mathbb{Z}_{x}^{\top}\mathbb{Z}_{y} \\[0.2cm]
\mathbb{Z}_{y}^{\top}\mathbb{Z} & \mathbb{Z}_{y}^{\top}\mathbb{Z}_{x} & \mathbb{Z}_{y}^{\top}\mathbb{Z}_{y}
 \end{pmatrix}.
\end{align}

\noindent
To determine the coefficients of the smoothing spline, it is necessary to find the minimum of the quadratic function \eqref{compfunction}. There exists just one minimum of $J_{pq}$ if and only if the matrix $\mathbf{G}(\rho)$ is positive definite. 
Since the matrix $\mathbb{N}$ is positive semi-definite,
the matrix $\mathbf{G}(\rho)$ is positive definite if and only if 
matrix 
$\begin{pmatrix}
\mathbb{Z} & \mathbb{Z}_{x} & \mathbb{Z}_{y}
\end{pmatrix}
$
is of full column rank. With respect to the notation (\ref{tensor_matrices}) one can easily see that this is fulfilled if and only if matrices $\mathbf{Z}_{k+1}(\mathbf{x})$ and $\mathbf{Z}_{l+1}(\mathbf{y})$, defined in (\ref{Z_matrices}), are of full column rank.
For this fulfilment it is sufficient to ensure that at least one point, from given $x_1,\ldots,x_m$, lies in the support of each $Z\!B$-spline $Z_{i}^{k+1}(x)$, $i=-k,\ldots,g$. The same must hold for $y_1,\ldots,y_n$ 
and $Z\!B$-splines $Z_{j}^{l+1}(y)$, $j=-l,\ldots,h$, which corresponds to the Schoenberg-Whitney condition, see \citep{deboor, machalova2002}. The necessary condition for the minimum of a quadratic function,
$$
\frac{\partial J_{pq}(\mathbf{w})}{\partial\mathbf{w}^{\top}}=\mathbf{0},
$$
forms a system of linear equations $2\,\mathbf{G}(\rho)\,\mathbf{w} - 2\,\mathbf{g} = \mathbf{0}$. Consequently, coefficients of the smoothing spline $s_{kl}(x,y)$ are obtained as solution of the system of linear equations
$$
\mathbf{G}(\rho)\,\mathbf{w} = \mathbf{g},
$$
with matrix $\mathbf{G}(\rho)$ and vectors $\mathbf{g},\, \mathbf{w}$ defined in \eqref{compfunction_vectors}, \eqref{compfunction_matrix}. As a result, the following theorem about existence and uniqueness smoothing splines can be formulated.

\begin{theorem} \label{biv_smoothing_th}
Let the data matrix $\mathbf{F}=\lbrace f_{ij}\rbrace_{i=1,j=1}^{m,n}$ representing measurements belonging to $L^2_0(\Omega)$ at observation points  $(x_i,y_j)~\in~\Omega,\ i=1,\dots,m,\ j=1,\dots,n$   be given. Then there exists a unique smoothing spline $s^{*}_{kl}(x,y)\in{\cal Z}_{kl}^{\Delta\lambda,\Delta\mu}(\Omega)$ minimizing  functional $J_{p,q}(s_{kl})$ 
if and only if matrices
$\mathbf{Z}_{k+1}(\mathbf{x})$ and $\mathbf{Z}_{l+1}(\mathbf{y})$,
defined in \eqref{Z_matrices}, are of full column rank.
The smoothing spline
$$
s_{kl}^{*}(x,y) = \begin{pmatrix} \mathbf{Z}_{k+1}(x) & 1 \end{pmatrix} \begin{pmatrix} \mathbf{Z}^{*} & \mathbf{v}^{*} \\ \mathbf{u}^{*\top} & 0 \end{pmatrix}
\begin{pmatrix} \mathbf{Z}_{l+1}^{\top}(y) \\ 1 \end{pmatrix}
$$
is determined by the spline coefficients $\mathbf{Z}^{*}$, $\mathbf{v}^{*}$ and $\mathbf{u}^{*}$ obtained as the unique solution of the system of linear equations
\begin{align}\label{eq_system}
\begingroup
\setlength\arraycolsep{0.2cm}
\begin{pmatrix}
\rho\,\mathbb{N} + \mathbb{Z}^{\top}\mathbb{Z} & \mathbb{Z}^{\top}\mathbb{Z}_{x} & \mathbb{Z}^{\top}\mathbb{Z}_{y} \\[0.2cm]
\mathbb{Z}_{x}^{\top}\mathbb{Z} & \mathbb{Z}_{x}^{\top}\mathbb{Z}_{x} & \mathbb{Z}_{x}^{\top}\mathbb{Z}_{y} \\[0.2cm]
\mathbb{Z}_{y}^{\top}\mathbb{Z} & \mathbb{Z}_{y}^{\top}\mathbb{Z}_{x} & \mathbb{Z}_{y}^{\top}\mathbb{Z}_{y}
\end{pmatrix}
\endgroup
\begin{pmatrix}
cs(\mathbf{Z}) \\[0.2cm] \mathbf{v} \\[0.2cm] \mathbf{u}
\end{pmatrix} =
\begin{pmatrix}
\mathbb{Z}^{\top} cs(\mathbf{F}) \\[0.2cm]
\mathbb{Z}_{x}^{\top} cs(\mathbf{F}) \\[0.2cm]
\mathbb{Z}_{y}^{\top} cs(\mathbf{F})
\end{pmatrix},
\end{align}
with matrices $\mathbb{Z}$, $\mathbb{Z}_{x}$, $\mathbb{Z}_{y}$ defined in \eqref{tensor_matrices} and matrices $\mathbb{S},\ \mathbb{D},\ \mathbb{K},\ \mathbb{M}$ and $\mathbb{N}$ defined in \eqref{matrix_SDK}, \eqref{matrix_M}.
\end{theorem}

The resulting smoothing spline is affected by the choice of degrees $k$, $l$ and the degrees of the respective derivatives \mbox{$p$, $q$}. Moreover, the quality of the resulting smoothing spline strongly depends on the choice of the number of knots and their positions, which should reflect the main features of the data. With fixed parameters, the optimal value of smoothing parameter $\rho$ can then be determined using generalized cross-validation (GCV) \cite{ramsay05}. Note that degrees of spline are typically selected by the user, to determine the number and position of the knots, advanced knots insertion algorithms can be employed, see \cite{Dierckx}.

The solution of the system of linear equations \eqref{eq_system} can be formally expressed as
\begin{equation*}
\begin{pmatrix}
cs(\mathbf{Z}) \\ \mathbf{v} \\ \mathbf{u}
\end{pmatrix} =
\mathbf{G}(\rho)^{-1}\,\overline{\mathbb{Z}}\,cs(\mathbf{F}),
\end{equation*}
where $\mathbf{G}(\rho)$ is defined in \eqref{compfunction_matrix} and $\overline{\mathbb{Z}} = \begin{pmatrix} \mathbb{Z} & \mathbb{Z}_{x} & \mathbb{Z}_{y}
\end{pmatrix}^{\top}$. Subsequently, function values of the smoothing spline at given points can be expressed using \eqref{ZB_repr_3} as
\begin{equation*}    
s_{kl}^{*}(\mathbf{x},\mathbf{y}) =\left(\overline{\mathbf{Z}}_{l+1}^{\top}(\mathbf{y})\otimes\overline{\mathbf{Z}}^{\top}_{k+1}(\mathbf{x})\right)\,\mathbf{P}\,\mathbf{G}(\rho)^{-1}\,\overline{\mathbb{Z}}\,cs(\mathbf{F}) = \mathbf{H}(\rho)\,cs(\mathbf{F}).
\end{equation*}
Here
$$
\mathbf{H}(\rho) = \left(\overline{\mathbf{Z}}_{l+1}^{\top}(\mathbf{y})\otimes\overline{\mathbf{Z}}_{k+1}^{\top}(\mathbf{x})\right)\,\mathbf{P}\,\mathbf{G}(\rho)^{-1}\,\overline{\mathbb{Z}}
$$
represents the projection matrix (called \textit{hat matrix}), with which we can define the GCV criterion for choosing $\rho$ as
\begin{equation}\label{GCV}
\mbox{GCV}(\rho)=\frac{\frac{1}{n\,m}\sum_{i=1}^{n}\sum_{j=1}^{m}\left(f_{ij}-s_{kl}^{*}(x_{i},y_{j})\right)^{2}}{\left(1-\mbox{trace}(\mathbf{H}(\rho))/n\,m\right)^{2}}.
\end{equation}

\section{Compositional splines in ${\cal B}^{2}(\Omega)$} \label{backtoBayes}
The $Z\!B$-spline representation enables to analyze density functions using popular methods embedded in the $L^{2}(\Omega)$ space. However, from a theoretical perspective, it is crucial to express spline functions directly in the original space ${\cal B}^{2}(\Omega)$. This can be achieved using the inverse clr transformation defined in \eqref{iclr}. For given \mbox{$\Omega = [a,b]\times[c,d]\subset\mathbb{R}^{2}$}, extended sequences of knots \eqref{exknotsX}, \eqref{exknotsY} and $k,\,l\in\mathbb{N}_{0}$, we denote by ${\cal C}_{kl}^{\Delta\lambda,\Delta\mu}(\Omega) \subset {\cal B}^{2}(\Omega)$ the vector space of \textit{compositional splines} $\xi_{kl}(x,y)$ defined on $\Omega$ of degree $k$ in $x$ with knots $\Delta\lambda$ and of degree $l$ in $y$ with knots $\Delta\mu$. Since the clr transformation is an isomorphism between ${\cal B}^{2}(\Omega)$ and $L^{2}_{0}(\Omega)$ (i.e. also between ${\cal C}_{kl}^{\Delta\lambda,\Delta\mu}(\Omega)$ and ${\cal Z}_{kl}^{\Delta\lambda,\Delta\mu}(\Omega)$) it follows that
$$
\dim\left({\cal C}_{kl}^{\Delta\lambda,\Delta\mu}(\Omega)\right) = (g+k+1)(h+l+1)-1.
$$
Moreover, it allows to define \textit{compositional $B$-splines ($C\!B$-splines)} in ${\cal B}^{2}(\Omega)$ as
\begin{align*}
\zeta_{ij}^{k+1,l+1}(x,y) &=_{{\cal B}^{2}(\Omega)} \exp\left[Z_{ij}^{k+1,l+1}(x,y)\right],\\
\zeta_{i}^{k+1}(x,y) &=_{{\cal B}^{2}(\Omega)} \exp\left[Z_{i}^{k+1}(x,y)\right],\\
\zeta_{j}^{l+1}(x,y) &=_{{\cal B}^{2}(\Omega)} \exp\left[Z_{j}^{l+1}(x,y)\right],
\end{align*}
for $i=-k,\dots,g-1,\ j=-l,\dots,h-1$. Examples of $C\!B$-splines are pictured in Figure \ref{ZB_and_CB_splines_1}.

\begin{figure}[h!]
    \centering
    \includegraphics[width=0.35\textwidth]{ZB-spline.png}
    \includegraphics[width=0.35\textwidth]{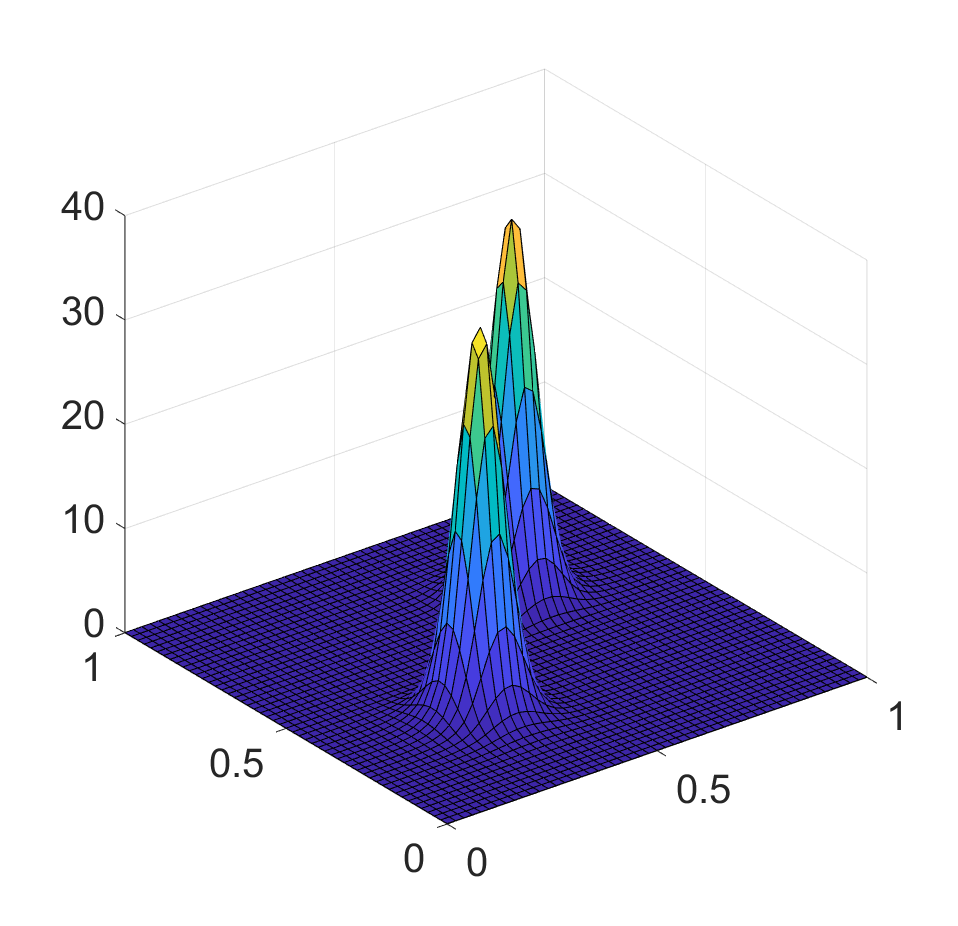} \\
    \includegraphics[width=0.35\textwidth]{ZB-spline_const_2.png}
    \includegraphics[width=0.35\textwidth]{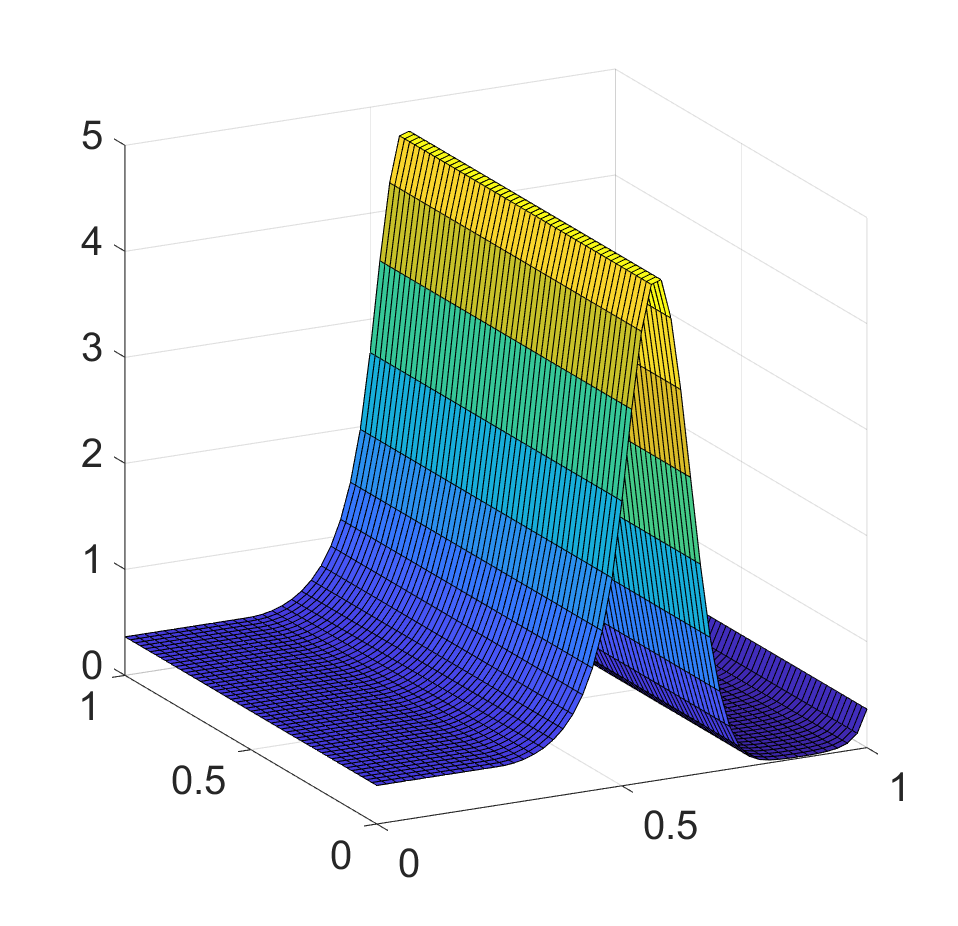}
    \caption{An example of a $Z\!B$-spline $Z_{ij}^{k+1,l+1}(x,y)$ (top, left) and corresponding $C\!B$-splines $\zeta_{ij}^{k+1,l+1}(x,y)$ (top, right) and an example of a $Z\!B$-spline $Z_{i}^{k+1}(x,y)$ (bottom, left) and corresponding $C\!B$-spline $\zeta_{i}^{k+1}(x,y)$ (bottom, right)}
    \label{ZB_and_CB_splines_1}
\end{figure}

Note that $C\!B$-splines meet the positivity and unit integral constraints, as required for PDFs. Finally, isometric properties of the clr transformation ensure that every compositional spline $\xi_{kl}(x,y)\in{\cal C}_{kl}^{\Delta\lambda,\Delta\mu}(\Omega)$ can be expressed uniquely as
$$
\xi_{kl}(x,y) = \bigoplus_{i=-k}^{g-1}\bigoplus_{j=-l}^{h-1} z_{ij}\,\odot\,\zeta_{ij}^{k+1,l+1}(x,y)\,\oplus\,\bigoplus_{i=-k}^{g-1}v_i\,\odot\,\zeta_{i}^{k+1}(x,y)\,\oplus\,\bigoplus_{j=-l}^{h-1}u_j\,\odot\,\zeta_{j}^{l+1}(x,y).
$$
In addition, isometric properties allow to decompose the spline $\xi_{kl}(x,y)$ orthogonally into the interactive part
$$
\xi_{kl}^{int}(x,y) = \bigoplus_{i=-k}^{g-1}\bigoplus_{j=-l}^{h-1} z_{ij}\,\odot\,\zeta_{ij}^{k+1,l+1}(x,y)
$$
and the independent part
$$
\xi_{kl}^{ind}(x,y) = \left(\bigoplus_{i=-k}^{g-1}v_i \odot\,\zeta_{i}^{k+1}(x,y)\right)\,\oplus\,\left(\bigoplus_{j=-l}^{h-1}u_j\,\odot\,\zeta_{j}^{l+1}(x,y)\right).
$$
A compositional spline represents an approximation of a density function in ${\cal B}^2(\Omega)$. Its representation and decomposition directly in ${\cal B}^{2}(\Omega)$ is a useful and strong theoretical result in terms of Bayes space methodology. However, in applications it is often more convenient to work with  the  corresponding and equivalent $Z\!B$-spline representation in $L_0^{2}(\Omega)$, as demonstrated in the following sections.

\section{Simulation study} \label{simul}
In this section, we will examine the properties of the proposed representation on simulated data, generated from a generalized (bivariate) beta distribution, with an emphasis on the quality of the spline approximation of the true bivariate beta density.


\subsection{Data simulation} \label{data:simu}
Consider the bivariate beta distribution with density
$$
f(x,y)\,=\,\frac{1}{B(\alpha_{0},\alpha_{1},\alpha_{2})} \frac{x^{\alpha_{1}-1} \, (1-x)^{\alpha_{0}+\alpha_{2}-1} \, y^{\alpha_{2}-1} \, (1-y)^{\alpha_{0}+\alpha_{1}-1}}{(1-xy)^{\alpha_{0}+\alpha_{1}+\alpha_{2}}},
$$
where $B(\alpha_{0},\alpha_{1},\alpha_{2}) = \Gamma(\alpha_{0})\Gamma(\alpha_{1})\Gamma(\alpha_{2})/\Gamma(\alpha_{0}+\alpha_{1}+\alpha_{2})$ is the generalized beta function, $\alpha_{0},\,\alpha_{1},\,\alpha_{2}\in(0,\infty)$ are parameters of the distribution and $x,\,y\in(0,1)$, hence, $\Omega=[0,1]\times[0,1]$ is the considered domain. Draws from the distribution were generated using the Monte Carlo accept-reject generation algorithm. Specifically, density $f(x,y)$ is the target probability density function, and let $g(x,y)$ be the density of  the bivariate uniform distribution on $\Omega$, which is used as a proposal probability density function. Finally, the condition
$$
f(x,y)\,\leq\,M g(x,y),\quad x,\,y\in(0,1),
$$
holds for some constant $M>0$. The accept-reject algorithm can be summarized in the following steps:
\begin{enumerate}
    \item Independently generate  values $X,\,Y$ and $U$ from a univariate uniform distribution on the interval $(0,1)$.
    \item If $U\leq\frac{f(X,Y)}{M\,g(X,Y)}$, then accept $(X,Y)$, otherwise repeat  step 1.
    \item $(X,Y)$ then has probability density function $f(x,y)$.
\end{enumerate}
The parameters of the beta distribution were set to $\alpha_{0}=\alpha_{1}=\alpha_{2}=3$ (note that the limiting case $\alpha_0+\alpha_1+\alpha_2 \rightarrow 0$ corresponds to the independence case), see Figure \ref{fig:beta_simul} (left). The maximum of $f(x,y)$ can be found numerically and it can be easily verified that its value is approximately $4.097$, which determines the lower bound for choosing the constant $M$ since $g(x,y)$ is constant at level 1. Therefore, the Monte Carlo rejection parameter was set to $M=4.1$. The algorithm was terminated after we reached 3000 observations with density $f(x,y)$. Subsequently, the values were aggregated in a histogram (Figure \ref{fig:beta_simul}, right), from which we estimated the probability density function using compositional splines, since the histogram approximates the PDF up to proportionality. Since there is no bivariate counterpart of Sturges' rule, the number of histogram classes (also referred to as histogram bins) was set heuristically to $m=n=10$ equidistant classes each in the direction of both axes - this choice will be further examined in the following subsection.

\begin{figure}[h!]
    \centering
    \includegraphics[width=0.3\textwidth]{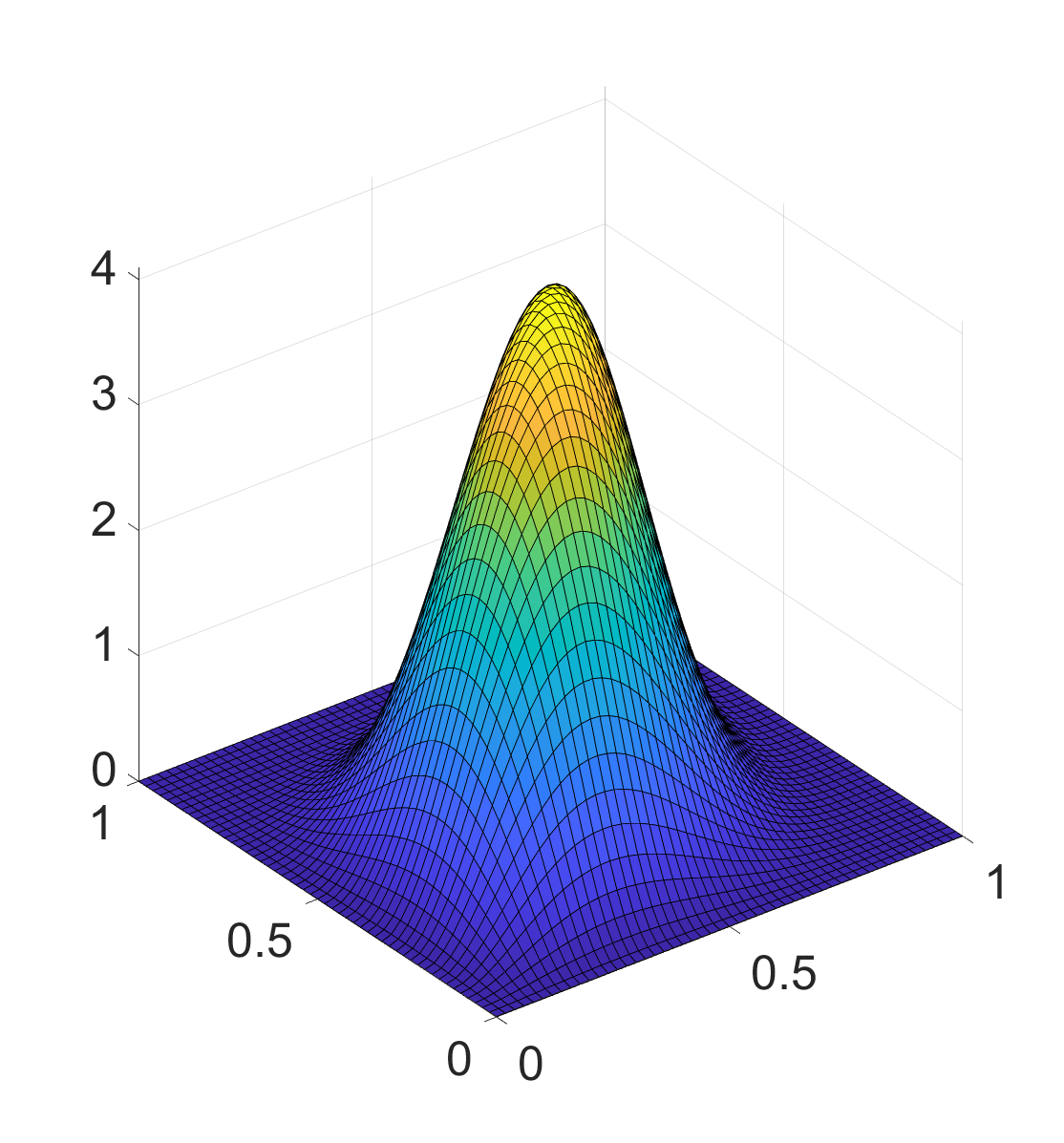}
    \includegraphics[width=0.3\textwidth]{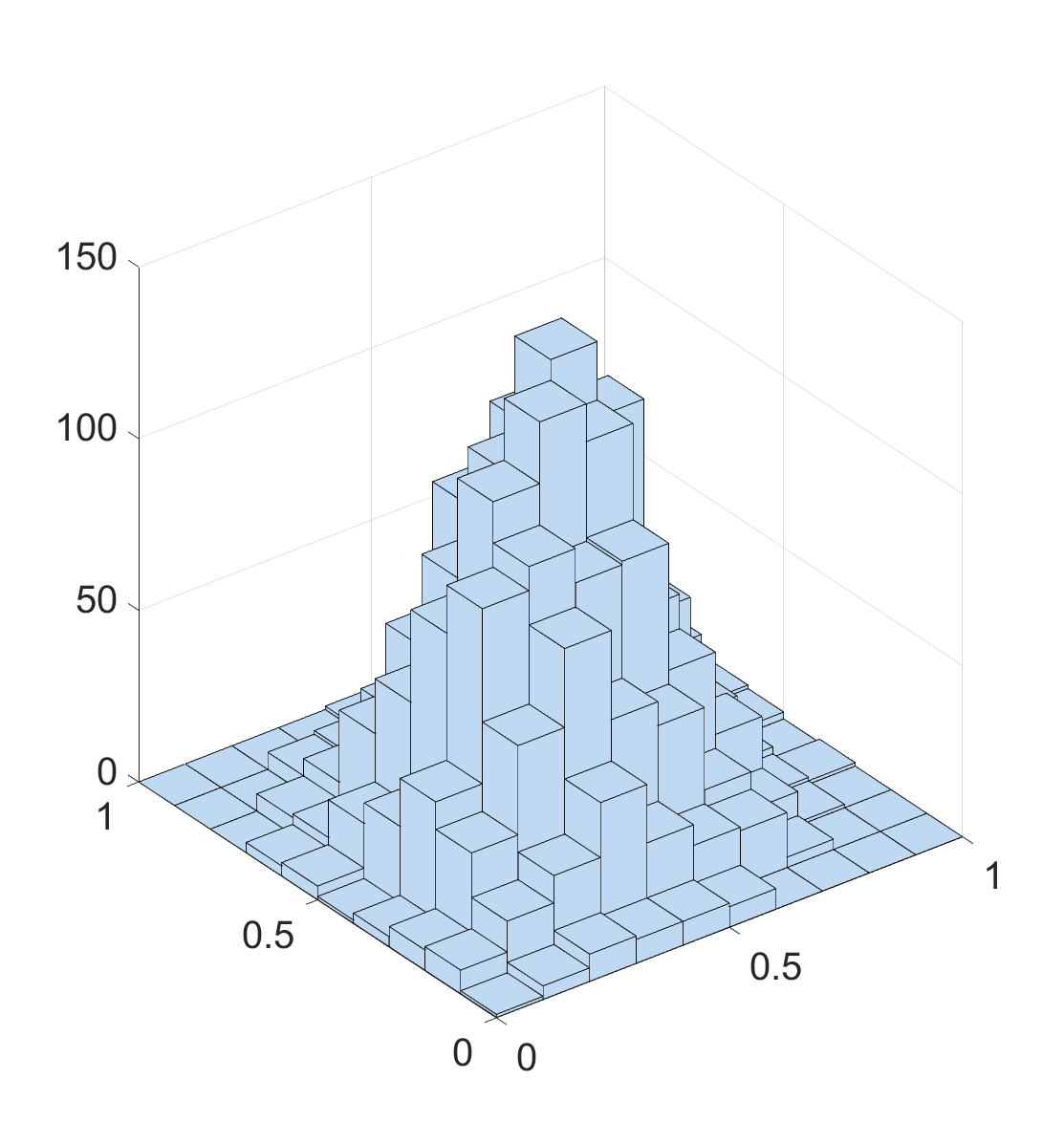}
    \caption{Theoretical density of a bivariate beta distribution with parameters $\alpha_{0}=\alpha_{1}=\alpha_{2}=3$ (left) and the  histogram of $3000$ generated values from this beta distribution 
    (right)}
    \label{fig:beta_simul}
\end{figure}

Let $\mathbf{x} = \mathbf{y} = (0.05,0.15, 0.25,0.35,0.45,0.55,0.65,0.75,0.85, 0.95)$ be the centers of each class in the direction of variables $x$ and $y$ and let $\mathbf{F} = (f_{ij})_{i=1,j=1}^{10,10}$ be the corresponding matrix of clr transformed histogram frequencies obtained using the discrete bivariate clr transformation defined in \citep{hron22}. Then $(\mathbf{x},\,\mathbf{y},\,\mathbf{F})$ represents the final preprocessed data. We repeated the Monte Carlo simulation 100 times and employed tensor product smoothing splines according to 
Theorem~\ref{biv_smoothing_th} with the following settings. Given the shape of the beta distribution, the degree of spline was set to $k=l=2$ in the direction of both variables, with the degree of the respective derivatives $p=q=1$. The knots were placed equidistantly to $\Delta\lambda = \Delta\mu = (0,\ 0.25,\ 0.5,\ 0.75,\ 1)$ and, finally, the optimal value of the smoothing parameter $\rho$ was determined by GCV \eqref{GCV}, calculated as the mean criterion from all simulated densities, resulting in $\rho=0.001$ (see Figure~\ref{fig:GCVerrors}).
 
\begin{figure}[h]
  \centering
\includegraphics[width=0.7\textwidth]{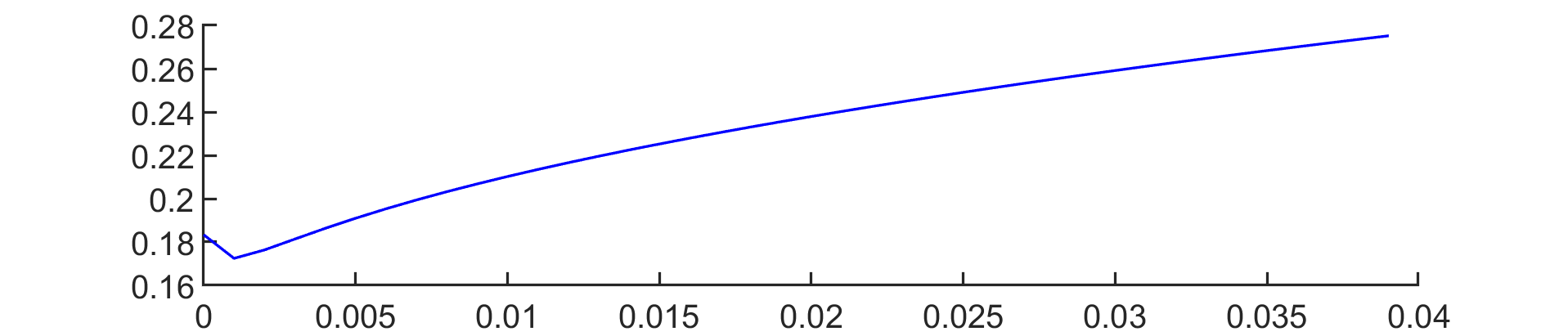}
  \caption{The mean GCV criterion for the choice of the optimal smoothing parameter $\rho$ for fitting the clr transformed histogram values by a tensor product smoothing spline.}
  \label{fig:GCVerrors}
\end{figure}

Selecting one of the random generations, the resulting $Z\!B$-spline approximation is depicted in Figure \ref{simulation_approximation}. Using the inverse clr transformation \eqref{iclr}, the  approximation in the clr-space can be converted back to the compositional spline in the original ${\cal B}^2(\Omega)$ space, where it forms the final estimate of the bivariate beta distribution. The result is visualized in Figure \ref{simulation_estimate}.
 
\begin{figure}[h!]
  \centering
  \includegraphics[width=0.35\textwidth]{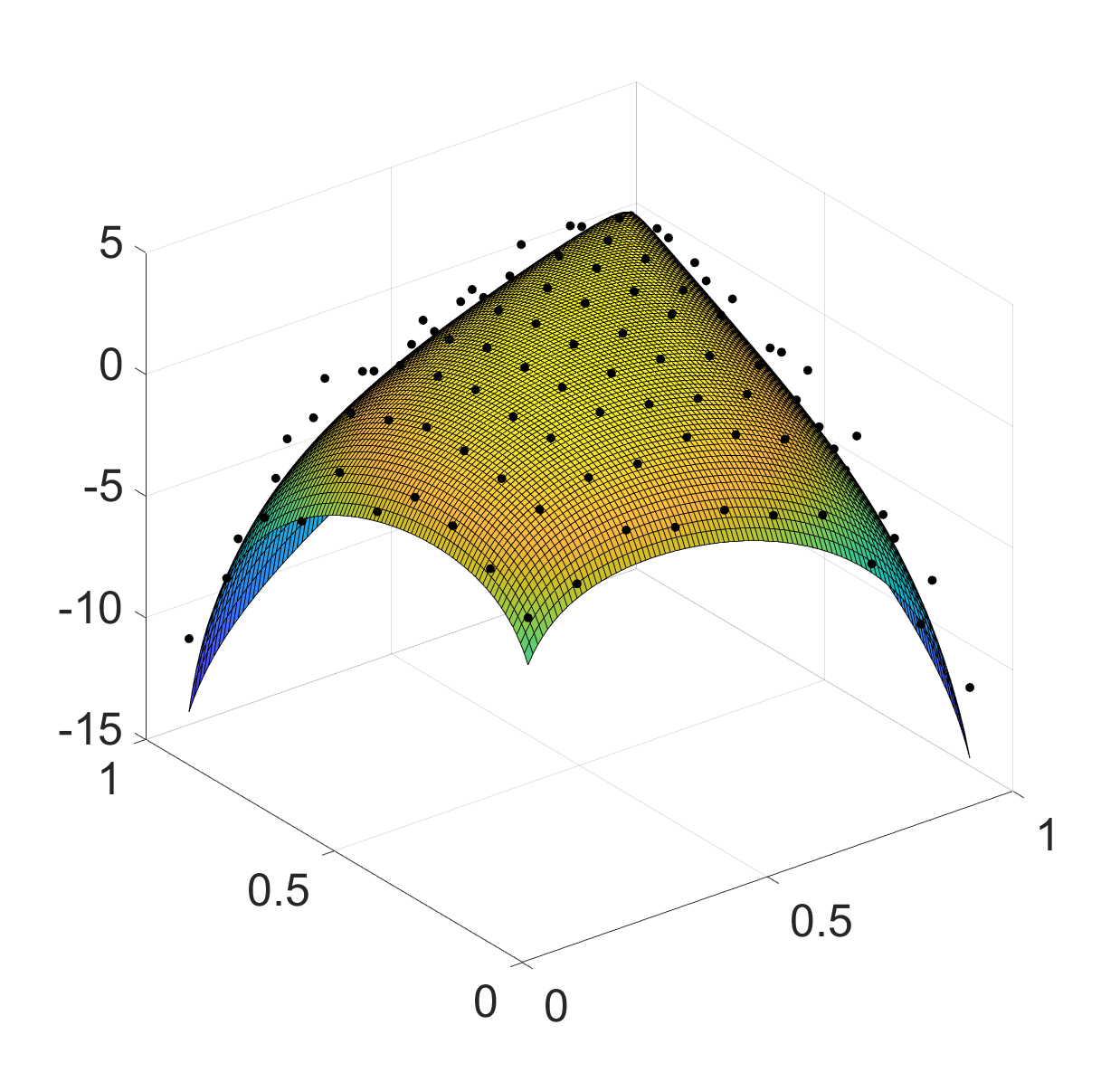}
  \includegraphics[width=0.35\textwidth]{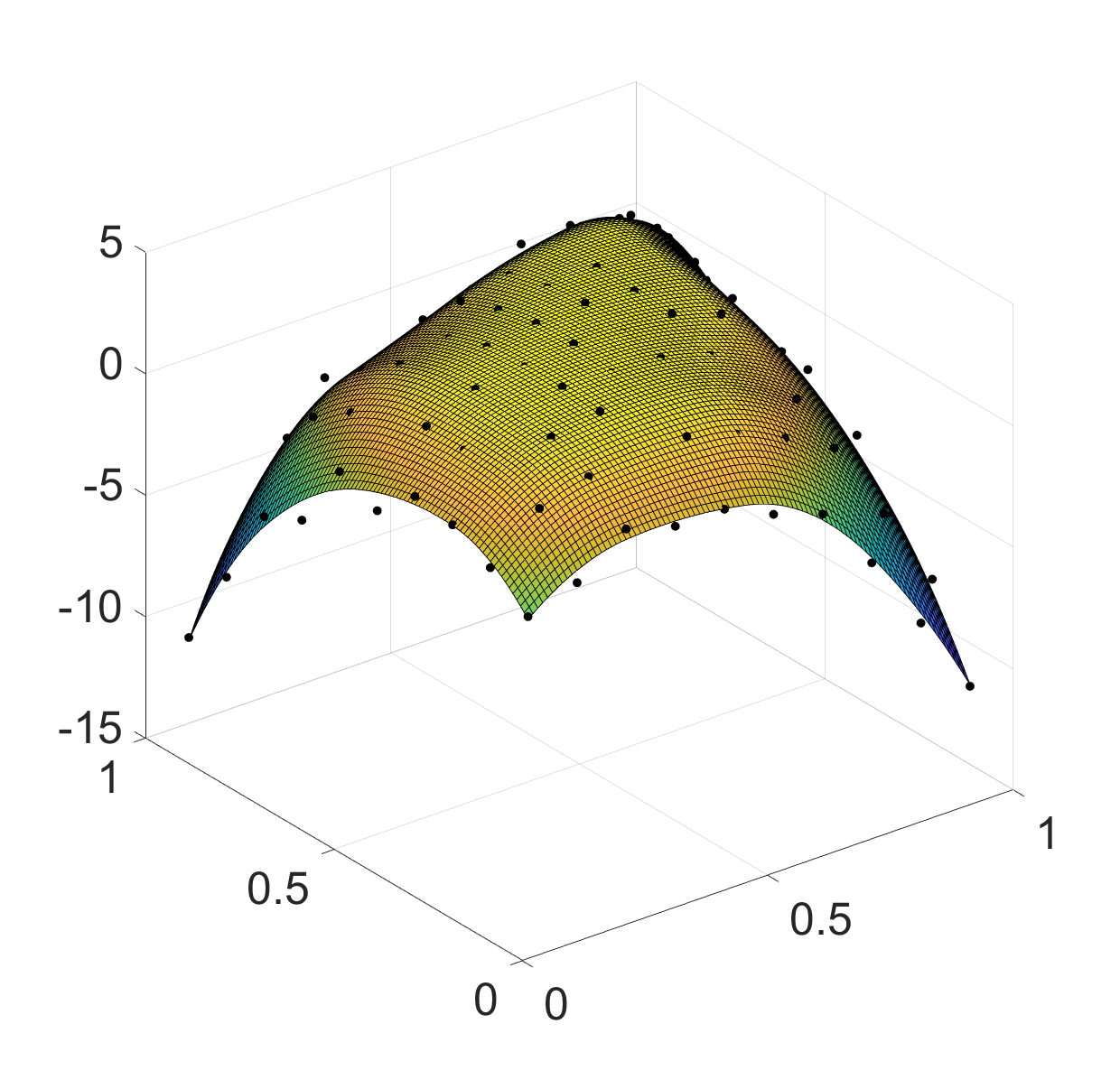}  
  \caption{The clr transformed density of the beta distribution (true density) with the clr transformed histogram data as points (left) and the corresponding $Z\!B$-spline approximation (density estimate), again together with highlighted transformed histogram data (right)}
  \label{simulation_approximation}
\end{figure}

\begin{figure}[h!]
  \centering
  \includegraphics[width=0.35\textwidth]{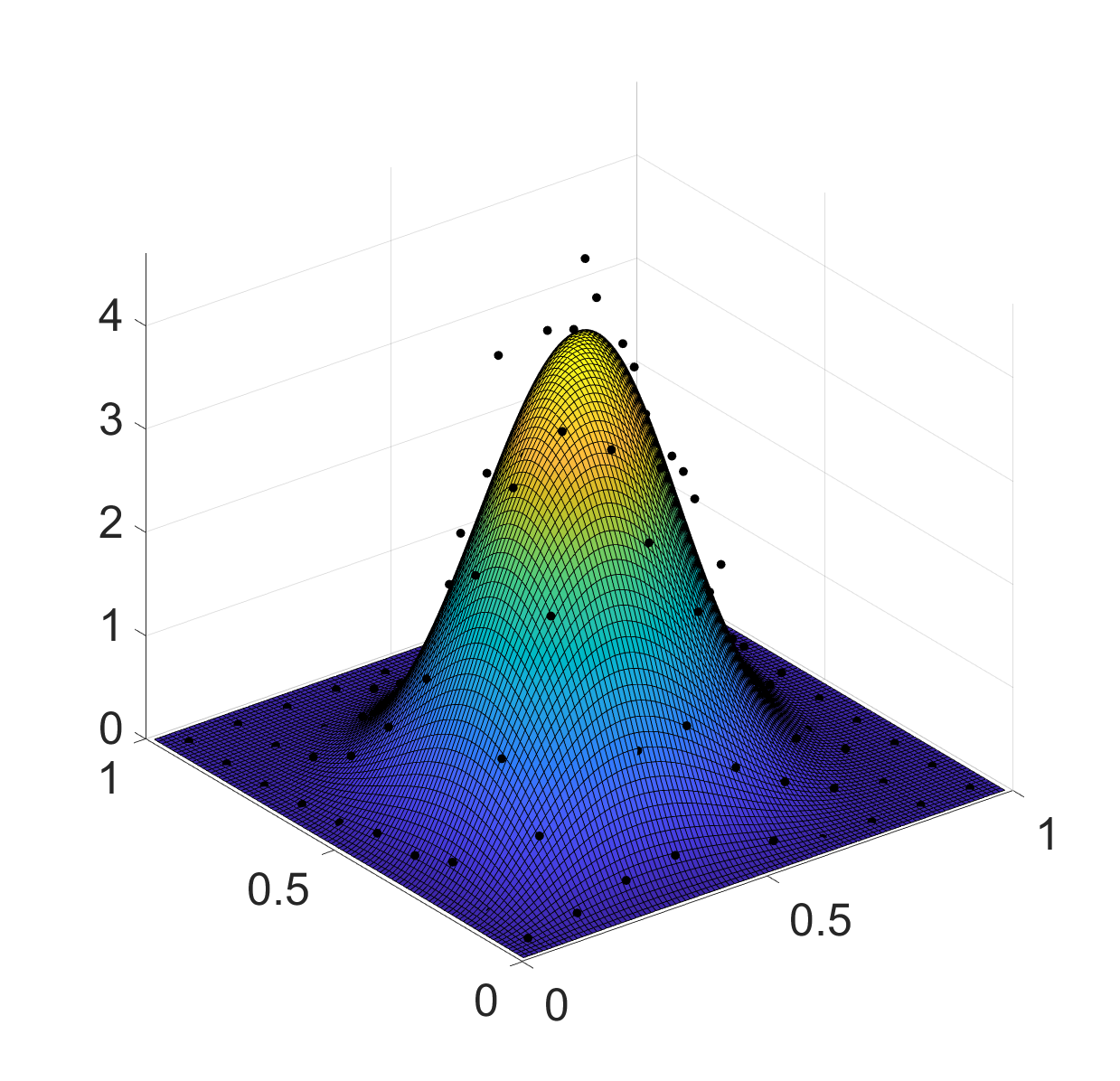}
  \includegraphics[width=0.35\textwidth]{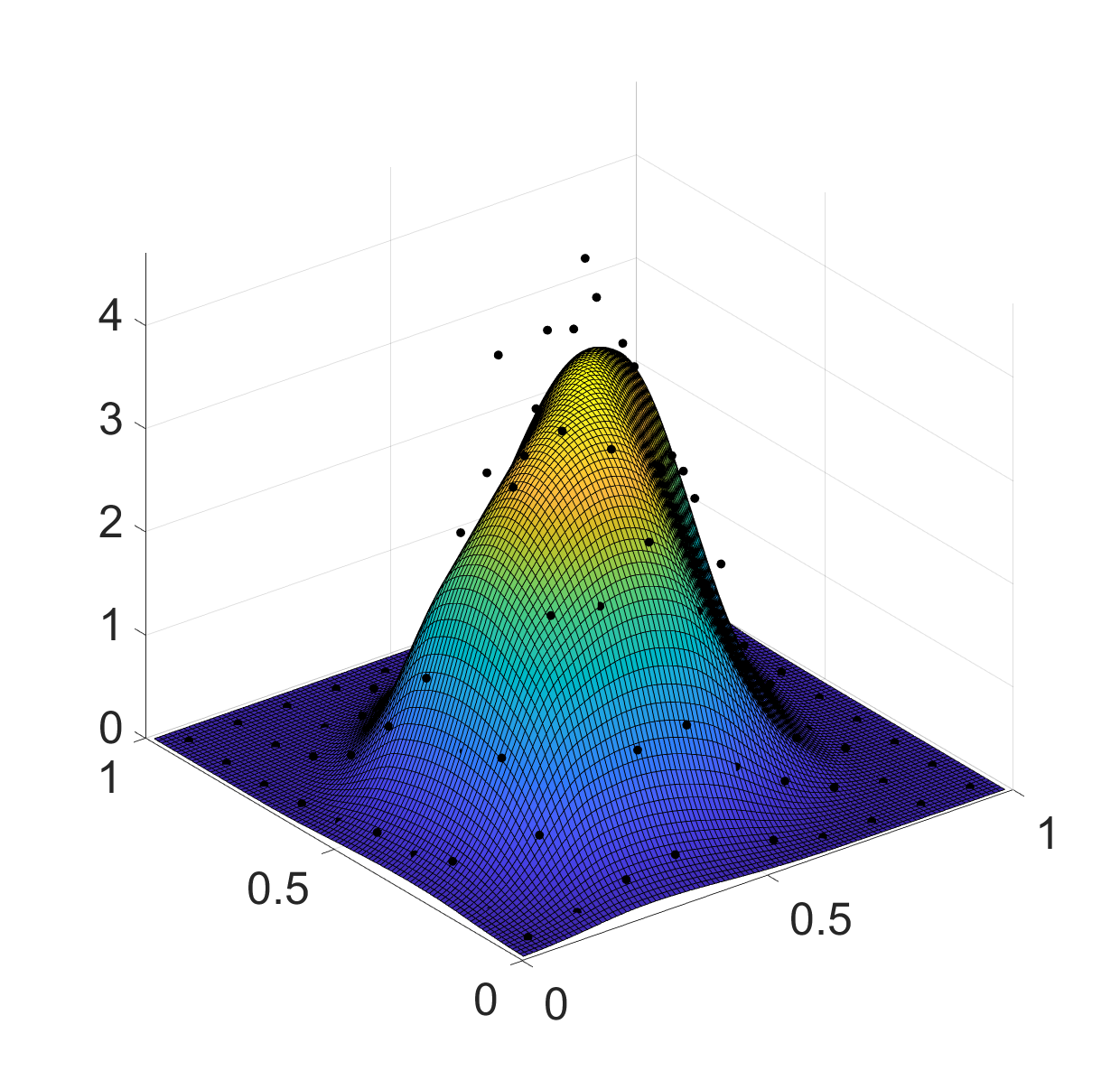}  
  \caption{The true density of the bivariate beta distribution (left) and corresponding compositional spline as its estimate (right) with highlighted iclr transformed histogram data}
  \label{simulation_estimate}
\end{figure}

As described in Section \ref{Bayes}, the bivariate density (and its estimate) can be decomposed into independent and interactive parts directly in ${\cal B}^2(\Omega)$ using Theorem \ref{Bmargin} and Corollary \ref{Bdecomp}. Alternatively, the clr transformed density function (and its $Z\!B$-spline representation) can be decomposed directly in $L^2_0(\Omega)$ using Theorems \ref{clrmarg}, \ref{clrrepr} and the results can then be converted to the original ${\cal B}^2(\Omega)$ space, which is much more convenient from a practical point of view. The corresponding decomposition in $L^2_0(\Omega)$ and ${\cal B}^2(\Omega)$ is depicted in Figures \ref{simulation_ZB_decomposition} and \ref{simulation_comp_decomposition}, respectively.

\begin{figure}[h!]
  \centering
  \includegraphics[width=0.35\textwidth]{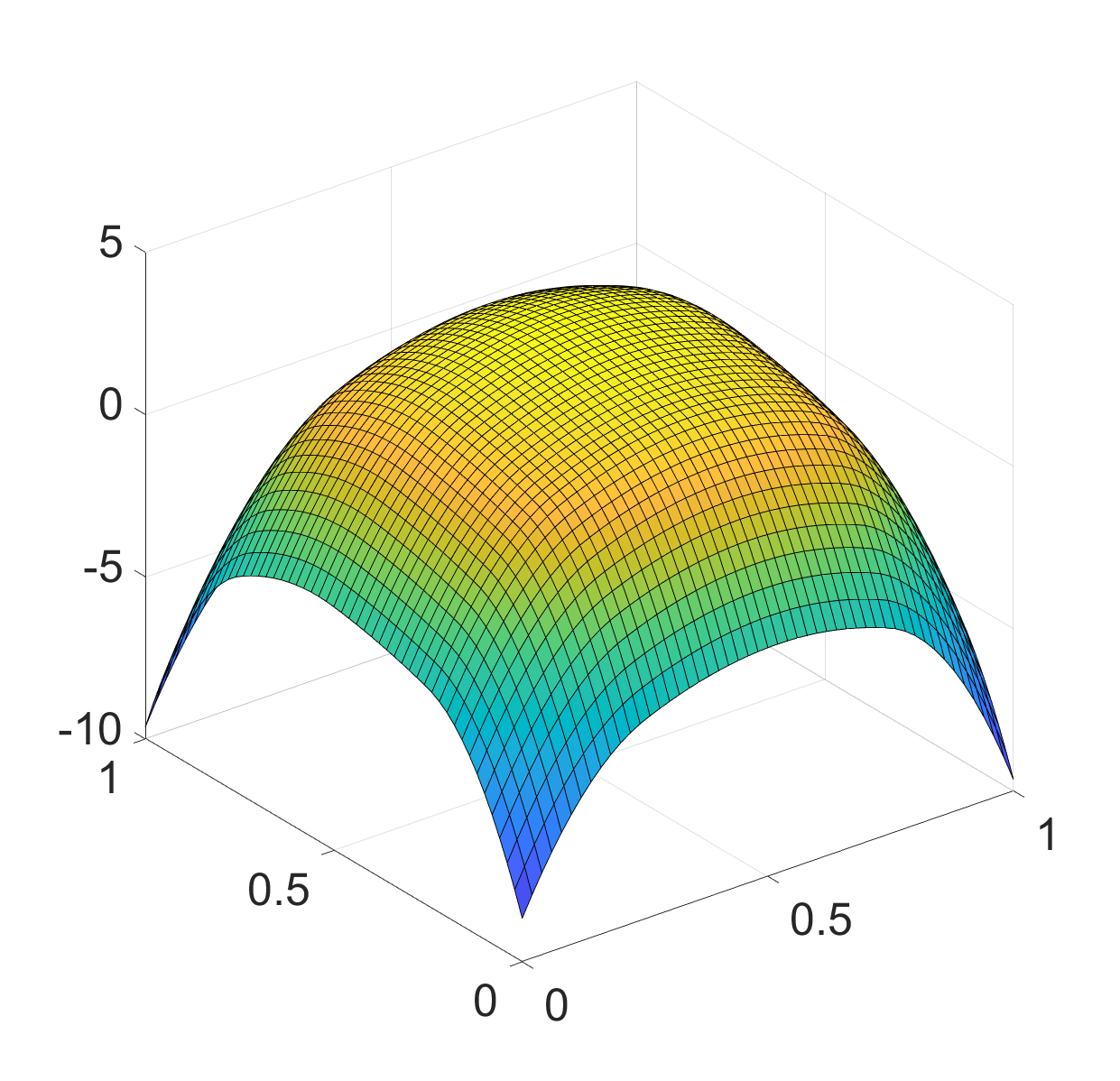}
  \includegraphics[width=0.35\textwidth]{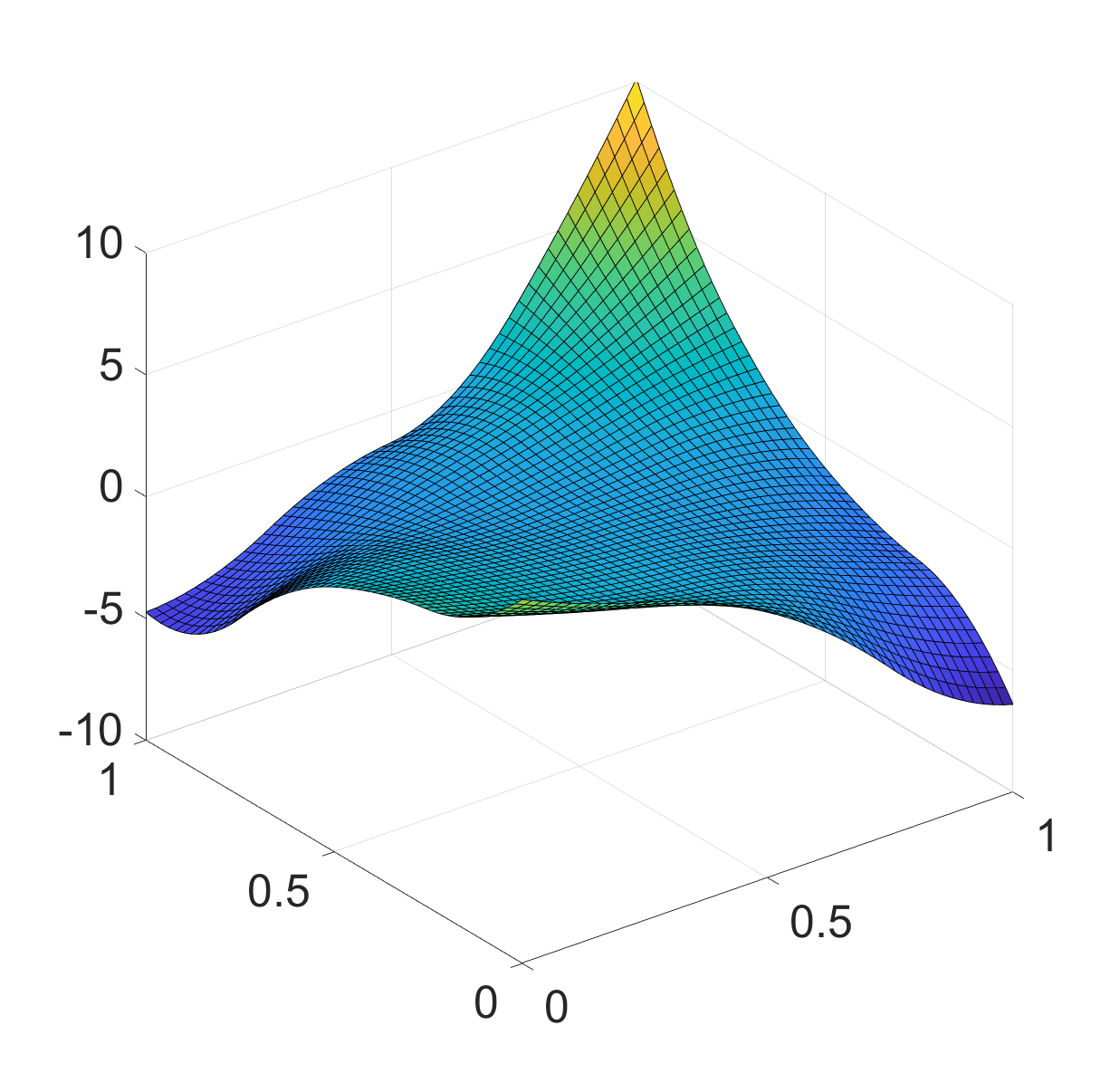}  
  \caption{Independent part (left) and interactive part (right) of the $Z\!B$-spline representation of Figure \ref{simulation_approximation} in $L^2_0(\Omega)$}
  \label{simulation_ZB_decomposition}
\end{figure}

\begin{figure}[h!]
  \centering
  \includegraphics[width=0.35\textwidth]{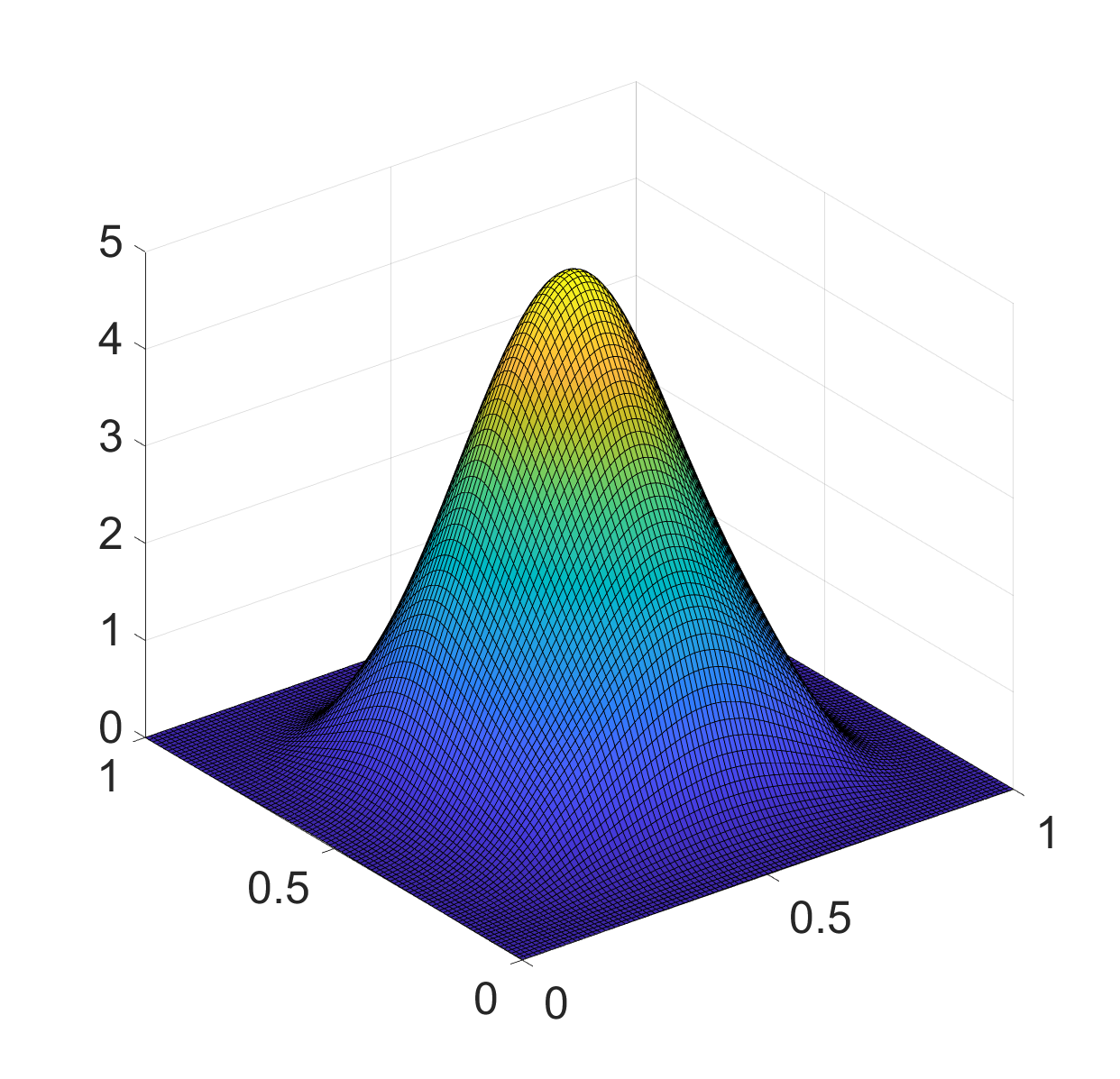}
  \includegraphics[width=0.35\textwidth]{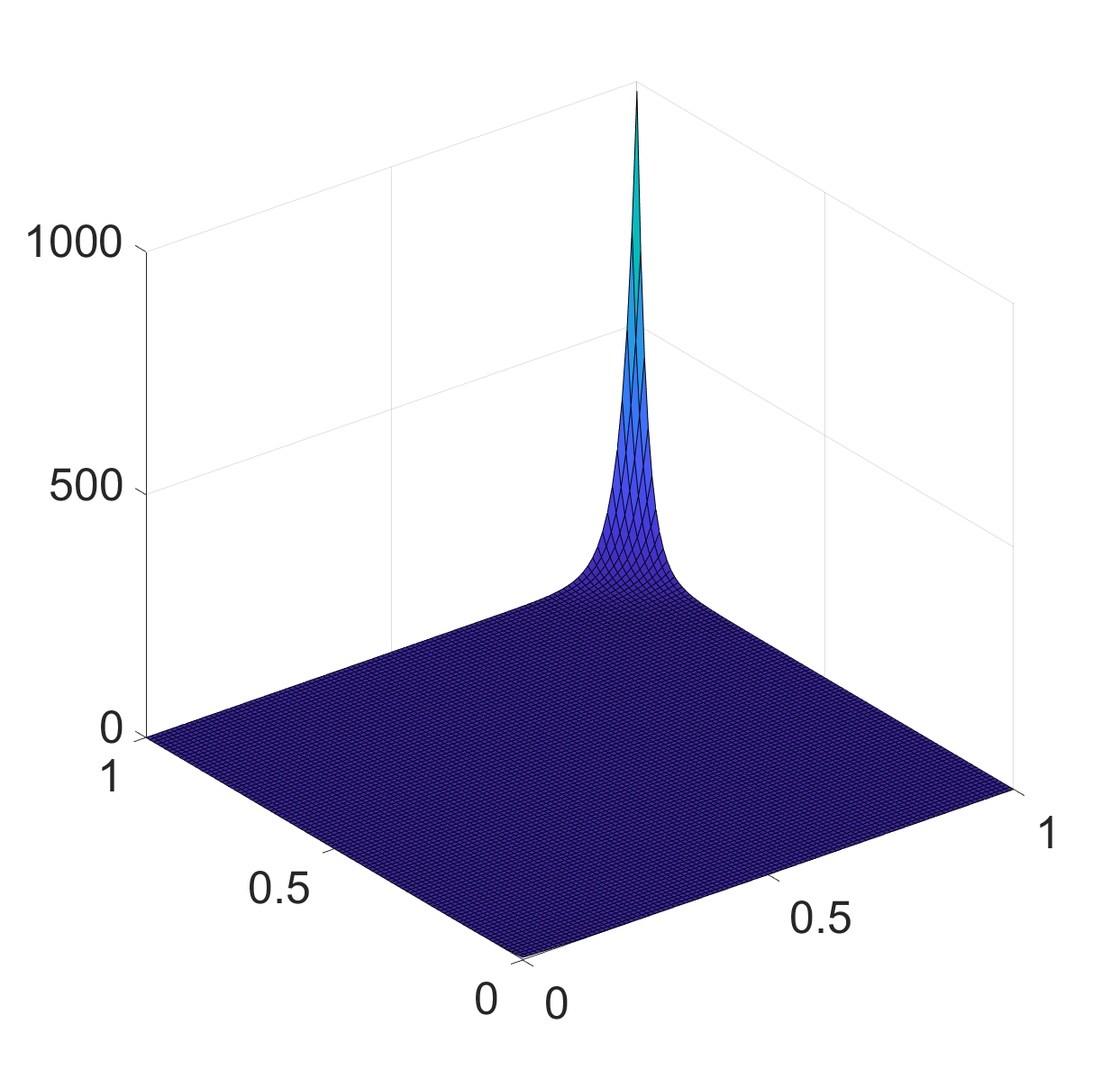}  
  \caption{Independent part (left) and interactive part (right) of the compositional spline of Figure \ref{simulation_estimate} in ${\cal B}^2(\Omega)$}
  \label{simulation_comp_decomposition}
\end{figure}


\subsection{Validation}
The shape and properties of the compositional spline as a density estimate depend on a choice of multiple parameters. This choice is essential for reliable estimation and further data analysis; therefore, it is crucial to determine their optimal values. Above, the histogram was built using $m=n=10$ classes for both $x$ and $y$, which was set heuristically. Thus, as a validation, in the following  we assess the quality of approximation using different numbers of classes. Note that the choice of the number of classes is not completely arbitrary since for a smoothing problem the condition $m\,n \; > \;(g+k)(h+l)+g+k+h+l$ has to be fulfilled. We assumed the original setting of the spline parameters, i.e. the degrees of the spline $k,\,l$, the degrees of the respective derivatives $p,\,q$, the sequences of knots $\Delta\lambda,\,\Delta\mu$ and smoothing parameter $\rho$, to be identical as in Section \ref{data:simu}. Different numbers $m=n$ of histogram bins were compared for the same $R=100$ Monte Carlo simulations as before. The compositional spline $\widehat{f_r}(x,y)$ (density estimate) for the r-th simulated dataset was compared with the true (reference) density $f(x,y)$ of the bivariate beta distribution using the bivariate integrated square error (ISE) \citep{compositional}, i.e.
$$
\mbox{ISE} = \Vert f(x,y) \ominus \widehat{f_{r}}(x,y)\Vert^{2}_{\mathcal{B}^2(\Omega)},\quad r = 1,\dots,\,R.
$$
Note that the criterion can be expressed and calculated equivalently in the $L^2_0$ space and results for one concrete choice of knots are displayed in the form of boxplots in Figure~\ref{beta_ISE_5_classes}.
\begin{figure}[h]
  \centering
  \includegraphics[width=.35\textwidth]{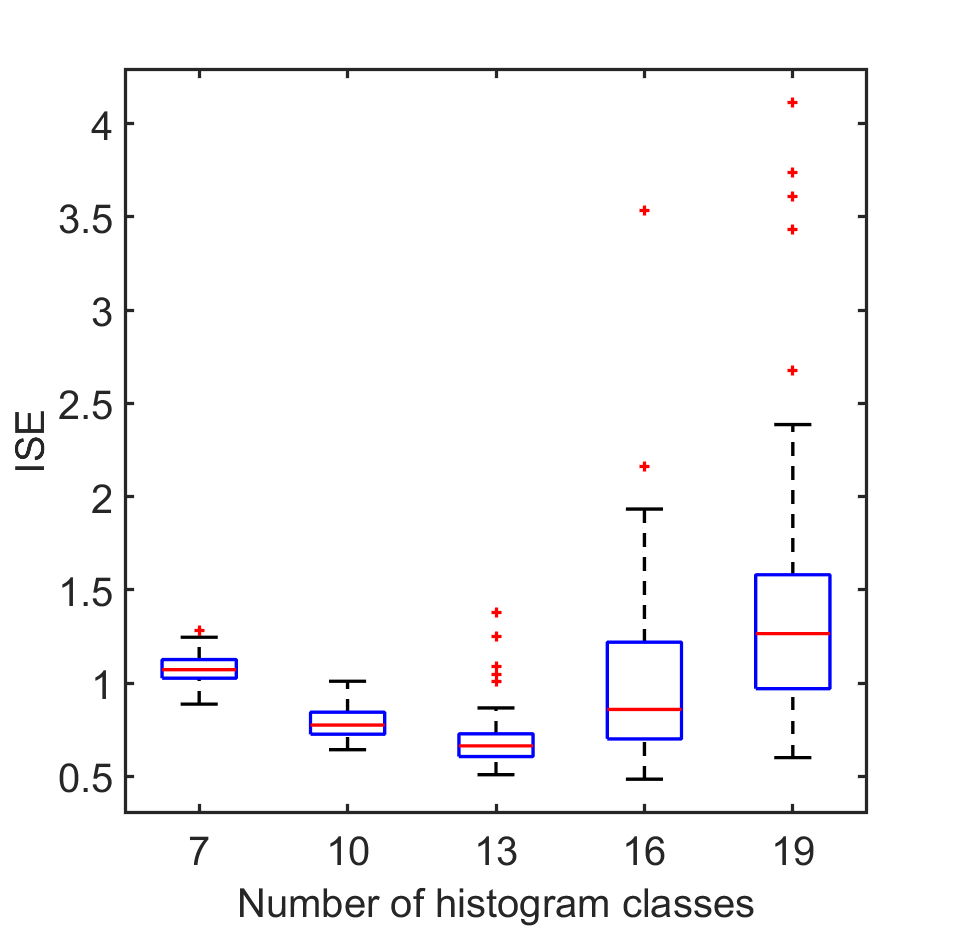}
  \caption{Boxplots of ISE between the reference density and its estimates for different numbers $m=n$ of histogram classes per direction using three interior knots with equidistant placing per direction}
  \label{beta_ISE_5_classes}
\end{figure}

Having set $g=h=3$ and $k=l=2$,  $m=n=6$ per direction is the minimum possible number of histogram classes. In this case, however, the error of approximation is high, because the insufficient number of bins causes oversmoothing of the resulting smoothing spline. The addition of more histogram bins improves the fit and reduces the approximation error. For high numbers, the ISE increases again, showing the typical bias-variance trade-off for the number of classes. Overall, a sufficient number of histogram bins seems necessary to not loose details of the density, while keeping the number small enough to allow for a reasonable average number of observations per bin.

In comparison, if we fix the number of classes to $m=n=13$ and preserve the setting $k=l=2$, the maximum possible number of interior knots is $g=h=10$. However, this approximation is not adequate, 
see the results of a~corresponding ISE comparison in Figure \ref{beta_ISE_13_knots}, which suggests that about 2 interior knots improve the fit.

To conclude, the quality of the resulting approximation depends on several parameters, especially on the number of histogram classes and the number (and position) of interior knots $g,\ h$. Since the reference density is here known, the approximation error can be measured effectively using ISE.  The optimal number of histogram classes and interior knots can then be determined as minimizing the ISE, which represents a compromise between oversmoothing (for lower numbers of bins and knots) and undersmoothing (for higher numbers of bins and knots). In a real data application, the ISE cannot be computed and (generalized) cross-validation needs to be employed, as illustrated in the next section.

\begin{figure}[h]
  \centering
  \includegraphics[width=.35\textwidth]{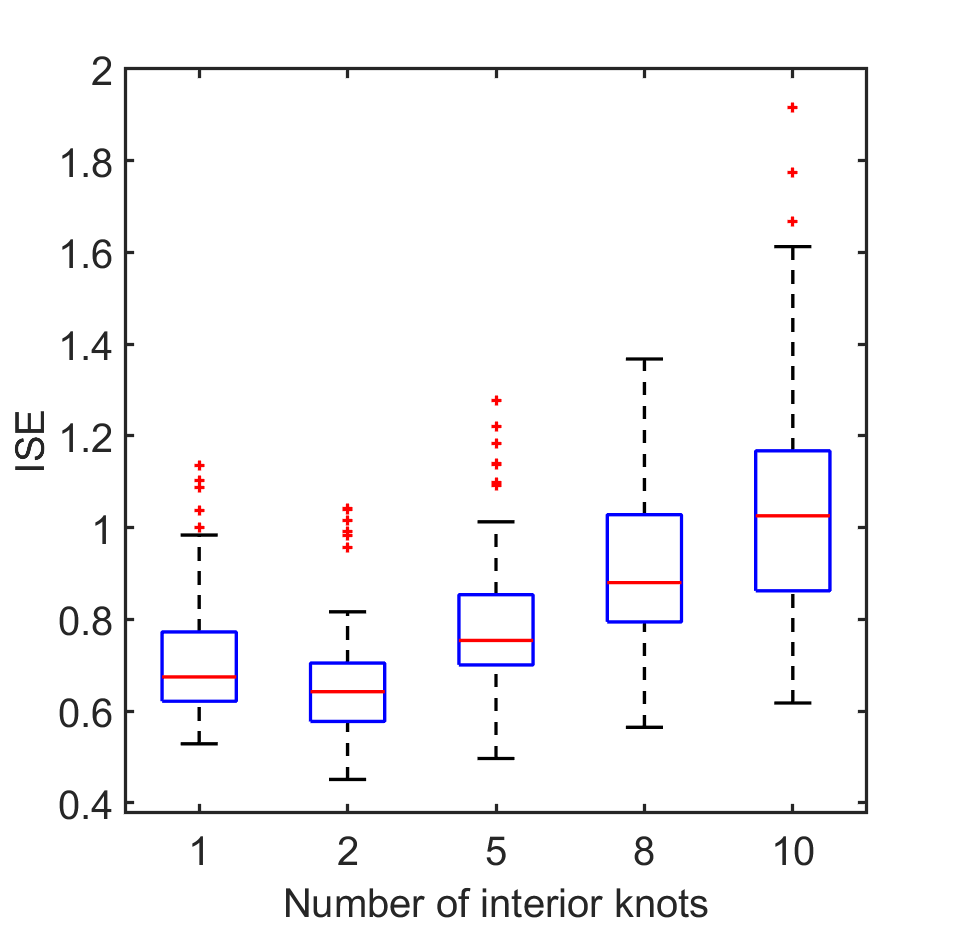}
  \caption{Boxplots of ISE between reference density and its estimates for different numbers of knots per direction using $m=n=13$ histogram classes per direction}
  \label{beta_ISE_13_knots}
\end{figure}

\section{Application} \label{appl}
In this section, we will apply the proposed methodology on empirical geochemical data from the Register of Contaminated Areas (Registr kontaminovan\'ych ploch) collected by the Department of Agriculture of the Czech Republic, which gives information about soil concentrations for various chemical elements in 76 Czech administrative districts \cite{grygar23, grygar24, eagri, vacha15}. Our focus is on distributions of log-transformed values of copper (Cu) and zinc (Zn), and some $N\!A$ values, i.e. not available measurements, occuring in either of the two elements were omitted. Accordingly, the number of observations varies for each district from 40 (Brno-m\v{e}sto) to 5956 (Vyso\v{c}ina).

Histograms with equally-spaced classes with a common range of $[0.588,4.58]$ for (log-)Cu and $[1.459,5.663]$ for (log-)Zn in all districts were then built from the raw log-transformed concentration data. Concentrations were measured in $mg\cdot kg^{-1}$. Zero frequencies, which occurred frequently in the resulting histograms but are problematic  for clr transformation,  were imputed heuristically but respecting the bivariate nature of the distributions using a stepwise algorithm. In every step,  zero values with at least one positive neighboring value were replaced with the geometric mean of neighboring non-zero frequencies multiplied by 2/3, which represents a rate of decline. The procedure was repeated until the histogram contained only positive values. Denote again by $\mathbf{x} = (x_1,\,\dots,\,x_m)$ and $\mathbf{y} = (y_1,\,\dots,\,y_n)$ the midpoints of histogram classes in the directions of variables $\mbox{Cu}=x$ and $\mbox{Zn}=y$, respectively. The corresponding histogram frequencies represent the aggregated raw data and 
$\mathbf{F} = \left(f_{ij}\right)_{i,j=1}^{m,n}$ the matrix of clr transformed histogram frequencies. Then $(\mathbf{x},\,\mathbf{y},\,\mathbf{F})$ constitutes the input data for density estimation via spline approximation.

After collection and preprocessing the data separately for every district, we proceeded to find the smoothing spline using Theorem \ref{biv_smoothing_th} with the following parameter setting, in line with recommendations from previous sections. Accordingly, we employed bicubic smoothing splines of degrees $k=l=3$ with the degrees of the derivatives $p=q=2$. We assumed the number of histogram classes to be the same for both variables and given the results of the simulation study, we set $m=n=13$ with equidistant sequences of interior knots with $g=h=3$ in the direction of both variables. The value of the smoothing parameter was determined using the average GCV criterion \eqref{GCV} across all districts, resulting in an optimal value of $\rho = 0.001$, see Figure \ref{fig:GCV_geodata}.

\begin{figure}[h]
  \centering
\includegraphics[width=0.7\textwidth]{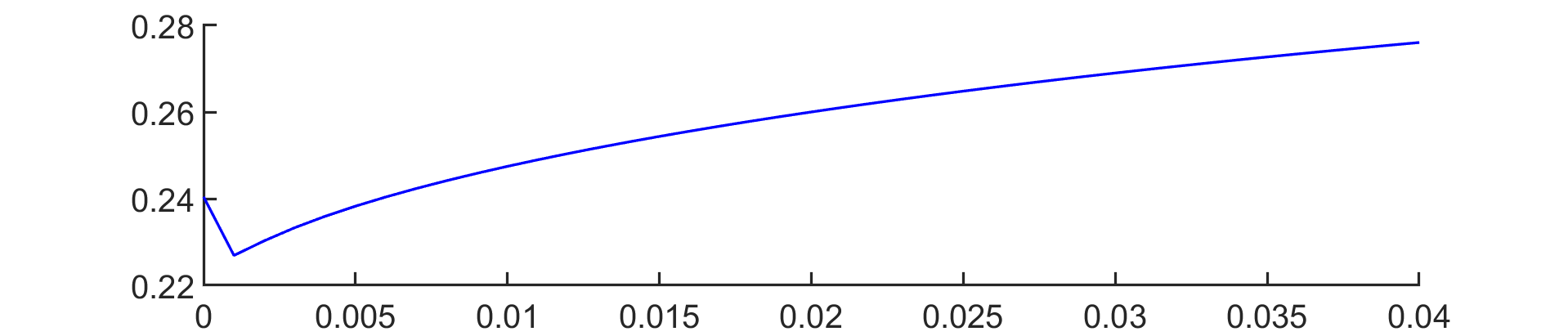}
  \caption{The mean GCV criterion across districts for the choice of an optimal smoothing parameter $\rho$ for fitting the clr transformed histogram values by tensor product smoothing splines}
  \label{fig:GCV_geodata}
\end{figure}

For comparison, we assumed various numbers of knots with other parameters fixed and, again, the condition for smoothing $m\,n > (g+k)(h+l)+g+k+h+l$ has to be taken into account. As the true density function is unknown, instead of the ISE the GCV criterion stated in \eqref{GCV} is used. GCV is calculated separately for each district, the results are displayed in Figure \ref{fig:application_GCV_knots}.

\begin{figure}[h]
    \centering
    \includegraphics[width=0.35\textwidth]{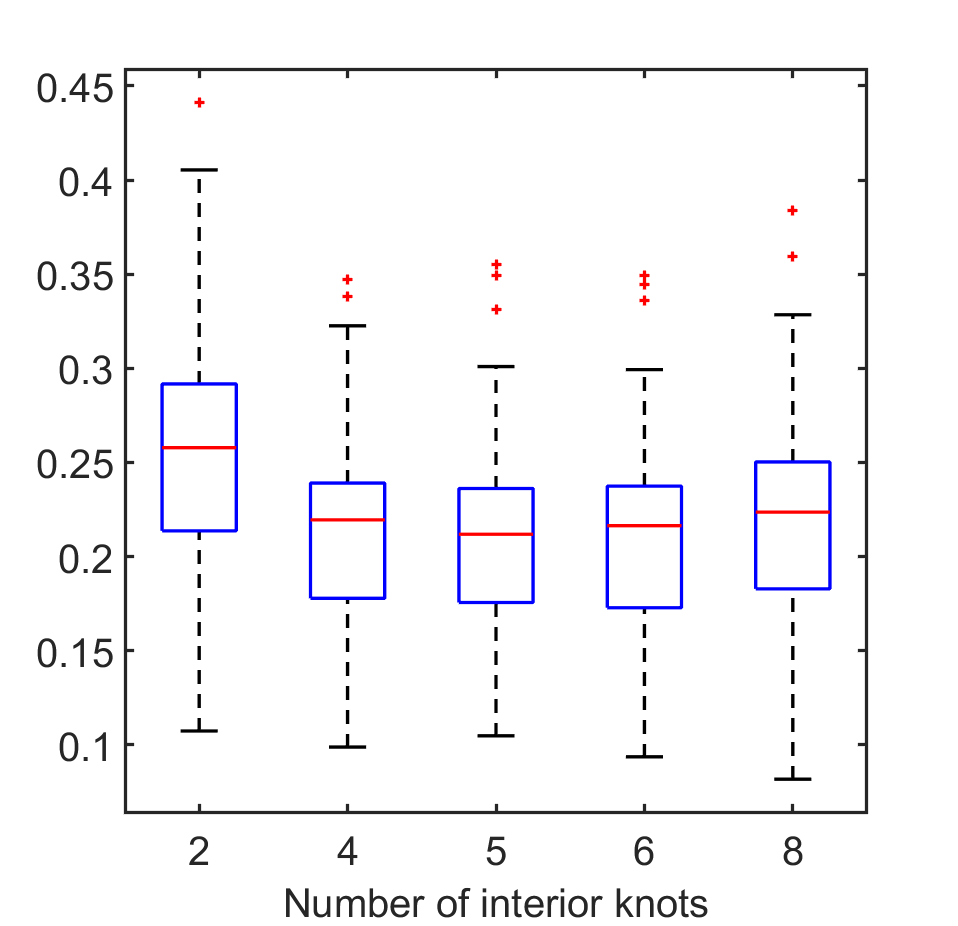}
    \caption{Boxplots of GCV for various numbers of interior knots calculated separately for each district}
    \label{fig:application_GCV_knots}
\end{figure}

The value of GCV varies for each district, however, the typical convex behaviour can be observed in the majority of the regions. On average, the GCV criterion reaches the minimum for $g=h=5$ interior knots, which represents a~reasonable compromise between robustness and the closeness of fit for the given number of histogram classes. The original histograms, the final estimates of the $Z\!B$-spline representations with the points of approximation and the corresponding compositional spline density estimates are visualized in Figure \ref{fig:application_representation}.

\begin{figure}[h!]
    \centering
    \includegraphics[width=0.3\textwidth]{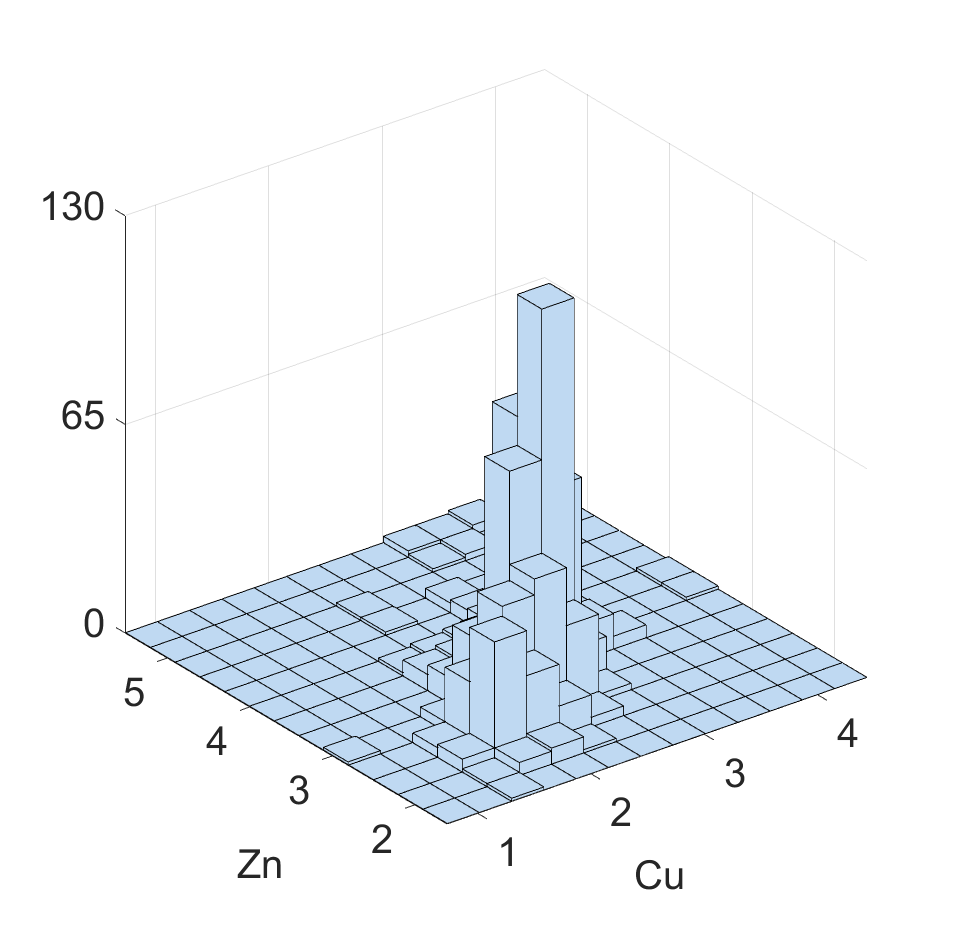}
    \includegraphics[width=0.3\textwidth]{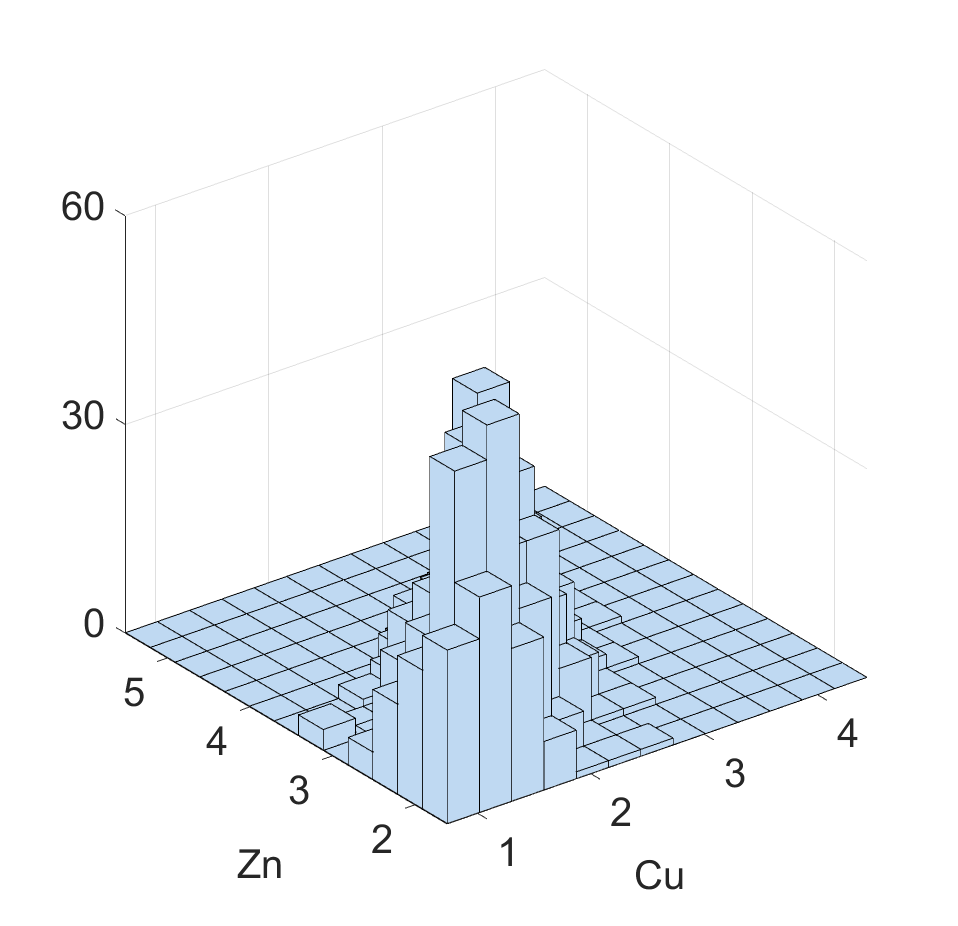}
    \includegraphics[width=0.3\textwidth]{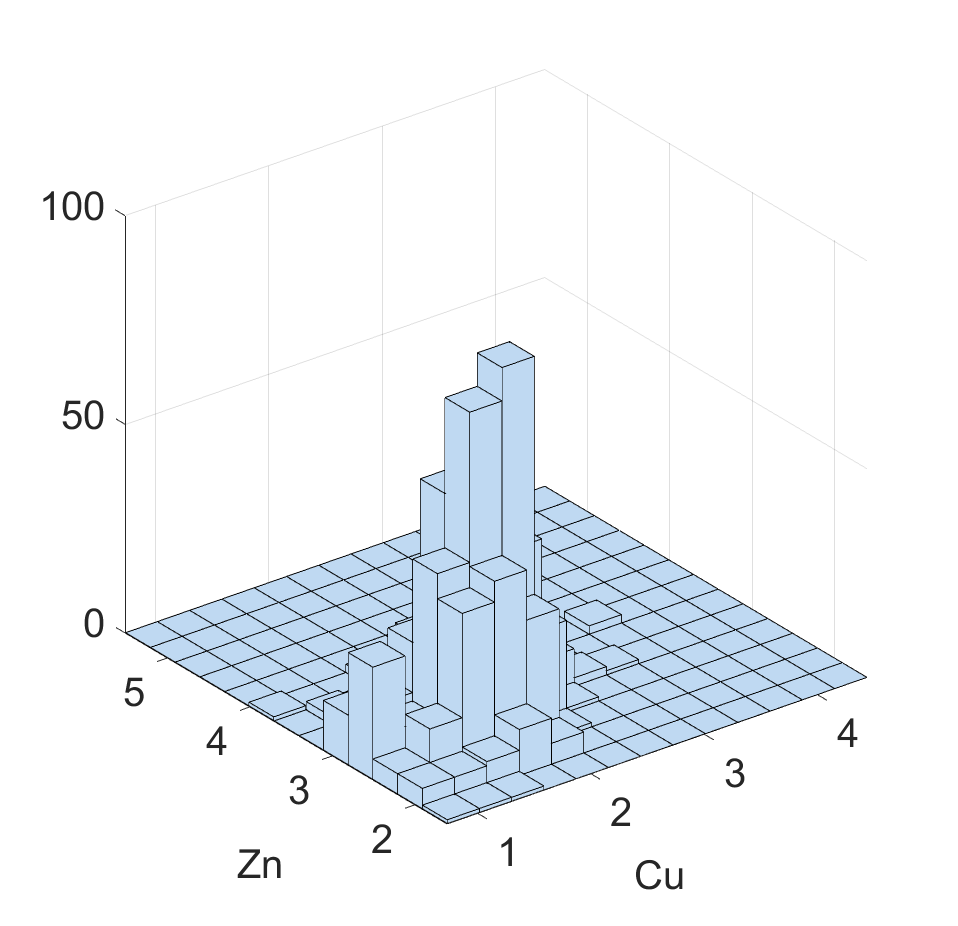}
    \includegraphics[width=0.3\textwidth]{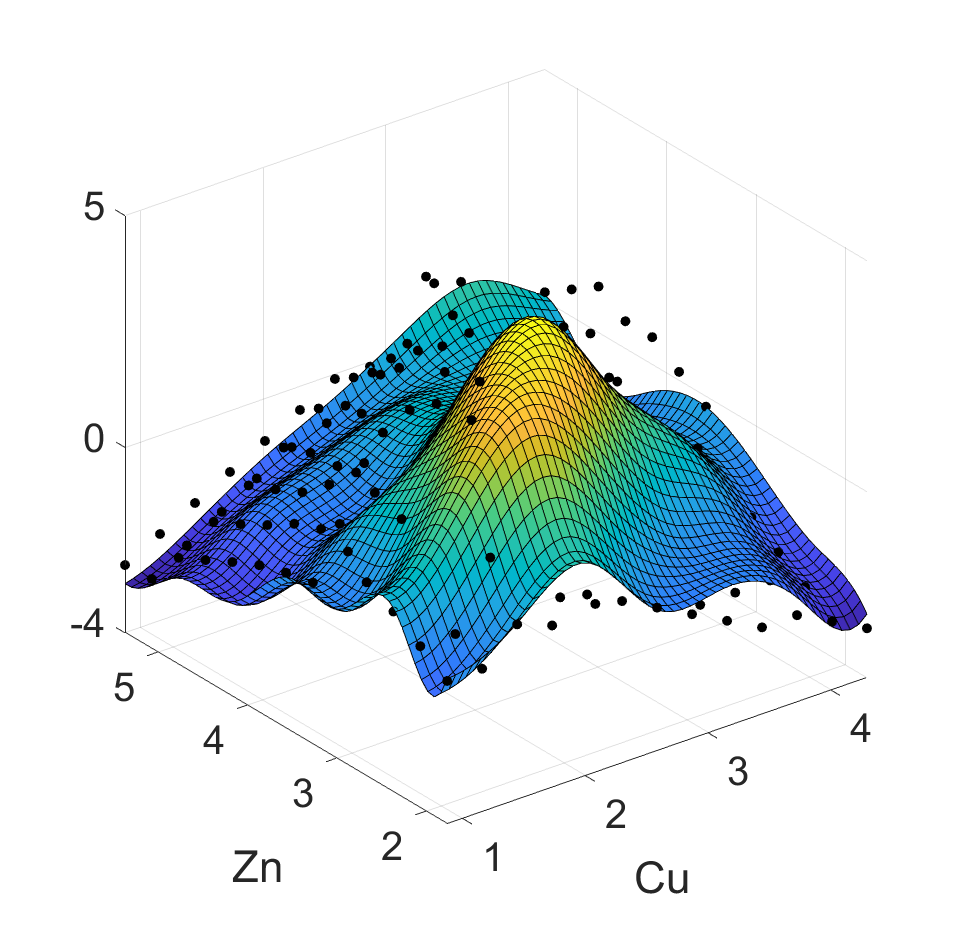}
    \includegraphics[width=0.3\textwidth]{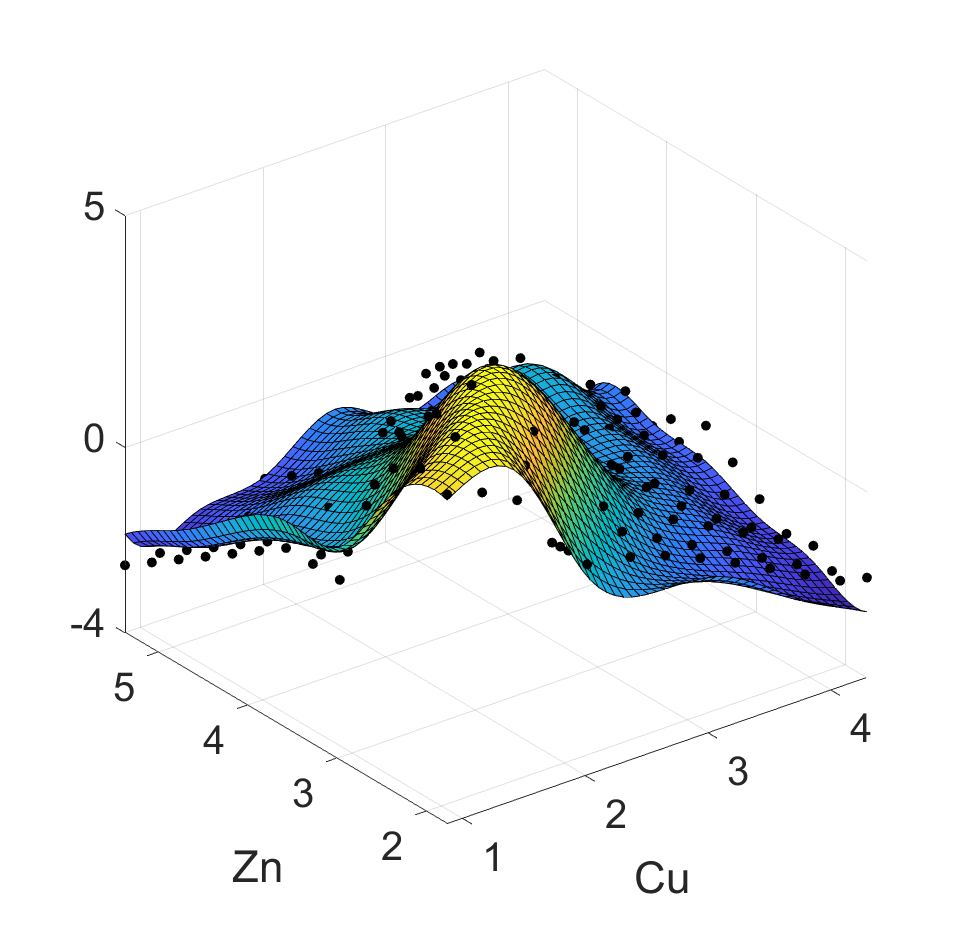}
    \includegraphics[width=0.3\textwidth]{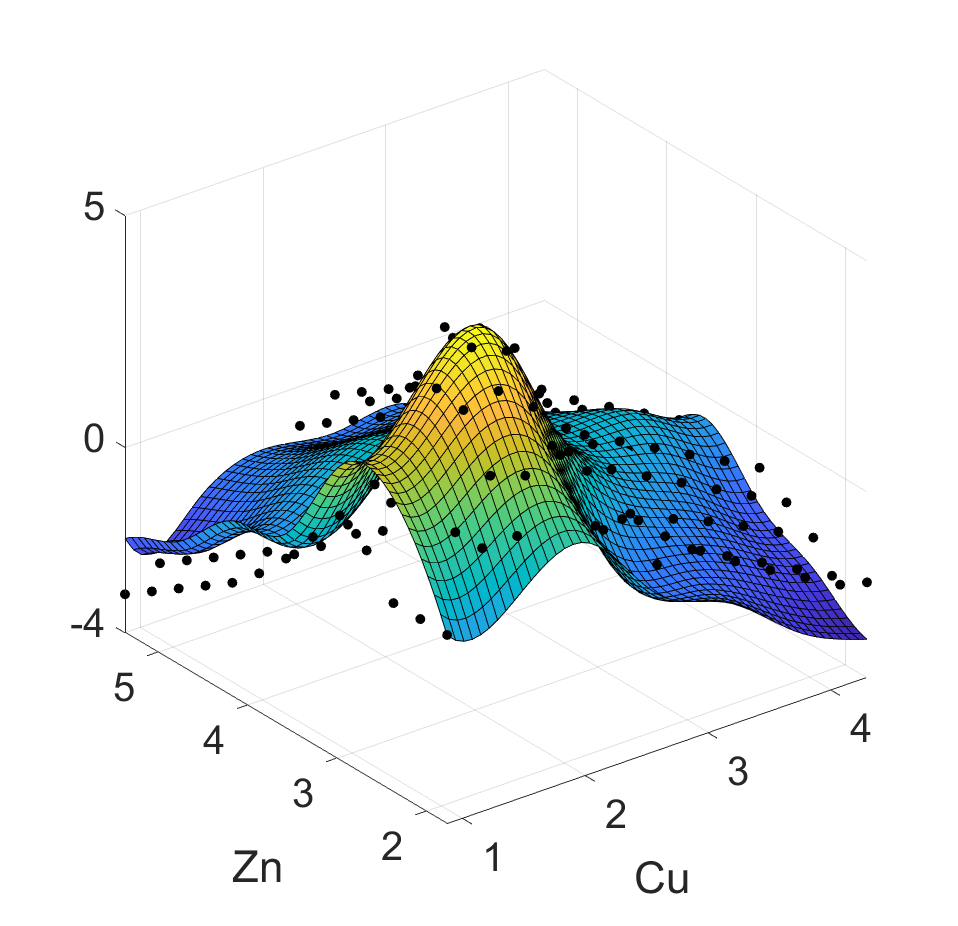}
    \includegraphics[width=0.3\textwidth]{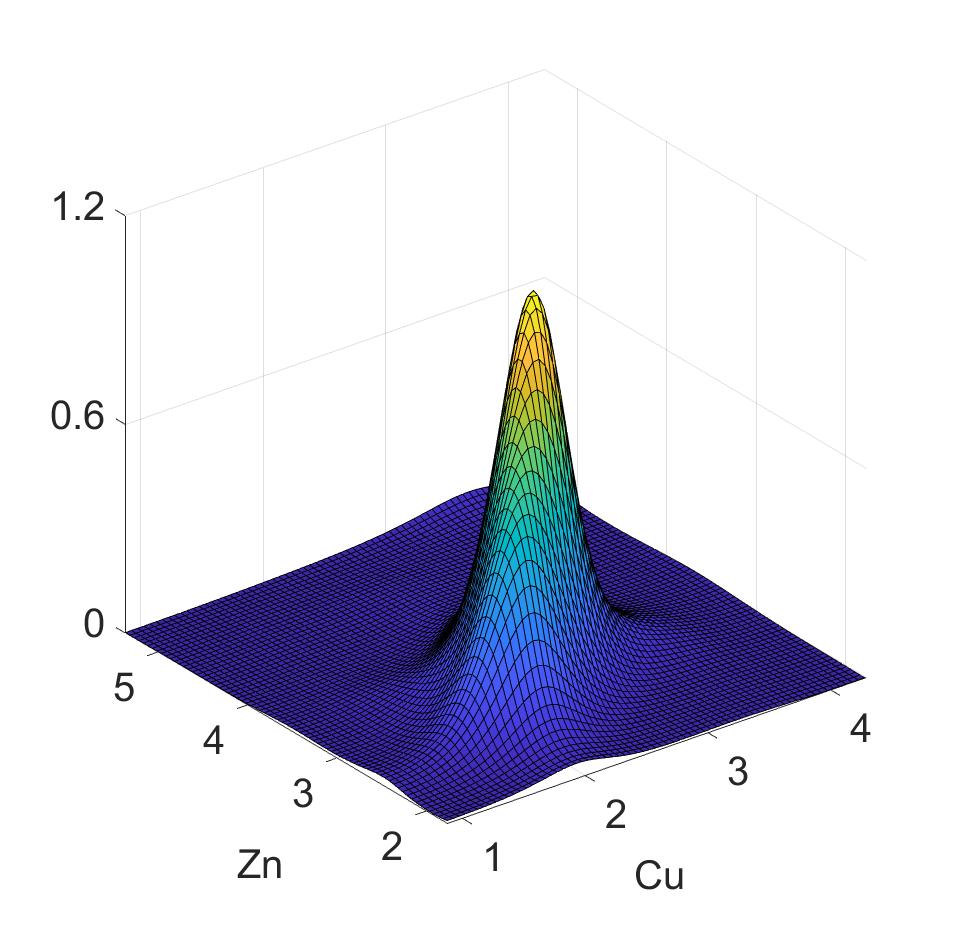}
    \includegraphics[width=0.3\textwidth]{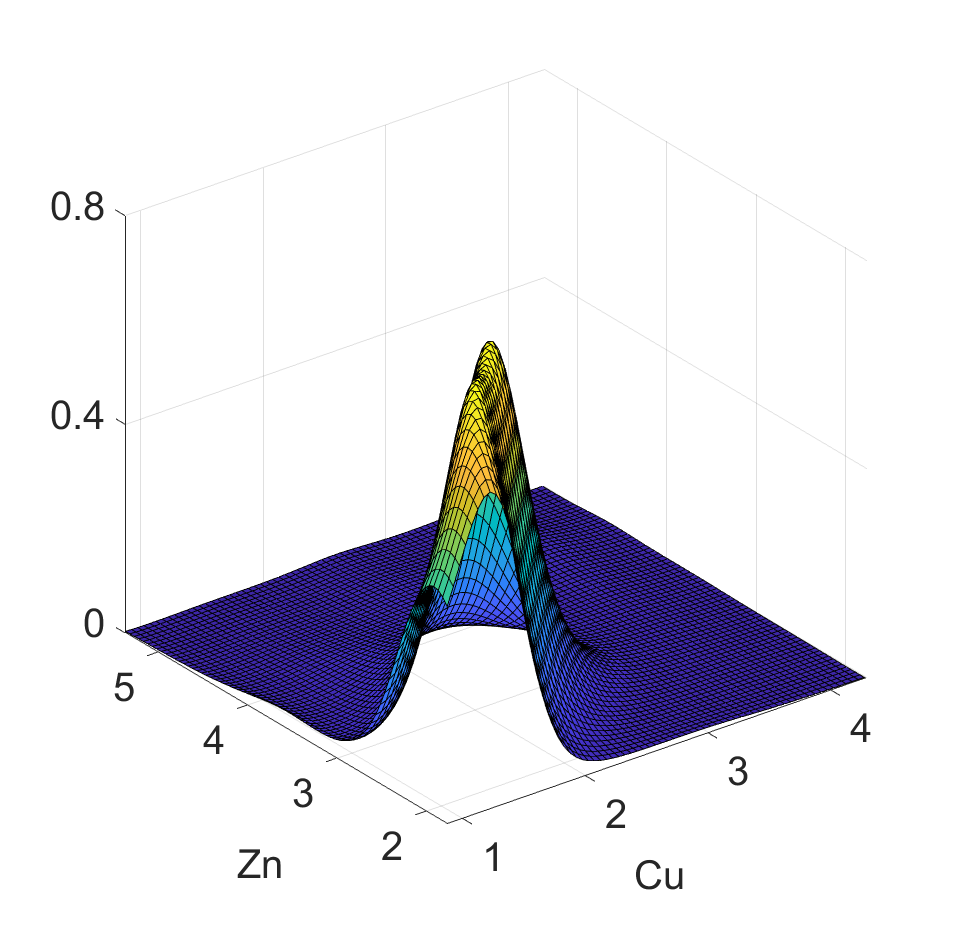}
    \includegraphics[width=0.3\textwidth]{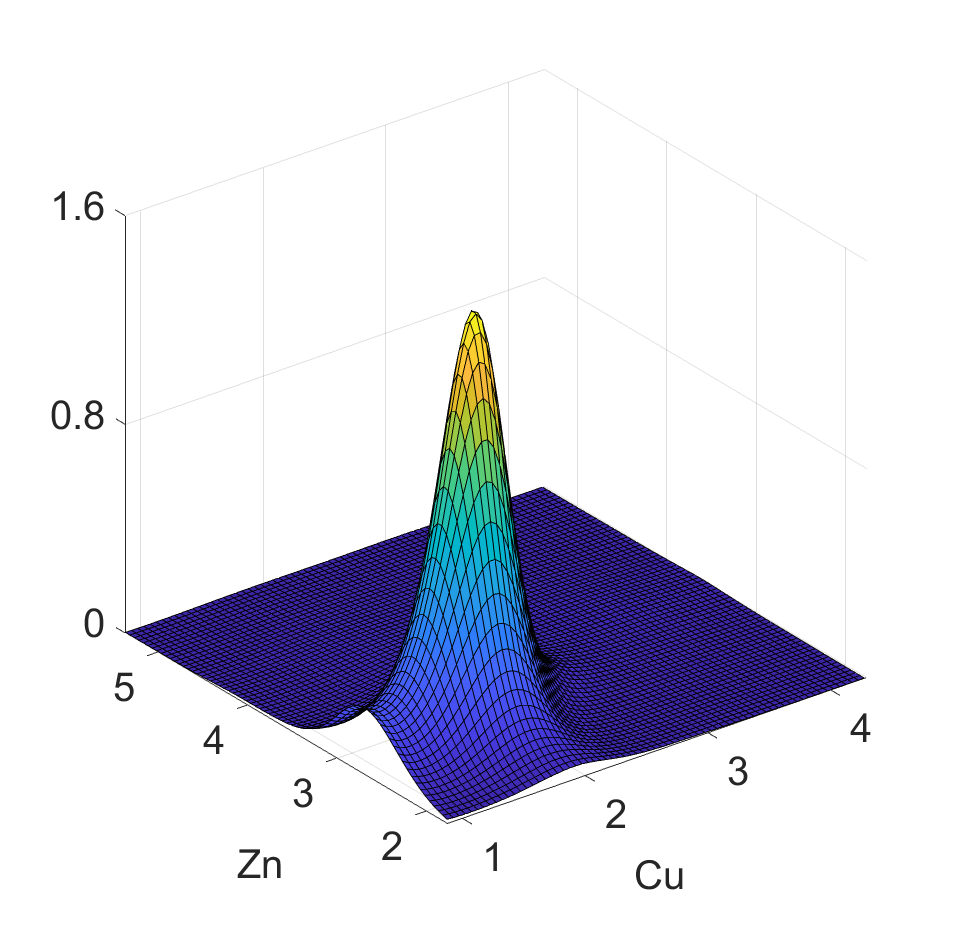}
    \caption{Histograms (top), corresponding $Z\!B$-spline representations (middle) and compositional splines as a density estimate (bottom) for regions Zl\'in (left), Liberec (middle) and \v{C}esk\'e Bud\v{e}jovice (right)}
    \label{fig:application_representation}
\end{figure}

Finally, the spline estimates can be decomposed into interactive and independent parts in $L^2_0(\Omega)$ and in ${\cal B}^2(\Omega)$ according to Section \ref{Bayes}, and for the previous choice of districts they are displayed in Figure \ref{fig:application_decomposition}.

From Figure \ref{fig:application_representation} we can see differences among districts such as clear multimodality of distributions. Note that this can be easily seen in the clr space, however, some important patterns would be lost if only inspecting the original densities in the ${\cal B}^2(\Omega)$. Moreover from decomposition of the densities (Figure \ref{fig:application_decomposition}), the main univariate patterns (from the independent part) as well as dependence structure of both elements (from the interactive part) can be clearly spotted. Specifically, in Zl\'in region there is a strong relation between high concentration values of both elements, which indicates the presence of polluted samples. In the other two regions, this pattern can not be seen and the main mode of the interaction densities corresponds to lower concentrations of Cu and Zn.

The introduced methodology represents a numerically stable approach for the  estimation of bivariate density functions. Since the resulting approximation is determined uniquely by the matrix of $Z\!B$-spline coefficients (or equivalently by $B$-spline coefficients as stated in Theorem \ref{biv_mn}), the resulting representation can be used effectively in subsequent statistical processing using popular methods of functional data analysis (FDA) \citep{kokoszka17}.

For a demonstration, we will briefly sketch the use of the spline approximation in basic descriptive statistics of FDA. We estimated densities for each district separately with a common parameter setting based on the previous results and validation. The mean function of the $Z\!B$-spline representation, which is easily obtained from averaging the spline coefficients in the representations for all districts, and the corresponding $C\!B$-spline are visualized in 
Figure~\ref{fig:application_means}. Compared to the densities of individual districts (Figure \ref{fig:application_representation}), the mean function eliminates the local effects and displays a nice unimodal fitted shape.

\begin{figure}[h!]
    \centering
    \includegraphics[width=0.32\textwidth]{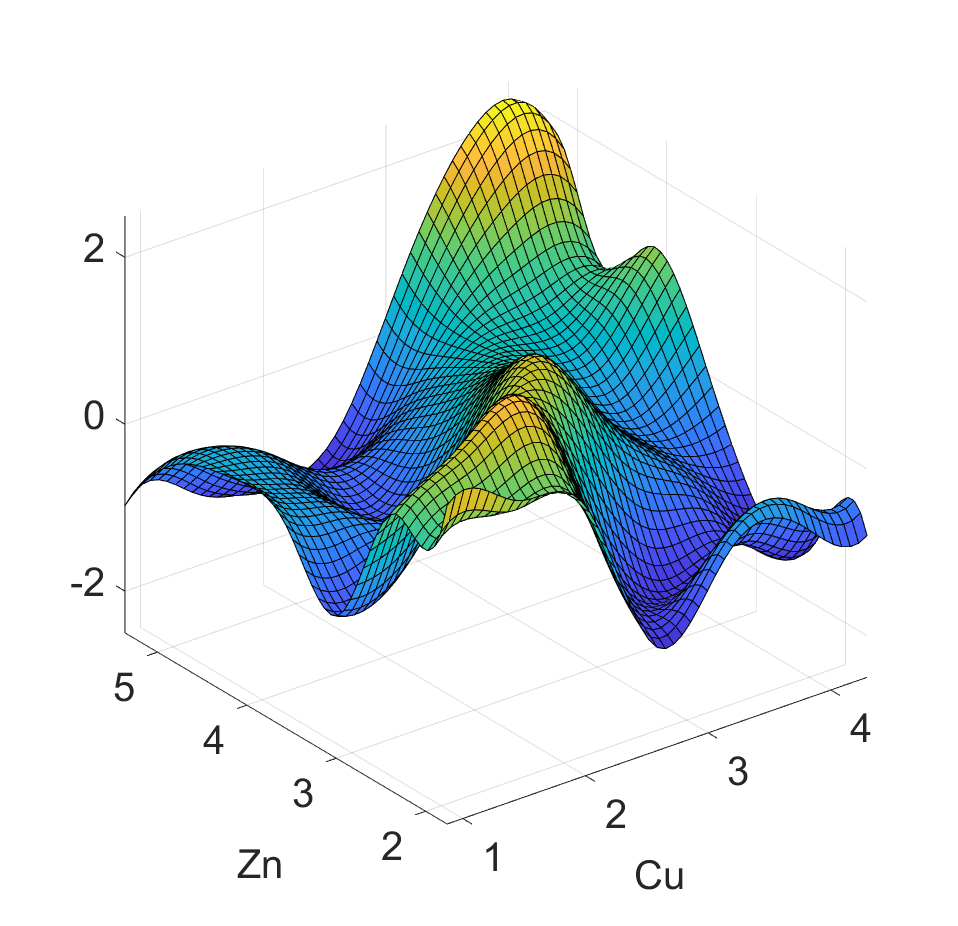}
    \includegraphics[width=0.32\textwidth]{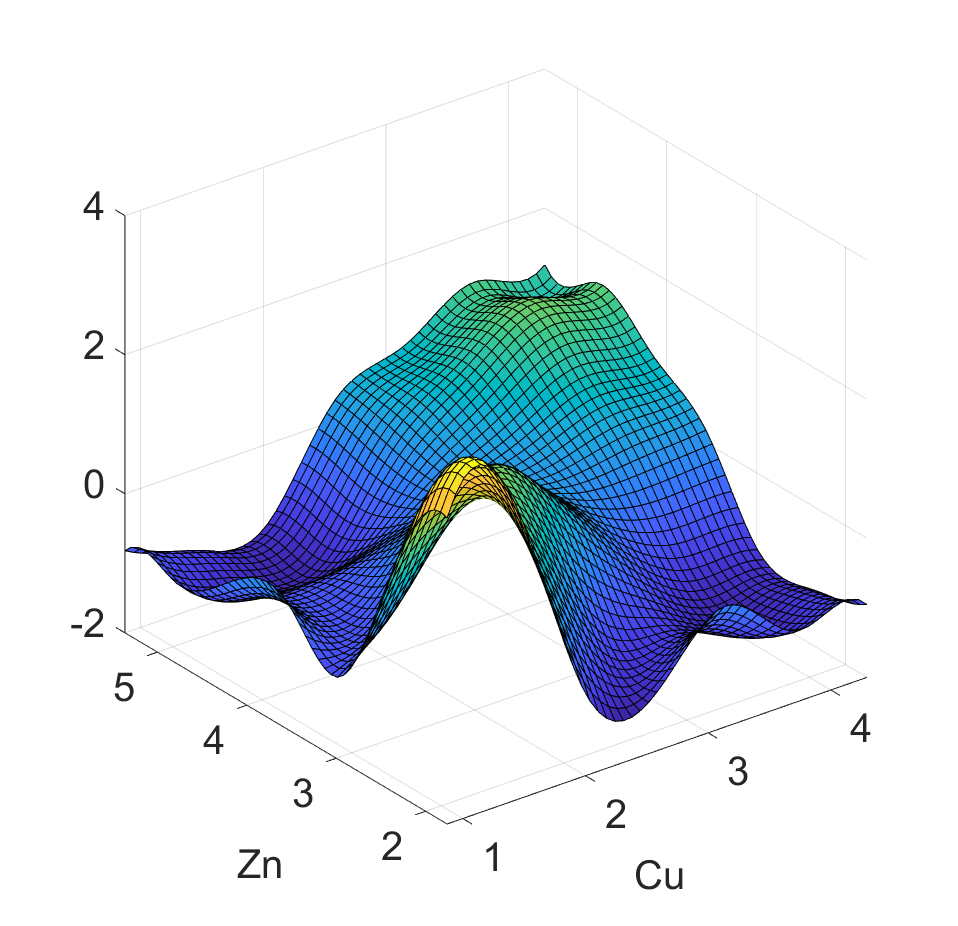}
    \includegraphics[width=0.32\textwidth]{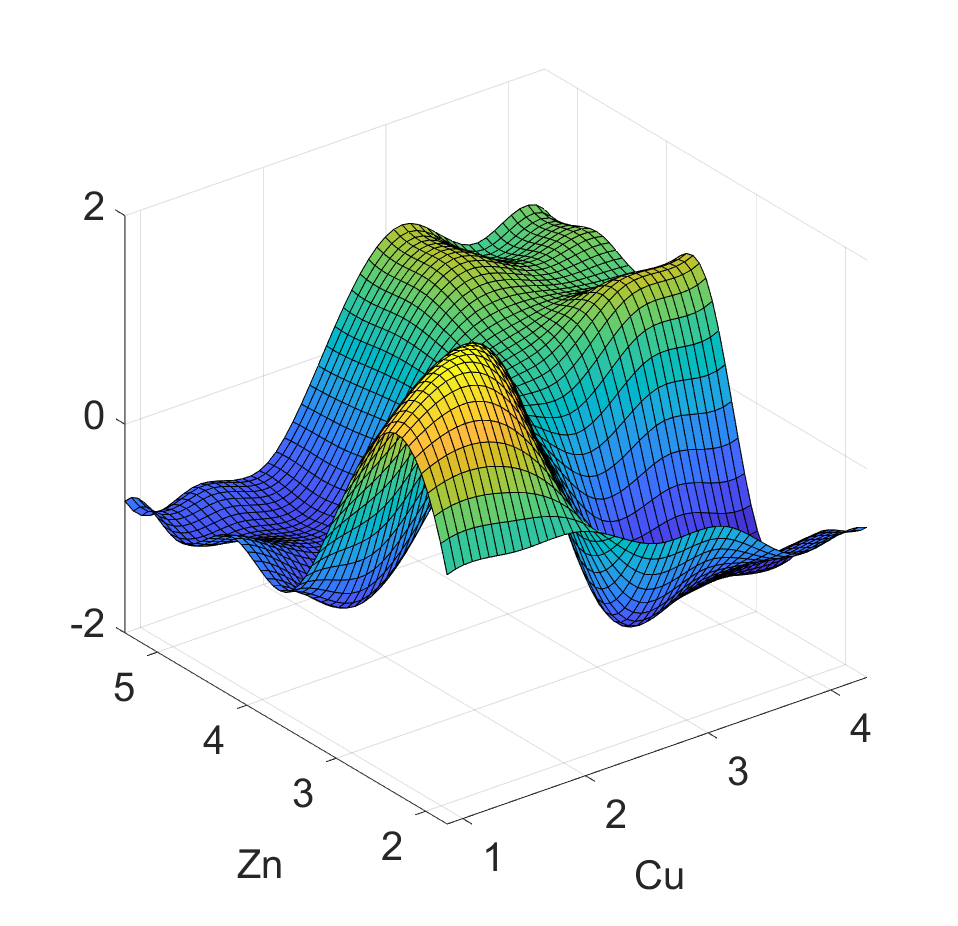}
    \includegraphics[width=0.32\textwidth]{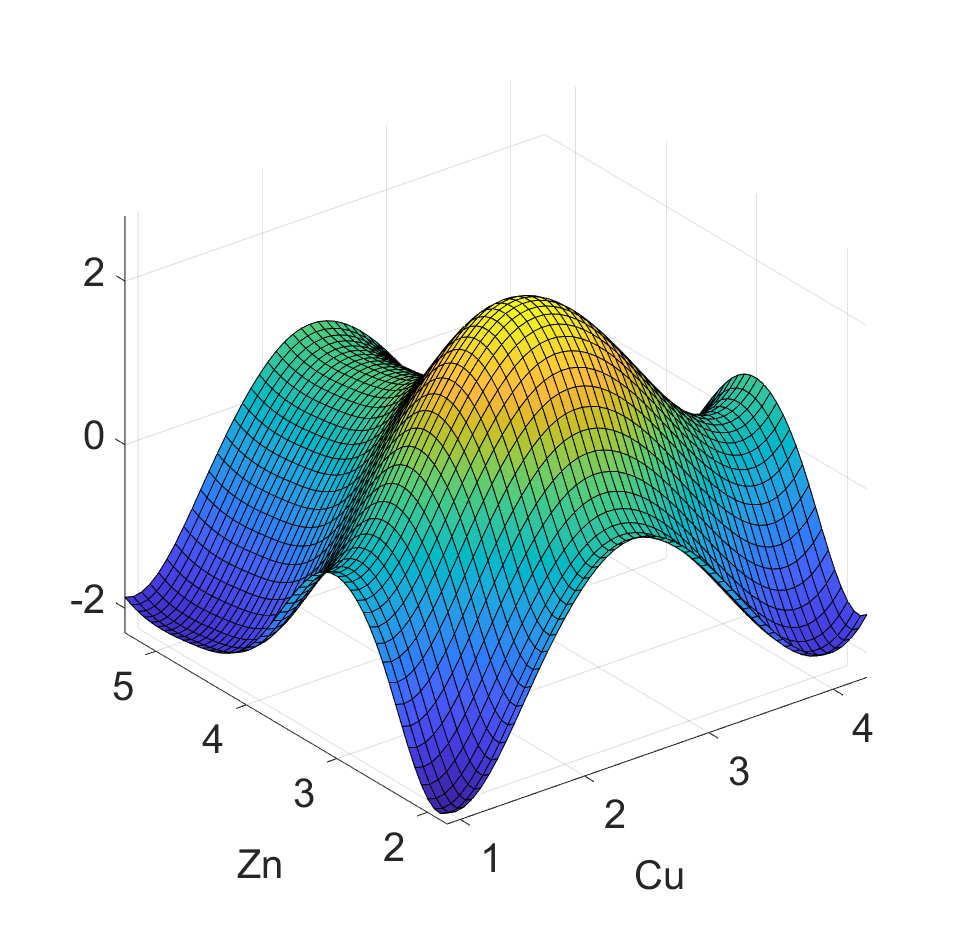}
    \includegraphics[width=0.32\textwidth]{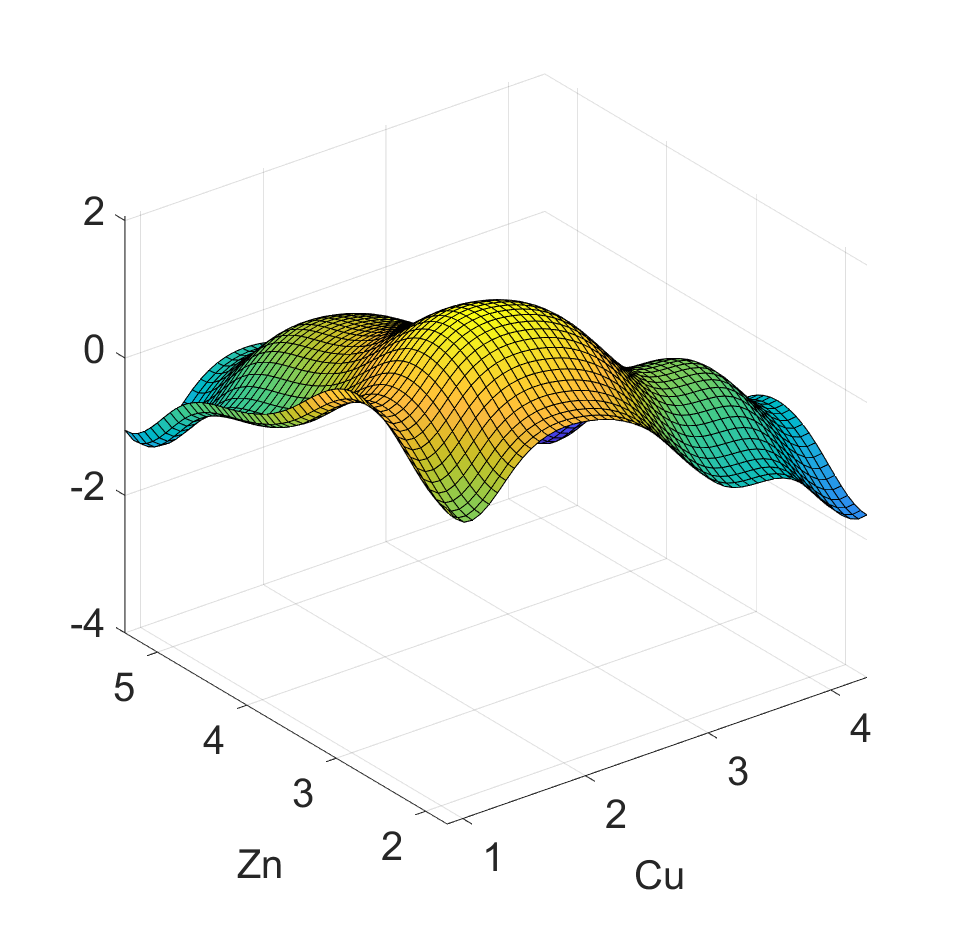}
    \includegraphics[width=0.32\textwidth]{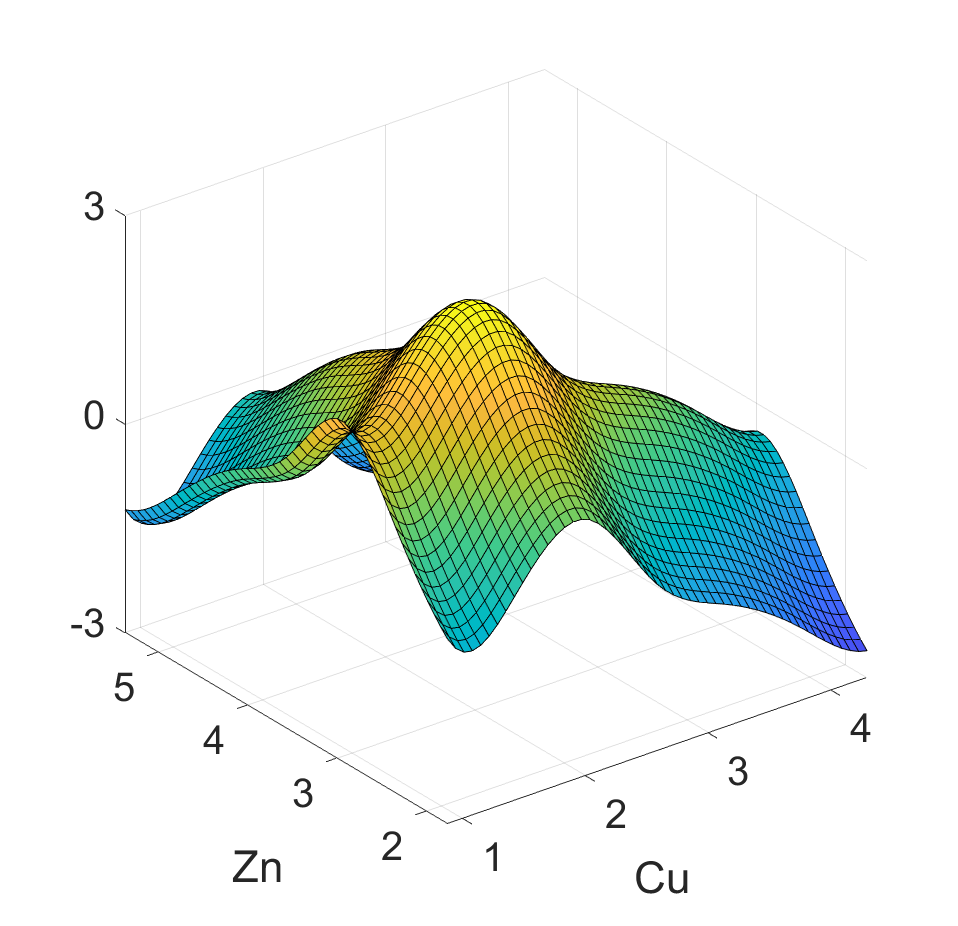}
    \includegraphics[width=0.32\textwidth]{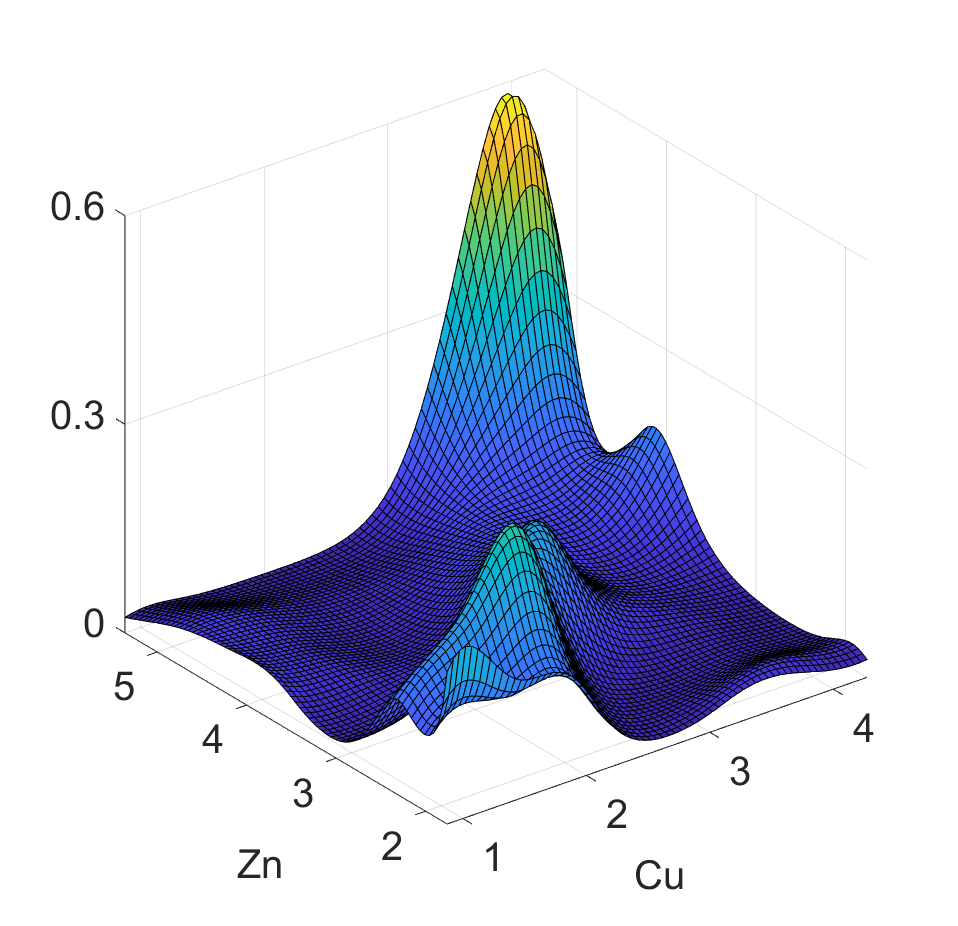}
    \includegraphics[width=0.32\textwidth]{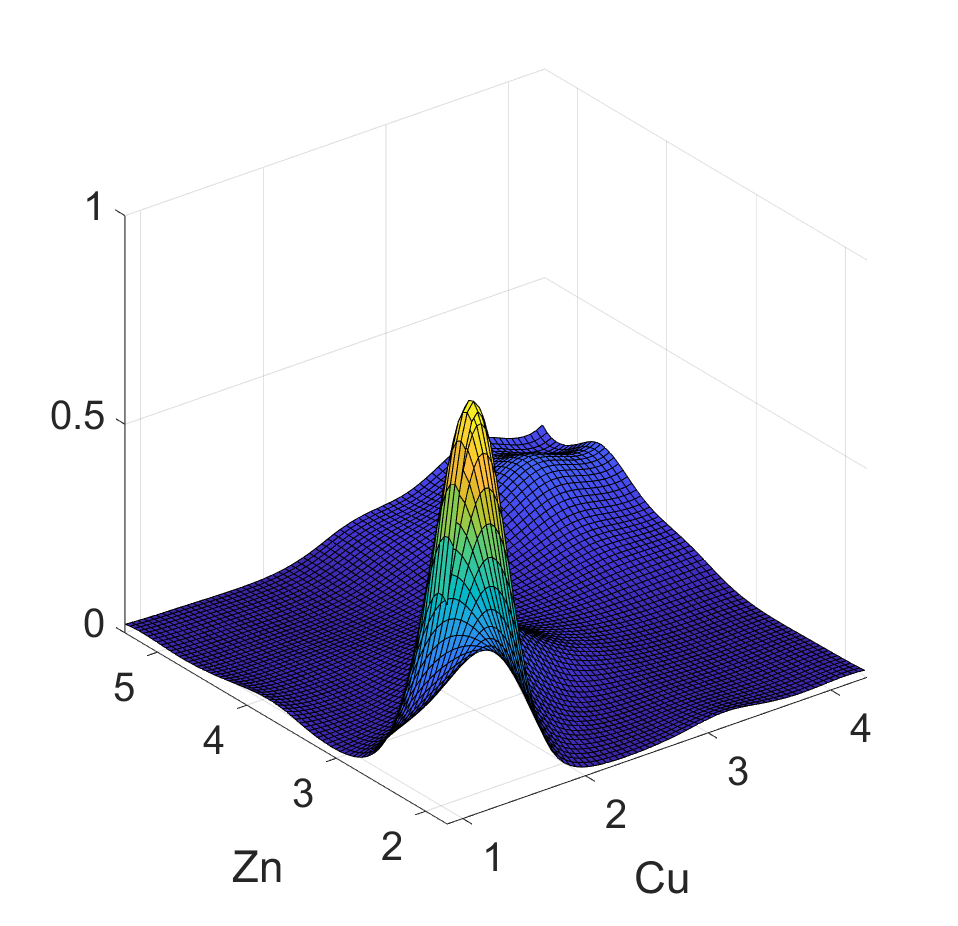}
    \includegraphics[width=0.32\textwidth]{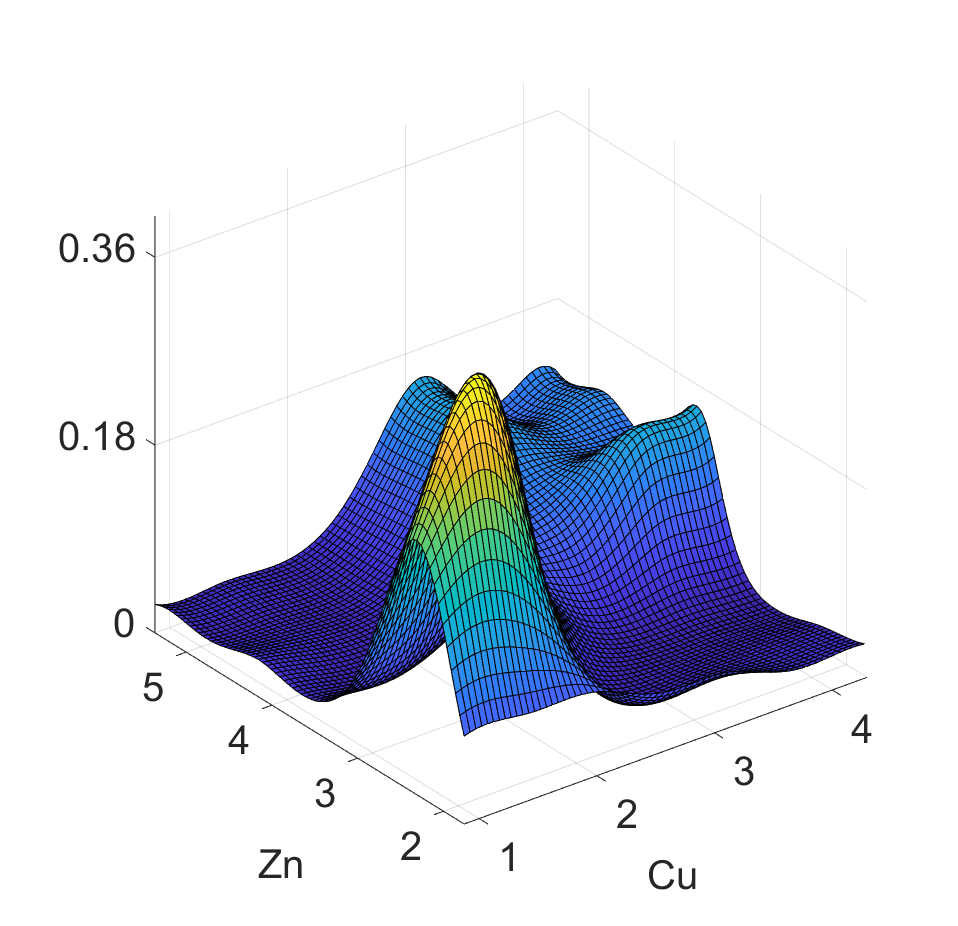}
    \includegraphics[width=0.32\textwidth]{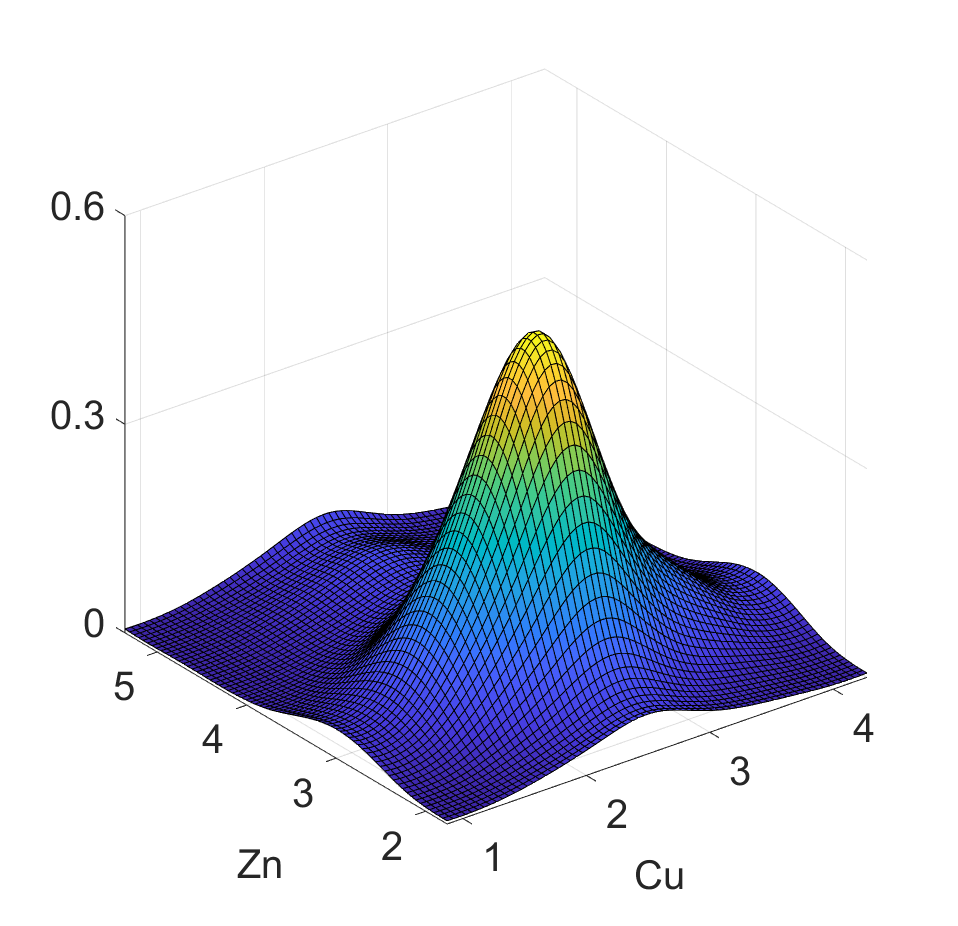}
    \includegraphics[width=0.32\textwidth]{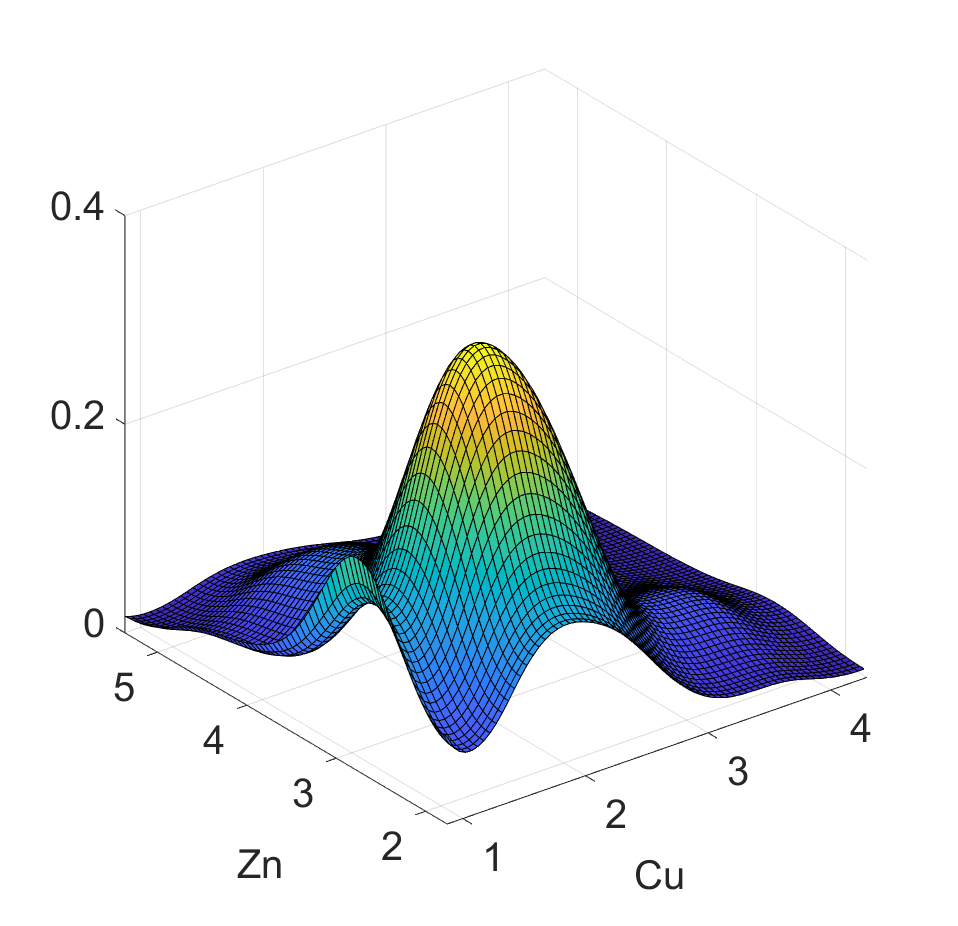}
    \includegraphics[width=0.32\textwidth]{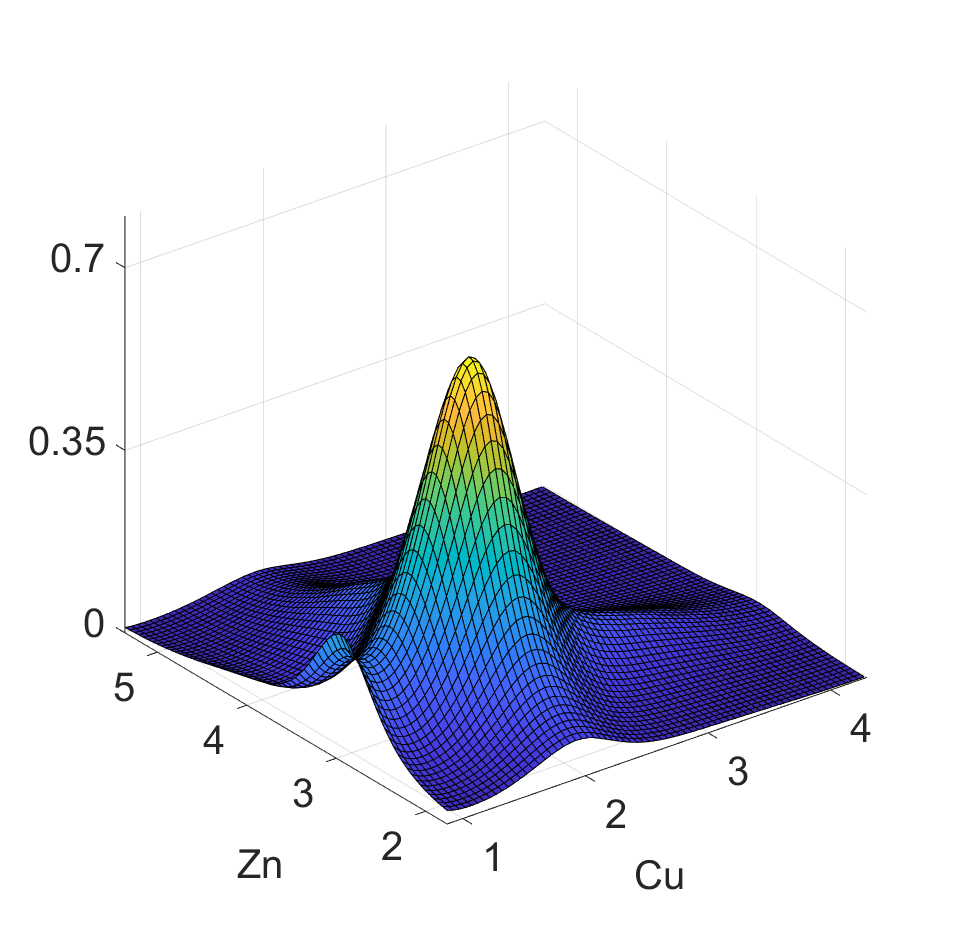}
    \caption{Decomposition of the $Z\!B$-spline representations in $L^2_0(\Omega)$ into their interactive parts (top) and independent parts (second row) and decomposition of the $C\!B$-splines into their interactive parts (third row) and independent part (bottom) for districts Zl\'in (left), Liberec (middle), \v{C}esk\'e Bud\v{e}jovice (right)}
    \label{fig:application_decomposition}
\end{figure}

\begin{figure}[h!]
  \centering
\includegraphics[width=0.35\textwidth]{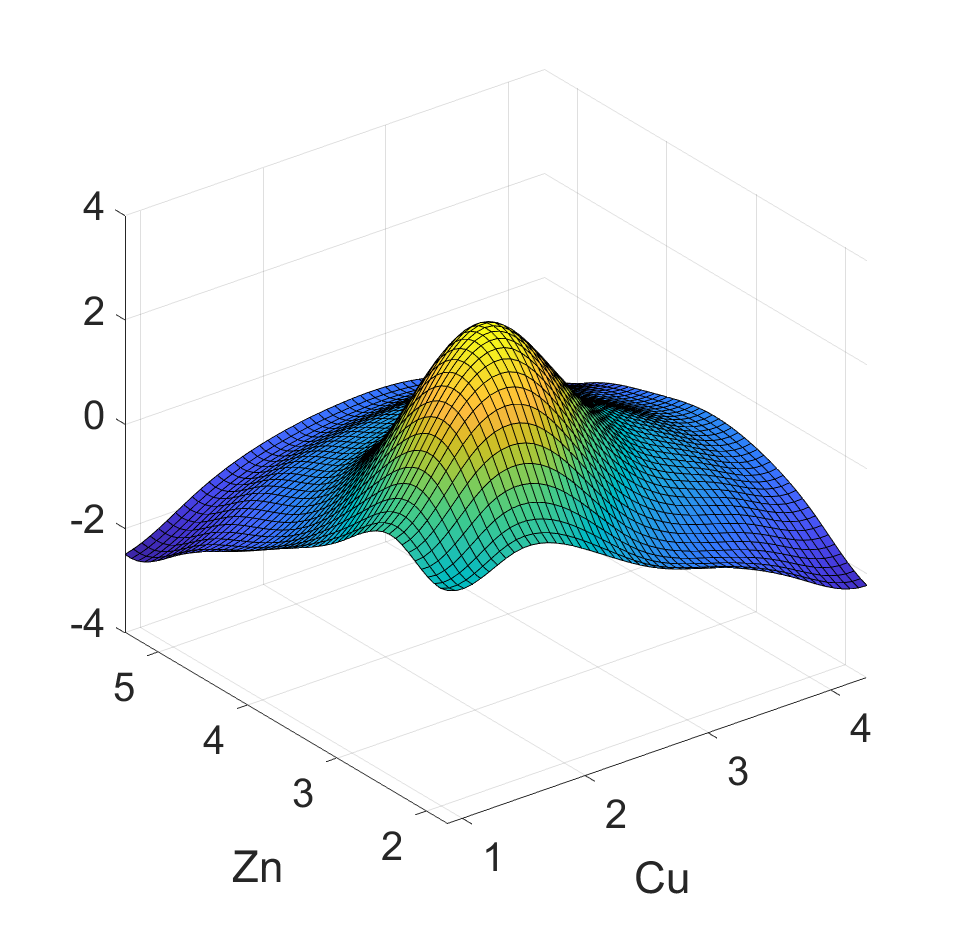}
\includegraphics[width=0.35\textwidth]{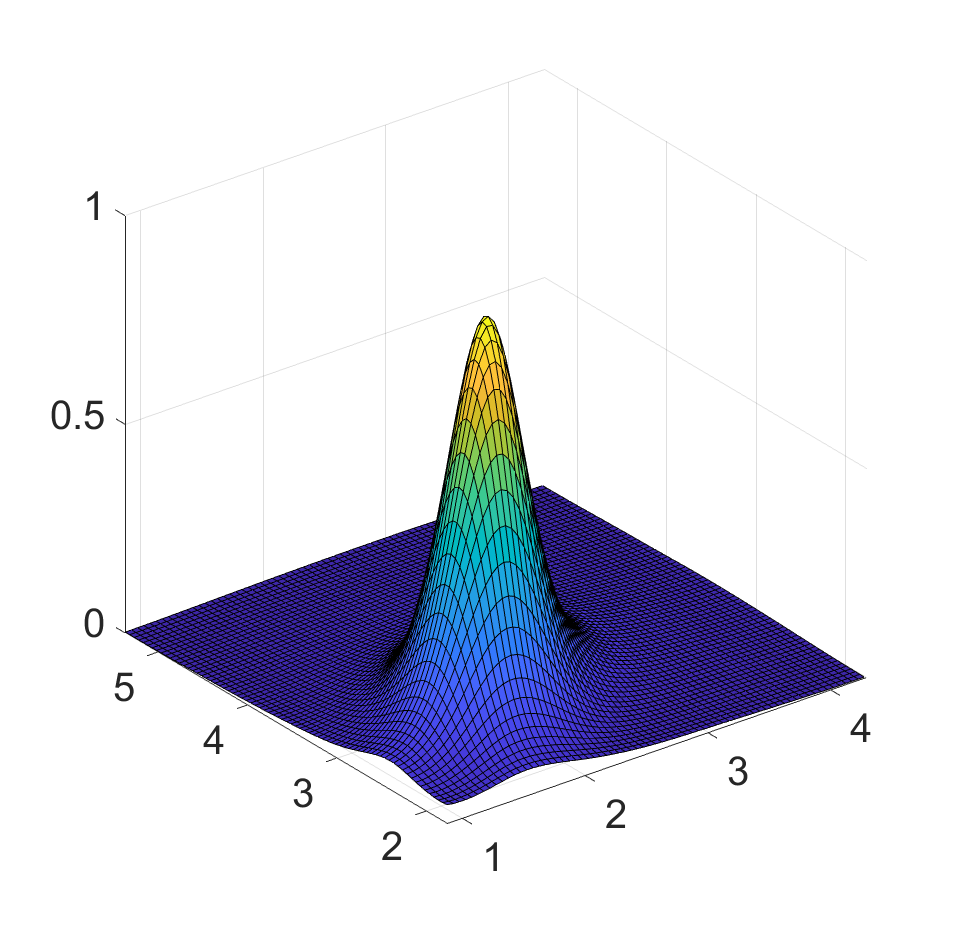}
  \caption{Mean functions calculated from all districts in $L^2_0(\Omega)$ (left) and in ${\cal B}^2 (\Omega)$ (right)}
  \label{fig:application_means}
\end{figure}

The standard deviation function, which can be calculated directly on spline coefficients as proposed in \citep{kokoszka17}, is displayed in Figure \ref{fig:application_var}. In addition to large variation in areas of small contamination, where most districts have their mode, but with variation between districts,  some patterns related  to higher concentrations of both elements can be observed. This is likely due to anthropogenic contamination, which is present in some of the districts \citep{grygar23}. 

\begin{figure}[h!]
  \centering
\includegraphics[width=0.35\textwidth]{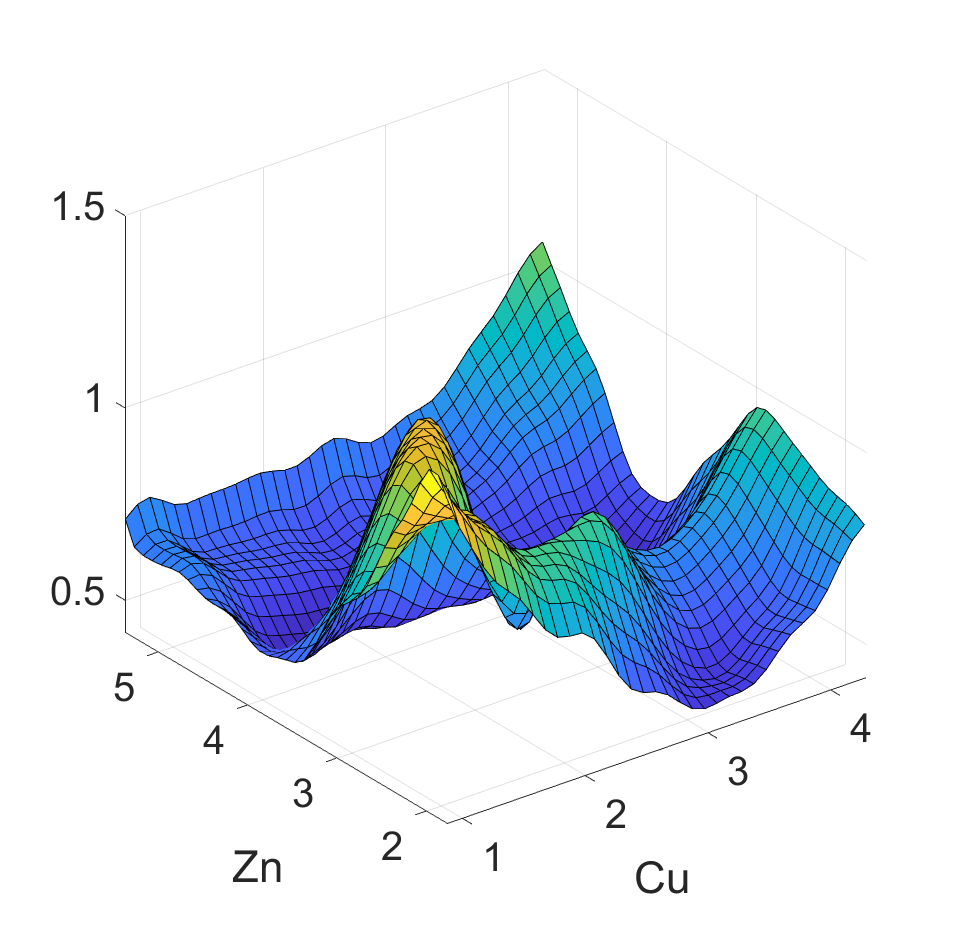}
\includegraphics[width=0.35\textwidth]{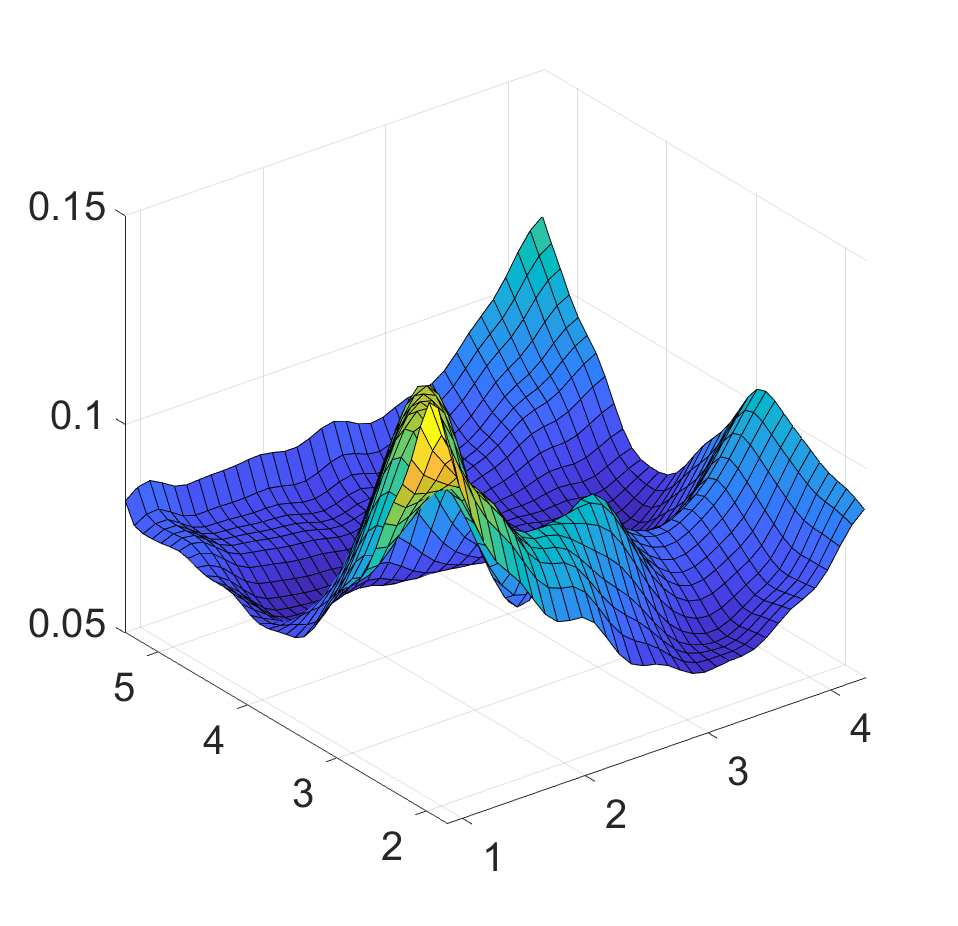}
  \caption{Standard deviation functions calculated from all regions in $L^2_0(\Omega)$ (left) and in ${\cal B}^2 (\Omega)$ (right)}
  \label{fig:application_var}
\end{figure}

Finally, spline coefficients for the estimated $Z\!B$-spline representation of all processed districts are contained in the \ref{secA2}, Table \ref{table}.

Note that approximation with compositional splines represents a parametric approach. As an alternative, kernel smoothing and other nonparametric methods can also be employed in the bivariate case, however, certain limitations start to appear when moving to higher dimensions, i.e. time consuming computations. Our approach, on the other hand, uses the advantage of a basis expansion, which works efficiently in higher dimensions, since the resulting estimate is described by the spline coefficients. This is highly beneficial for statistical processing of PDFs using popular methods of  FDA. Moreover, nonparametric methods cannot capture the decomposition of the resulting density into its interactive and independent part unlike the presented approach. Therefore, compositional splines could serve as a potential cornerstone for approximation and decomposition of PDFs also in higher dimensions and, by providing an easy to handle approximation step, enable to make use of the benefits of the Bayes spaces methodology for statistical processing of PDFs.

\newpage
\section{Conclusion}\label{concl}
The proposed bivariate compositional splines enable a representation of probability density functions in the Bayes space  ${\cal B}^2$, the space respecting their intrinsic properties. The resulting spline approximation can also be expressed in the Lebesgue space $L^2$ using the centred log ratio transformation, which induces a zero integral constraint and forms an isomorphism between the  ${\cal B}^2$ and $L_0^2$  spaces. The clr transformation allows further processing of PDFs with popular FDA methods designed for the  $L^2$ space. In order to estimate the transformed PDFs in $L_0^2$, a suitable basis of bivariate $Z\!B$-splines  was developed. The new $Z\!B$-splines  have the advantage of fulfilling the necessary zero integral condition for clr transformed PDFs naturally. In this way, the proposed bivariate $Z\!B$-splines  represent a natural extension of univariate $Z\!B$-splines introduced in \citep{compositional}.
However, we showed that the basis construction requires more than the tensor product of univariate $Z\!B$-splines, as in the case of bivariate $B$-splines bases in the $L^2$ space, but we needed to design additional basis functions to complete the basis. This concept lays the foundations for further generalizations to the multivariate case.

Moreover, the Bayes space methodology provides a theoretical framework for orthogonal decomposition of densities into their independent and interactive parts, as introduced in \citep{hron22} and extended in \citep{genest22}. A second advantage of the new bivariate compositional splines is that they allow a corresponding decomposition of the basis and the resulting estimated densities. This might be a potential keystone for studying the dependency structure of two random variables and open new horizons for the decomposition of multivariate dependency structures in general. 

The simulation study and the application to empirical large-scale geochemical data clearly demonstrate the practical potential of the spline representation. In the application, the resulting spline approximations for  bivariate density observations from all districts pave the way to numerically stable and effective functional data analysis based on popular spline-based approaches, as we briefly illustrated using descriptive functional data analysis at the end of Section~\ref{appl}.

Overall, the presented approach introduces the foundations for spline approximations of bivariate densities by respecting their intrinsic properties, and allows further analysis of PDFs in the context of functional data analysis. The extension of the introduced methodology to the general multivariate case is highly desirable and is currently under development.
\\\\
\textbf{Acknowledgments.} We acknowledge the support of this research and researchers by the following grants: \\ IGA\_PrF\_2024\_006 Mathematical models, the Czech Science Foundation grant 22-15684L and the German Research Foundation (DFG) grant GR 3793/8-1.
\\\\
\noindent \textbf{Author contribution statement:} J.M. conceived the main conceptual idea. S.\v{S}., J.M. and J.B. developed the theory of bivariate compositional splines. K.H. and S.G. supervised the statistical part of the paper. K.H. designed the simulation study. S.\v{S}. performed the simulation study and computations in the application part. S.\v{S}. prepared the initial manuscript draft. All authors have discussed the results and contributed to the final version of the manuscript.

\appendix
\section{Additional properties of the $Z\!B$-spline representation}\label{secA1}
The independent and interactive parts of a spline $s_{kl}(x,y)\in{\cal Z}_{kl}^{\Delta\lambda,\Delta\mu}(\Omega)$ form its orthogonal decomposition as formulated in Section \ref{Bayes}. Alternatively, this can  also be proven directly in $L^{2}_{0}(\Omega)$ as stated in the following Theorem.
\\
\begin{theorem} \label{biv_orthogonality}
The interactive part $s_{kl}^{int}(x,y)$ and the independent part $s_{kl}^{ind}(x,y)$ of a spline $s_{kl}(x,y)\in{\cal Z}_{kl}^{\Delta\lambda,\Delta\mu}(\Omega)$ are orthogonal.
\end{theorem}
\begin{proof}
The goal is to show that the inner product of the interactive and the independent part is equal to zero. Thus
\begin{align*}
&\left(s_{kl}^{ind}(x,y), s_{kl}^{int}(x,y) \right)_{L^2} = \iint\limits_{\Omega} s_{kl}^{ind}(x,y)s_{kl}^{int}(x,y)\ \mbox{d}x\mbox{d}y \\
&= \int\limits_{a}^{b}\int\limits_{c}^{d} \left[\mathbf{Z}_{k+1}^{\top}(x)\,\mathbf{v} + \mathbf{u}^{\top}\mathbf{Z}_{l+1}(y)\right]\left[\mathbf{Z}_{k+1}^{\top}(x)\,\mathbf{Z}\,\mathbf{Z}_{l+1}(y)\right] \mbox{d}x\mbox{d}y \\
&= \int\limits_{a}^{b}\int\limits_{c}^{d} \left[\mathbf{Z}_{k+1}^{\top}(x)\,\mathbf{v}\,\mathbf{Z}_{k+1}^{\top}(x)\,\mathbf{Z}\,\mathbf{Z}_{l+1}(y) + \mathbf{u}^{\top}\mathbf{Z}_{l+1}(y)\,\mathbf{Z}_{k+1}^{\top}(x)\,\mathbf{Z}\,\mathbf{Z}_{l+1}(y)\right]\mbox{d}x\mbox{d}y \\
&= \int\limits_{a}^{b}\mathbf{Z}_{k+1}^{\top}(x)\,\mathbf{v}\,\mathbf{Z}_{k+1}^{\top}(x)\,\mathbf{Z}\left(\int\limits_{c}^{d}\mathbf{Z}_{l+1}(y)\ \mbox{d}y\right)\,\mbox{d}x + \int\limits_{c}^{d}\mathbf{u}^{\top}\mathbf{Z}_{l+1}(y)\left(\int\limits_{a}^{b}\mathbf{Z}_{k+1}^{\top}(x)\,\mbox{d}x\right)\,\mathbf{Z}\,\mathbf{Z}_{l+1}(y)\,\mbox{d}x = 0,
\end{align*}
since
\begin{equation*}
\int\limits_{a}^{b}\mathbf{Z}_{k+1}^{\top}(x)\,\mbox{d}x = \mathbf{0}\quad\quad\textnormal{and}\quad\quad\int\limits_{c}^{d}\mathbf{Z}_{l+1}(y)\,\mbox{d}y = \mathbf{0}.
\end{equation*}
\end{proof}

As a consequence of Theorems \ref{biv_mn}, \ref{decomposition} and \ref{biv_orthogonality}, the $L^{2}$-norm of the independent and interactive parts of spline $s_{kl}(x,y)$ can be calculated effectively as stated in the following theorem.
\\
\begin{theorem}\label{zb_norm}
The $L^{2}$-norm of the interactive part $s_{kl}^{int}(x,y)$ and the independent part $s_{kl}^{ind}(x,y)$ of spline $s_{kl}(x,y)\in{\cal Z}_{kl}^{\Delta\lambda,\Delta\mu}(\Omega)$ can be obtained as
\begin{align*}
\Vert s_{kl}^{int}(x,y)\Vert_{L^2} = \left(cs(\mathbf{Z})^{\top}\mathbb{N}_{kl}\,cs(\mathbf{Z})\right)^{1/2}
\end{align*}
and
\begin{align*}
\Vert s_{kl}^{ind}(x,y)\Vert_{L^2} = \left(cs(\mathbf{Q})^{\top}\,\mathbb{M}_{kl}\,cs(\mathbf{Q})\right)^{1/2},
\end{align*}
where $\mathbf{Q} = \left(\mathbf{V}\,\mathbf{K}_{g k}\,\mathbf{D}_{\lambda}\right)^{\top} + \mathbf{U}\,\mathbf{K}_{hl}\,\mathbf{D}_{\mu}$ and $cs(\mathbf{Q})$ is its vectorized form, $\mathbb{M}_{kl}=\mathbf{M}_{l}^{y}\otimes\mathbf{M}_{k}^{x}$, where
\begin{align*}    
\mathbf{M}_{l}^{y} = \int_{c}^{d}\mathbf{B}_{l+1}^{\top}(y)\,\mathbf{B}_{l+1}(y)\,\textnormal{d}y,\quad
\mathbf{M}_{k}^{x} = \int_{a}^{b}\mathbf{B}_{k+1}^{\top}(x)\,\mathbf{B}_{k+1}(x)\,\textnormal{d}x
\end{align*}
and $\mathbb{N}_{kl} = \mathbb{K}\,\mathbb{D}\,\mathbb{M}_{kl}\,\mathbb{D}^{\top}\,\mathbb{K}^{\top}$, where $\mathbb{D}=\mathbf{D}_{\mu}\otimes\mathbf{D}_{\lambda}$, $\mathbb{K}=\mathbf{K}_{h l}\otimes\mathbf{K}_{g k}$.
\end{theorem}
\begin{proof}
According to the definition of the $L^2$-norm, the properties of the tensor product and relations \eqref{ZBX_to_B}, \eqref{ZBY_to_B}, we take the square of the norm to obtain
\begin{align*}
&\lVert s_{kl}^{int}(x,y) \rVert^{2}_{L^2} = \iint\limits_{\Omega} \left[ s_{kl}^{int}(x,y)\right]^{2} \mbox{d}x\mbox{d}y = \iint\limits_{\Omega}\left[\mathbf{Z}_{k+1}^{\top}(x)\,\mathbf{Z}\,\mathbf{Z}_{l+1}(y)\right]^{2} \mbox{d}x\,\mbox{d}y \\
&= \iint\limits_{\Omega}\left[\left(\mathbf{Z}_{l+1}^{\top}(y)\,\otimes\,\mathbf{Z}_{k+1}^{\top}(x)\right)\,cs(\mathbf{Z})\right]^{2}\mbox{d}x\,\mbox{d}y  \\
&= \iint\limits_{\Omega} \left[\left(\mathbf{B}_{l+1}(y)\otimes\mathbf{B}_{k+1}(x)\right)^{\top}\,\left(\mathbf{D}_{\mu}\otimes\mathbf{D}_{\lambda}\right)^{\top}\,\left(\mathbf{K}_{h l}\otimes\mathbf{K}_{g k}\right)^{\top}\,cs(\mathbf{Z})\right]^{2} \mbox{d}x\,\mbox{d}y\\
&= cs(\mathbf{Z})^{\top}\,\mathbb{K}\,\mathbb{D}\,\iint\limits_{\Omega} \left[\mathbf{B}_{l+1}(y)\otimes\mathbf{B}_{k+1}(x)\right]^{\top}\left[\mathbf{B}_{l+1}(y)\otimes\mathbf{B}_{k+1}(x)\right]\,\mbox{d}x\,\mbox{d}y\,\mathbb{D}^{\top}\,\mathbb{K}^{\top}\,cs(\mathbf{Z}) \\
&=cs(\mathbf{Z})^{\top}\,\mathbb{K}\,\mathbb{D}\,\left(\mathbf{M}_{l}^{y}\,\otimes\,\mathbf{M}_{k}^{x}\right)\,\mathbb{D}^{\top}\,\mathbb{K}^{\top}\,cs(\mathbf{Z}) = cs(\mathbf{Z})^{\top}\mathbb{K}\,\mathbb{D}\,\mathbb{M}_{kl}\,\mathbb{D}^{\top}\,\mathbb{K}^{\top}\,cs(\mathbf{Z})  \\
&= cs(\mathbf{Z})^{\top}\,\mathbb{N}_{kl}\,cs(\mathbf{Z}).
\end{align*}
Analogously for the independent part
\begin{align*}
&\Vert s_{kl}^{ind}(x,y) \Vert^{2}_{L^2} = \iint\limits_{\Omega}\left[s_{kl}^{ind}(x,y)\right]^{2}\mbox{d}x\,\mbox{d}y = \iint\limits_{\Omega}\left[\mathbf{Z}_{k+1}^{\top}(x)\,\mathbf{v} + \mathbf{u}^{\top}\mathbf{Z}_{l+1}(y)\right]^{2}\,\mbox{d}x\,\mbox{d}y \\
&=\iint\limits_{\Omega}\left[\mathbf{B}_{k+1}^{\top}(x)\,\left(\left(\mathbf{V}\,\mathbf{K}_{g k}\,\mathbf{D}_{\lambda}\right)^{\top} + \mathbf{U}\,\mathbf{K}_{hl}\,\mathbf{D}_{\mu}\right)\,\mathbf{B}_{l+1}(y)\right]^{2}\,\mbox{d}x\,\mbox{d}y \\
&= \iint\limits_{\Omega}\left[\left(\mathbf{B}_{l+1}(y)\,\otimes\,\mathbf{B}_{k+1}(x)\right)^{\top}\,cs(\mathbf{Q})\right]^{2}\,\mbox{d}x\, \mbox{d}y \\
&= cs(\mathbf{Q})^{\top}\left(\mathbf{M}_{l}^{y}\,\otimes\,\mathbf{M}_{k}^{x}\right)\,cs(\mathbf{Q}) = cs(\mathbf{Q})\,\mathbb{M}_{kl}\,cs(\mathbf{Q}).
\end{align*}
Taking square roots of the obtained formulas completes the proof.
\end{proof}
\begin{corollary}
Since the independent part and the interactive part are orthogonal, the square of the $L^{2}$-norm of the spline $s_{kl}(x,y)$ can be calculated by the formula
\begin{align*}
\Vert s_{kl}(x,y) \Vert^{2}_{L^2} = \Vert s_{kl}^{int}(x,y) \Vert^{2}_{L^2} + \Vert s_{kl}^{ind}(x,y) \Vert^{2}_{L^2} = cs(\mathbf{Z})^{\top}\,\mathbb{N}_{kl}\,cs(\mathbf{Z}) + cs(\mathbf{Q})\,\mathbb{M}_{kl}\,cs(\mathbf{Q}).
\end{align*}
\end{corollary}

\section{Spline coefficients of processed districts} \label{secA2}
The following table states $Z\!B$-spline coefficients and corresponding $B$-spline coefficients for processed districts Zl\'in, Liberec and \v{C}esk\'e Bud\v{e}jovice used in Section \ref{appl}.

\begin{table}
\centering
\caption{Table of $Z\!B$-spline and $B$-spline coefficients for regions Zl\'in, Liberec and \v{C}esk\'e Bud\v{e}jovice}
\label{table}
\footnotesize
\begin{tabular}{r r r r r r r r}
    \hline
    \hline
    \multicolumn{8}{l}{$Z\!B$-spline coefficients $\mathbf{R} = \{ r_{ij}\}_{\,i,\,j}^{\,g+k+1,\,h+l+1}$ for the Zl\'in district} \\
    \hline
    0.0433 & 0.0818 & 0.1694 & 0.0345 & 0.0713 & 0.0850 & 0.0302 & -0.2207 \\
    0.0491 & 0.0420 & 0.2877 & -0.1483 & 0.0352 & 0.1858 & 0.0696 & -0.7231 \\
    0.3329 & 0.7511 & 1.9396 & 0.5676 & 0.3729 & 0.5489 & 0.1974 & -0.7382 \\
    0.2733 & 0.6452 & 2.8907 & 2.3161 & 1.4441 & 1.2999 & 0.4630 & 0.3199 \\
    0.0364 & -0.0951 & 1.0259 & 2.0287 & 1.7751 & 1.4275 & 0.5143 & 0.9355 \\
    0.0924 & 0.0973 & 0.5932 & 0.9820 & 0.5913 & 0.5155 & 0.1845 & 0.5201 \\
    0.0210 & 0.0010 & 0.1433 & 0.2894 & 0.1398 & 0.1390 & 0.0503 & 0.1152 \\
   -0.1512 & -0.5758 & -0.0754 & 1.0684 & 0.7085 & 0.3582 & 0.0990 & 0 \\
   \hline
   \hline
    \multicolumn{8}{l}{$B$-spline coefficients $\mathbf{B} = \{ b_{ij}\}_{\,i,\,j}^{\,g+k+1,\,h+l+1}$ for the Zl\'in district} \\
    \hline
   -0.7148 & -1.8594 & 0.6025 & -0.8219 & -1.5061 & -1.8114 & -2.9477 & -2.8679 \\
   -2.2410 & -3.0606 & 0.3277 & -1.1327 & -1.3827 & -1.3842 & -2.7270 & -2.7004 \\
    2.2693 & 1.1367 & 4.3953 & -1.0186 & -1.5851 & -0.5921 & -2.0637 & -2.0048 \\
    0.2200 & 0.1799 & 5.4364 & 4.8330 & -0.3488 & 0.0216 & -1.1824 & -1.2009 \\
   -1.9102 & -2.3770 & -1.2707 & 5.8008 & 1.6921 & -0.3096 & -0.1350 & -0.0726 \\
   -1.0482 & -1.2567 & -2.1944 & -0.9155 & -1.7062 & -0.4711 & 1.6821 & 2.2828 \\
   -3.2379 & -2.5991 & -2.2223 & -0.6011 & -0.7012 & -1.4391 & 0.0889 & 0.4786 \\
   -2.2306 & -1.5426 & -1.3241 & -0.2901 & 0.0570 & -1.3289 & 0.0655 & 0.4360 \\
   \hline
   \hline
    \multicolumn{8}{l}{$Z\!B$-spline coefficients $\mathbf{R} = \{ r_{ij}\}_{\,i,\,j}^{\,g+k+1,\,h+l+1}$ for the Liberec district} \\
    \hline
    0.0822 & 0.3022 &  0.3638 &  0.2847 &  0.2263 &  0.0799 &  0.0325 & -0.0313 \\
    0.2844 & 0.9746 &  1.2803 &  0.8982 &  0.6173 &  0.1903 &  0.0702 & -0.1824 \\
    0.4076 & 1.4870 &  2.4502 &  1.8265 &  1.5480 &  0.7192 &  0.2612 &  0.5677 \\
    0.2235 & 1.0008 &  2.3545 &  2.6235 &  1.9302 &  0.8914 &  0.2995 &  1.2520 \\
    0.1473 & 0.5951 &  1.6356 &  2.1922 &  1.4420 &  0.7238 &  0.2401 &  1.1634 \\
    0.0755 & 0.2796 &  0.7868 &  0.9706 &  0.5916 &  0.2972 &  0.0938 &  0.8248 \\
    0.0275 & 0.0968 &  0.2759 &  0.3408 &  0.2138 &  0.1123 &  0.0357 &  0.2831 \\
    0.0966 & 0.0802 &  0.9229 &  1.3868 &  1.2724 &  0.9913 &  0.3157 & 0 \\
   \hline
   \hline
    \multicolumn{8}{l}{$B$-spline coefficients $\mathbf{B} = \{ b_{ij}\}_{\,i,\,j}^{\,g+k+1,\,h+l+1}$ for the Liberec district} \\
    \hline
    3.0944 &  3.4294 &  2.0780 & -0.1924 & -0.8317 & -2.3376 & -2.8736 & -3.0419 \\
    3.4611 &  3.4208 &  2.4859 & -1.0557 & -1.5342 & -2.5306 & -2.9516 & -2.8495 \\
    3.3903 &  3.5940 &  5.4920 &  1.4622 &  1.3264 & -0.5323 & -2.2937 & -2.4143 \\
    0.0240 & -0.2903 &  3.6754 &  3.5223 & -0.0110 & -0.0963 & -1.4440 & -1.0801 \\
   -0.2226 & -1.5479 &  0.5793 &  1.1190 & -0.4102 &  0.2312 & -1.5801 & -1.4115 \\
   -0.9224 & -2.0663 & -1.0615 & -1.0502 &  0.2007 &  0.3728 & -1.0128 & -0.8193 \\
   -1.8597 & -2.7711 & -1.8448 & -1.4468 & -0.7157 & -1.0602 & -2.4475 & -2.4121 \\
   -2.0475 & -2.8773 & -2.0842 & -1.5640 & -0.7802 & -1.0777 & -2.3001 & -2.2629 \\
   \hline
   \hline
    \multicolumn{8}{l}{$Z\!B$-spline coefficients $\mathbf{R} = \{ r_{ij}\}_{\,i,\,j}^{\,g+k+1,\,h+l+1}$ for the \v{C}esk\'e Bud\v{e}jovice district} \\
    \hline
    0.0267 &  0.0604 &  0.2582 &  0.1228 &  0.0721 &  0.0493 &  0.0133 &  0.0817 \\
    0.0799 &  0.2058 &  0.8363 &  0.5654 &  0.3498 &  0.2476 &  0.0780 & -0.0503 \\
    0.0454 &  0.2200 &  1.5853 &  1.3883 &  0.4894 &  0.3522 &  0.0965 &  0.9622 \\
    0.0687 &  0.2754 &  1.9651 &  3.0379 &  1.4591 &  0.9125 &  0.2766 &  1.4043 \\
    0.0616 &  0.1327 &  0.9996 &  1.9312 &  0.7566 &  0.5475 &  0.1520 &  1.2296 \\
    0.0314 &  0.0591 &  0.4995 &  0.9390 &  0.3205 &  0.2703 &  0.0698 &  0.8475 \\
    0.0122 &  0.0221 &  0.1712 &  0.3206 &  0.1089 &  0.0959 &  0.0243 &  0.2871 \\
   -0.0901 & -0.3202 &  0.8131 &  1.5073 &  1.4200 &  0.9185 &  0.3194 & 0 \\
   \hline
   \hline
    \multicolumn{8}{l}{$B$-spline coefficients $\mathbf{B} = \{ b_{ij}\}_{\,i,\,j}^{\,g+k+1,\,h+l+1}$ for the \v{C}esk\'e Bud\v{e}jovice district} \\
    \hline
    0.8664 &  0.3968 &  4.8031 &  0.3323 & -0.0606 & -0.7093 & -1.7995 & -1.7552 \\
   -0.0117 & -0.2693 &  4.1336 &  0.0212 & -1.1995 & -1.7711 & -3.1866 & -3.2630 \\
    1.1094 &  1.6213 &  6.8401 &  3.1792 & -0.0198 &  0.9286 & -0.1635 & -0.0026 \\
    0.3423 &  0.1406 &  3.6783 &  4.2724 & -0.8830 & -1.4211 & -2.6119 & -2.6405 \\
   -0.8244 & -1.4695 & -0.4175 &  0.4232 &  0.4598 & -0.2625 & -0.9426 & -1.0169 \\
   -1.5953 & -1.6423 & -0.2077 & -1.1434 &  0.6658 & -1.1062 & -1.3563 & -1.6383 \\
   -2.4858 & -2.4536 & -1.1516 & -1.8904 & -0.0889 & -2.3923 & -2.2711 & -2.6963 \\
   -2.6117 & -2.5114 & -1.2302 & -1.9674 & -0.0611 & -2.4958 & -2.1909 & -2.6877 \\
   \hline
   \end{tabular}
\end{table}

\newpage
\bibliographystyle{unsrtnat}

\end{document}